\documentclass[onecolumn]{IEEEtran}
\usepackage{paralist}
\newcommand{\pr}[1]{\left({#1}\right)}
\newcommand{\supp}[1]{{\textnormal{support}\pr{#1}}}
\newcommand{\wy}{w_{Y|X}}
\newcommand{\wz}{w_{Z|X}}

\newcommand{\myr}[1]{{#1}}
\newcommand{\modfirst}[1]{{#1}}

\usepackage{amsmath,stmaryrd,acronym,amssymb,amsthm,dsfont,graphicx}
\graphicspath{{graphics/}}
\acrodef{OOK}[OOK]{On-Off Keying}

\acrodef{AEP}{Asymptotic Equipartition Property}
\acrodef{AoA}{Angle of Arrival}
\acrodef{AWGN}{Additive White Gaussian Noise}
\acrodef{BER}{Bit-Error-Rate}
\acrodef{BEC}{Binary Erasure Channel}
\acrodef{BPSK}{Binary Phase-Shift Keying}
\acrodef{BSC}{Binary Symmetric Channel}
\acrodef{CDF}[CDF]{Cumulative Distribution Function}
\acrodef{CLT}[CLT]{Central Limit Theorem}
\acrodef{CSI}[CSI]{Channel State Information}
\acrodef{DMC}[DMC]{Discrete Memoryless Channel}
\acrodef{DMS}[DMS]{Discrete Memoryless Source}
\acrodef{iid}[i.i.d.]{independent and identically distributed}
\acrodef{lhs}[l.h.s.]{left-hand-side}
\acrodef{rhs}[r.h.s.]{right-hand-side}
\acrodef{LPD}[LPD]{Low Probability of Detection}
\acrodef{LDPC}[LDPC]{Low-Density Parity-Check}
\acrodef{MAC}[MAC]{multiple-access channel}
\acrodef{MIMO}[MIMO]{Multiple-Input Multiple-Output}
\acrodef{MISO}{Multiple-Input Single-Output}
\acrodef{PDF}[PDF]{Probability Density Function}
\acrodef{PMF}[PMF]{Probability Mass Function}
\acrodef{PPM}[PPM]{Pulse Position Modulation}
\acrodef{PSD}{Power Spectral Density}
\acrodef{QPSK}{Quadrature Phase-Shift Keying}
\acrodef{SIMO}{Single-Input Multiple-Output}
\acrodef{SNR}{Signal-to-Noise Ratio}
\acrodef{wrt}[w.r.t.]{with respect to}
\acrodef{WSS}{Wide Sense Stationary}

\DeclareMathAlphabet{\eurm}{U}{eur}{m}{n}
\DeclareMathAlphabet{\mathbsf}{OT1}{cmss}{bx}{n}
\DeclareMathAlphabet{\mathssf}{OT1}{cmss}{m}{sl}
\DeclareMathAlphabet{\mathcsf}{OT1}{cmss}{sbc}{n}

\DeclareSymbolFont{bsfletters}{OT1}{cmss}{bx}{n}  
\DeclareSymbolFont{ssfletters}{OT1}{cmss}{m}{n}
\DeclareMathSymbol{\bsfGamma}{0}{bsfletters}{'000}
\DeclareMathSymbol{\ssfGamma}{0}{ssfletters}{'000}
\DeclareMathSymbol{\bsfDelta}{0}{bsfletters}{'001}
\DeclareMathSymbol{\ssfDelta}{0}{ssfletters}{'001}
\DeclareMathSymbol{\bsfTheta}{0}{bsfletters}{'002}
\DeclareMathSymbol{\ssfTheta}{0}{ssfletters}{'002}
\DeclareMathSymbol{\bsfLambda}{0}{bsfletters}{'003}
\DeclareMathSymbol{\ssfLambda}{0}{ssfletters}{'003}
\DeclareMathSymbol{\bsfXi}{0}{bsfletters}{'004}
\DeclareMathSymbol{\ssfXi}{0}{ssfletters}{'004}
\DeclareMathSymbol{\bsfPi}{0}{bsfletters}{'005}
\DeclareMathSymbol{\ssfPi}{0}{ssfletters}{'005}
\DeclareMathSymbol{\bsfSigma}{0}{bsfletters}{'006}
\DeclareMathSymbol{\ssfSigma}{0}{ssfletters}{'006}
\DeclareMathSymbol{\bsfUpsilon}{0}{bsfletters}{'007}
\DeclareMathSymbol{\ssfUpsilon}{0}{ssfletters}{'007}
\DeclareMathSymbol{\bsfPhi}{0}{bsfletters}{'010}
\DeclareMathSymbol{\ssfPhi}{0}{ssfletters}{'010}
\DeclareMathSymbol{\bsfPsi}{0}{bsfletters}{'011}
\DeclareMathSymbol{\ssfPsi}{0}{ssfletters}{'011}
\DeclareMathSymbol{\bsfOmega}{0}{bsfletters}{'012}
\DeclareMathSymbol{\ssfOmega}{0}{ssfletters}{'012}

\newcommand{\calC}{{\mathcal{C}}}
\newcommand{\calD}{{\mathcal{D}}}
\newcommand{\calE}{{\mathcal{E}}}
\newcommand{\calF}{{\mathcal{F}}}

\newcommand{\calN}{{\mathcal{N}}}
\newcommand{\calO}{{\mathcal{O}}}
\newcommand{\calP}{{\mathcal{P}}}

\newcommand{\calR}{{\mathcal{R}}}

\newcommand{\calS}{{\mathcal{S}}}
\newcommand{\calU}{{\mathcal{U}}}
\newcommand{\calV}{{\mathcal{V}}}
\newcommand{\calX}{{\mathcal{X}}}
\newcommand{\calY}{{\mathcal{Y}}}

\newcommand{\calZ}{{\mathcal{Z}}}

\newcommand{\E}[2][]{{\mathbb{E}_{#1}}{\left(#2\right)}}       
\renewcommand{\P}[2][]{{\mathbb{P}_{#1}}{\left(#2\right)}}
       
\newcommand{\D}[2]{{{\mathbb{D}}\!\left({#1\Vert#2}\right)}}
\newcommand{\avgD}[2]{{{\mathbb{D}}\!\left({#1\Vert#2}\right)}}

\newcommand{\avgI}[1]{{{\mathbb{I}}\!\left(#1\right)}}

\newcommand{\Hb}[1]{{\mathbb{H}_b}\left(#1\right)}

\newcommand{\abs}[1]{\ensuremath{\left|#1\right|}}              
\newcommand{\eqdef}{\ensuremath{\triangleq}}                    
\newcommand{\intseq}[2]{\ensuremath{\llbracket{#1},{#2}\rrbracket}}  
\newcommand{\indic}[1]{\ensuremath{\mathds{1}\!\left\{#1\right\}}}

\renewcommand{\leq}{\leqslant}
\renewcommand{\geq}{\geqslant}

\newcommand{\proddist}{%
  \mathchoice{\raisebox{1pt}{$\displaystyle\otimes$}}
             {\raisebox{1pt}{$\otimes$}}
             {\raisebox{0.5pt}{\scalebox{0.7}{$\scriptstyle\otimes$}}}
             {\raisebox{0.4pt}{\scalebox{0.6}{$\scriptscriptstyle\otimes$}}}}
\newcommand{\pn}{{\proddist n}}

\usepackage[textwidth=0.68in,textsize=footnotesize,disable]{todonotes}
\presetkeys{todonotes}{color=redArcom, linecolor=redArcom,bordercolor=redArcom}{}

\newtheorem{theorem}{Theorem}[section]

\newtheorem{lemma}{Lemma}[section]

\newtheorem{proposition}{Proposition}[section]

\acrodef{ROC}[ROC]{Receiver Operation Characteristic}
\acrodef{PPM}[PPM]{Pulse-Position Modulation}

\presetkeys{todonotes}{color=redArcom, linecolor=redArcom,bordercolor=redArcom}{}

\acrodef{AVC}{arbitrary varying channel}
\newcommand{\td}[1]{\widetilde{#1}}

\begin{document}
\title{Covert Capacity of Non-Coherent Rayleigh-Fading Channels}
\author{Mehrdad Tahmasbi, Anne Savard, and Matthieu R. Bloch}
\maketitle


\begin{abstract}
  The covert capacity is characterized for a non-coherent fast Rayleigh-fading wireless channel, in which a legitimate user wishes to communicate reliably with a legitimate receiver while escaping detection from a warden. It is shown that the covert capacity is achieved with an amplitude-constrained input distribution that consists of a finite number of mass points including one at zero and numerically tractable bounds are provided. It is also conjectured that distributions with two mass points in fixed locations are optimal.
\end{abstract}

\section{Introduction}
\label{sec:introduction}
In cognitive radio networks or adversarial communication settings, situations arise in which legitimate users may attempt to communicate covertly, in the sense of achieving a low probability of detection. Motivated by such applications,~\cite{Bash2013} proposed an information-theoretic model to study the throughput at which two users could reliably and covertly communicate over an \ac{AWGN} channel in the presence of an adversary who observes the transmission through another noisy channel. The optimal covert communication throughput has been shown to satisfy a \emph{square root law}, by which the maximum number of bits is on the order of $\sqrt{n}$ bits over $n$ uses of the channel. The square root law was subsequently established for some quantum channels~\cite{Bash2015a} and proved to hold without requiring secret keys for binary symmetric channels under some channel conditions~\cite{Che2013}. The exact pre-constant associated to the square root law, which plays the role of a \emph{covert capacity}, has since been nearly completely characterized for point-to-point discrete and \ac{AWGN} classical channels~\cite{Wang2016b,Bloch2015b,Tahmasbi2017}, as well as some classical-quantum channels~\cite{Wang2016c,Sheikholeslami2016}. With the notable exception of~\cite{Tahmasbi2017}, the covert capacity is typically derived when using the relative entropy as a proxy metric for covertness. Recent results~\cite{Yan2018} offer a more nuanced perspective and show that the optimal signaling scheme for covert communication over \ac{AWGN} channels at finite length is metric-dependent; nevertheless, the present work still uses relative entropy to characterize covert capacity because of its convenient mathematical properties.

For~\acp{DMC}, the covert-capacity achieving input distribution takes the form of  \modfirst{sparse signalling} corresponding to those symbols that might arouse suspicion if transmitted, are used a fraction $1/\sqrt{n}$ of the time if $n$ is the block length. Perhaps surprisingly, \modfirst{sparse signalling} does \emph{not} achieve the covert-capacity of \ac{AWGN} channels \cite{Kadampot2019}, as the optimal coding scheme exploits instead Gaussian or \ac{BPSK}~\cite{Wang2016b} signaling with an average power vanishing as $O\left(1/n\right)$. In other words, encoding information in the \emph{phase} of modulation symbols together with a diffuse power is crucial for optimality. Gaussian signaling has therefore been used to further study covertness over Gaussian and wireless channels, as in~\cite{Sobers2017,Sobers2017a} to show the benefits of uninformed jammers, in~\cite{Bash2016} to analyze the role of randomized timing, in~\cite{Yan2019} to study the effect of randomized power allocation, and in~\cite{Hu2018a} to analyze covert relaying strategies. We note that all aforementioned works exploit random Gaussian codebooks, which simplifies the covertness analysis by reducing the optimal attack to a radiometer. In contrast, we analyze covertness with non-random codebooks using the conceptual approach laid out in~\cite{Bloch2015b}.

While Gaussian codebooks provides valuable insight into the properties of coding schemes for covert communications over \ac{AWGN} channels, operating in the vanishing-power regime as suggested by the results might prove challenging. In particular, not only may phase-lock loops fail to properly track the phase of the transmitted signals but symbols with low amplitude may also be severely affected by phase noise, resulting in a significant degradation of the transmission reliability. These effects are also likely to be amplified by the presence of fading in wireless links. The objective of the present paper is to develop insight into this problem by characterizing the covert capacity of non-coherent fast Rayleigh-fading channels (Theorem~\ref{th:general-capacity-non-coherent} in Section~\ref{sec:system-model}), in which the phase is uniformly distributed over $[0;2\pi[$; although no channel state information is available to the transmitter and receivers, some symbol-level synchronization is assumed.

Our analysis of the covert capacity for non-coherent channels builds upon the ideas initially developed in~\cite{Smith1969,Smith1971} for amplitude constrained channels and extended to~\cite{Abou-Faycal2001} for memoryless non-coherent Rayleigh fading channels under an average power constraint. In particular we show that an optimal covert capacity achieving input distribution is discrete, with one mass point located at zero and subject to an amplitude constraint. While the discrete nature of the distribution may not be a surprise, the fact that the location of the mass \modfirst{points} is bounded results from the specific nature of the covertness constraint. We also conjecture that two mass points in \emph{fixed} locations is actually optimal, which is supported by numerical results although we do not have a formal proof. Overall, our results suggest that, in the presence of phase uncertainty, \modfirst{sparse signaling} might be an efficient modulation scheme for covert communication. 

Our proof technique follows for the most part the high-level approach outlined in~\cite{Smith1969,Smith1971,Abou-Faycal2001}; however, the covert communication constraint makes the analysis more intricate as the optimal capacity-achieving input distribution turns out depend on the block length. In particular, the  converse arguments for single-letterization lead to a parameter-dependent constrained optimization problem, in which the parameter should be taken to zero as the blocklength goes to infinity (see the statement of Theorem~\ref{th:general-capacity-non-coherent} and~\eqref{eq:optimation-a-nu} in Section~\ref{sec:converse-proof}). This requires us to analyze the fine dependence of the objective function and the Lagrange multipliers as a function of a parameter using ideas from sensitivity analysis~\cite{boyd2004convex}.

The rest of the paper is organized as follows. In Section~\ref{sec:system-model}, we introduce the precise model for covert communication over non-coherent Rayleigh-fading channels and discuss our characterization of the covert capacity. In Section~\ref{sec:proof-theor-non-coherent}, we develop the proof of our main result, with the achievability proof in Section~\ref{sec:achievability-proof} and the converse proof in Section~\ref{sec:converse-proof}.
\modfirst{
\section{Notation and Conventions}
Let $(\calS, \calF)$ be a measurable space. When $\calS$ is a subset of $\mathbb{R}$, we always consider the $\sigma$-algebra induced by Borel sets, which converts $\calS$ to a measurable space. Let $f:\calS\to \mathbb{R}$ be measurable and $\mu$ be a measure over $\calS\subset \mathbb{R}$. We call $f$ integrable if $\int_\calS |f| d\mu < \infty$. We then denote the Lebesgue's integral by $\int_{\calS} f(x) d\mu$.  If $\calS = ]a, b[$ and $\mu$ is the Lebesgue's measure over $\calS$, we  denote $\int_{\calS} f(x) d\mu = \int_a^b f(x) dx = \int_a^b f $. If $\mu$ is a probability measure, $X:\calS \to \mathbb{R}$ is a random variable, and $A$ is an event, we use $\P[\mu]{A}$ and $\E[\mu]{X}$ to denote $\mu(A)$ and $\int_{\calS} X(s) d\mu$, respectively. When the probability measure $\mu$ is discrete,  it can be characterized with a \ac{PMF} $P:\calS \to [0, 1]$ satisfying $\mu(A) = \sum_{s\in A} P(s)$. When the probability measure $\mu$ is continuous,  it can be characterized with a \ac{PDF} $f:\calS \to [0, \infty[$ satisfying $\mu(A) = \int_{A} f(s) ds$. We do not  distinguish between a probability measure and its \ac{PMF} or \ac{PDF} (if they exist).  The product of two measures $\mu$ and $\mu'$ is defined in the standard way and is denoted by $\mu \otimes \mu'$. We define the relative entropy between  two probability measures $\mu$ and $\mu'$ as $\D{\mu}{\mu'} \eqdef \E[\mu]{\log \frac{d\mu}{d\mu'}}$, where $\frac{d\mu}{d\mu'}$ is the Radon-Nikodym derivative. We also define the $\chi_2$ divergence as $\chi_2(\mu\|\mu') \eqdef \E[\mu']{\pr{ \frac{d\mu}{d\mu'}}^2} -1$. We define $\avgI{X;Y} \eqdef \D{\mu_{XY}}{\mu_X\otimes \mu_Y}$  where $\mu_{XY}$, $\mu_X$ and $\mu_Y$ denote the probability measures associated to $(X, Y)$, $X$, and $Y$, respectively. 

Let $\calX$ and $\calY$ be two subsets of $\mathbb{R}$. A channel $\wy$ from $\calX$ to $\calY$ is a mapping $x\mapsto \mu_x$ where $\mu_x$ is a probability measure on $\calY$. If $\mu_x$ is always continuous, we write $\wy(y|x)$ to denote the \ac{PDF} of $\mu_x$. If $\mu$ is a probability measure on $\calX$ and $\wy:x\mapsto \mu_x'$ is a channel from $\calX$ to $\calY$, we define a joint probability measure $w_{Y|X} \times \mu$  on $\calX\times \calY$ as
\begin{align}
(\mu \times w_{Y|X})(\calE) \eqdef \int \mu_x'(\calE_x) d\mu,
\end{align}
where $\calE_x \eqdef \{(\widetilde{x}, \widetilde{y})\in\calE: \widetilde{x} = x\}$. We also define the marginal probability measure induced on $\calY$ by $w_{Y|X} \circ \mu$. If $X$ and $Y$ denote the joint random variables associated to the measure $\mu\times w_{W|X}$, we allow ourselves to denote their mutual information by $I(\mu, \wy) \eqdef \avgI{X;Y}$. 

We shall use the standard asymptotic notations such as $O(\cdot)$, $o(\cdot)$, $\Omega(\cdot)$, $\omega(\cdot)$ and $\Theta(\cdot)$.
}

\section{System model and notations}
\label{sec:system-model}

We consider the fast Rayleigh-fading wireless channel illustrated in Fig.~\ref{fig:rayleigh_covert_channel}, in which at every time instant, the input-output relationships are given by
\begin{align}
  \label{eq:input-output}
  Y=H_{m}X+N_{m} \quad \textnormal{and} \quad Z=H_{w}X+N_{w},	 
\end{align}
where $X$ is the channel input, $Y$ is the received signal at the legitimate receiver, and $Z$ is the received signal at the warden attempting to detect the transmission.  The fading coefficients $H_{m}$ and $H_{w}$ are independent complex circular Gaussian random variables with zero-mean and variances $\theta_{m}^2$ and $\theta_{w}^2$, respectively. The noises $N_{m}$ and $N_{w}$ are also  independent zero-mean complex circular random variables with variance $\sigma_{m}^2$ and $\sigma_{w}^2$, respectively. Furthermore, we assume that the channels are stationary and memoryless. The fading coefficients are unknown to all parties, who only have access to their statistical distributions. Since the phase of the fading parameters is uniform, information can only be encoded into the magnitude of $X$; additionally, $\abs{Y}^2$ and $\abs{Z}^2$ become sufficient statistics for detection.
\modfirst{Hence, as shown in~\cite{Abou-Faycal2001}, upon re-labeling $\abs{X}^2$ by $X$ and the outputs $\abs{Y}^2$ and $\abs{Z}^2$ by $Y$ and $Z$, the non-coherent channel is effectively a new memoryless channel with input and output symbols in $[0, \infty[$ and transition probabilities }
\begin{align}
  w_{Y|X}(y|x)&=\frac{1}{\theta_m^2x+\sigma_m^2}\exp\left(-\frac{y}{\theta_m^2x+\sigma_m^2}\right) \quad\text{ and }\quad  w_{Z|X}(z|x)=\frac{1}{\theta_w^2x+\sigma_w^2}\exp\left(-\frac{z}{\theta_w^2x+\sigma_w^2}\right) .\label{eq:input-output-1}
\end{align}
By properly normalizing $Y$ and $Z$, we can assume that $\sigma_w = \sigma_m = 1$, and by normalizing $X$, we can further assume that $\theta_w = 1$. Thus, we can parameterize the channel by a single parameter\footnote{Note that $\theta_m$ in~\eqref{eq:input-output-1} is different from $\theta_m$ in~\eqref{eq:input-output-2}.} $\theta_m$, for which the transition probabilities are
\begin{align}
p_x(y)\eqdef  w_{Y|X}(y|x)&=\frac{1}{\theta_m^2x+1}\exp\left(-\frac{y}{\theta_m^2x+1}\right) \quad\text{ and }\quad  q_x(z)\eqdef w_{Z|X}(z|x)=\frac{1}{x+1}\exp\left(-\frac{z}{x+1}\right) .\label{eq:input-output-2}
\end{align}
Although the input and output sets of the channels are all equal to $[0, \infty[$, we distinguish them with the labels $\calX$, $\calY$, and $\calZ$ for the input set, the output of main channel, and the output of the warden's channel, respectively. 
\begin{figure}[h]
\centering
  \includegraphics[width=0.7\textwidth]{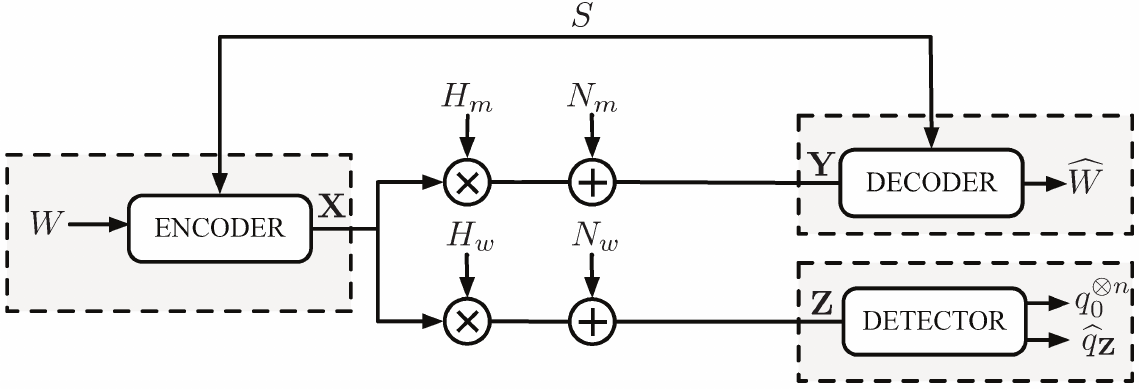}
  \caption{Covert Wireless Channel}
  \label{fig:rayleigh_covert_channel}
\end{figure}

We next formally describe  the covert communication problem in the wireless setting; as depicted in Fig.~\ref{fig:rayleigh_covert_channel},  the transmitter aims to communicate a message $W \in \intseq{1}{M_n}$ by encoding it into a sequence $\mathbf{X} = (X_1, \cdots, X_n)$ of $n$ symbols using a publicly known coding scheme. Upon observing the corresponding noisy sequence $\mathbf{Y} = (Y_1, \cdots, Y_n)$, the receiver forms an estimate $\widehat{W}$ of $W$. The encoding and decoding may also use a pre-shared secret key $S$ with an arbitrary distribution over a measurable space.\footnote{We show in our achievability proof that a key uniformly distributed over a discrete set with size $O(M_n)$ is sufficient to achieve the covert capacity.} The objective of the warden is to detect the presence of a transmission based on its noisy observation $\mathbf{Z} = (Z_1, \cdots, Z_n)$. The requirements for reliable and covert communication may be formalized as follows. We let $\widehat{q}_{\mathbf{Z}}$ denote the output distribution induced by the coding scheme and $q_0^{\otimes n}$ the product output distribution expected in the absence of communication when the channel input is set to $x=0$. The performance of an $(M_n,n)$ code transmitting one of $M_n$ message over $n$ channel uses is then measured in terms of the average probability of error $\mathbb{P}(\widehat{W} \neq W)$ and in terms of the relative entropy $\mathbb{D}(\widehat{q}_{\mathbf{Z}}\|q_0^{\otimes n})$.\footnote{\modfirst{The constraint ${\mathbb{D}(\widehat{q}_{\mathbf{Z}}\|q_0^{\otimes n})} \leq \delta$ ensures that, regardless of the test performed by the adversary, the sum of the probability of missed detection and false alarm is lower-bounded by $1-\sqrt{\delta}$. Please refer to \cite[Appendix~A]{Bloch2015b} for a detailed discussion of the operational meaning of an upper-bound on the relative entropy.}}\footnote{\modfirst{    The choice of this specific relative entropy to measure covertness is driven in part by the ease of analysis using channel resolvability techniques. One could of course consider alternative metrics, such as variational distance or a relative entropy with a reversed order of arguments, as discussed in~\cite{Bloch2015b, Tahmasbi2017}. While the operational meaning of these other metrics remains the same, the analysis and the exact dependence on the constraint $\delta$ is metric-specific.}} \modfirst{Let $\delta>0$.} We say that a covert throughput $R$  is \modfirst{$\delta$-}achievable if there exist $(M_n,n)$ codes of increasing block length $n$ such that 

\begin{align}
  \label{eq:throughput}
\log M_n = \omega(\log n), \quad   \lim_{n\rightarrow\infty}\modfirst{\mathbb{P}(\widehat{W} \neq W)}=0,\quad   \limsup_{n\rightarrow\infty} \modfirst{\mathbb{D}(\widehat{q}_{\mathbf{Z}}\|q_0^{\otimes n})} \leq \delta, \quad \liminf_{n \rightarrow \infty}\frac{\log M_n}{\sqrt{n \modfirst{\mathbb{D}(\widehat{q}_{\mathbf{Z}}\|q_0^{\otimes n})}}}\geq R.
\end{align}
The covert capacity, $C_{\text{no-CSI}}\modfirst{(\delta)}$, is defined as the supremum of all  \modfirst{$\delta$-}achievable covert throughputs. Note that we do not specify $\delta$ in our terminology of achievable throughput, since it turns out that the normalization of $\log M_n$ in~(\ref{eq:throughput}) removes the dependence on $\delta$. 


\modfirst{
\begin{theorem}
\label{th:general-capacity-non-coherent}

Let $\widetilde{\Omega}^{>0}$ be the set of discrete probability measures over $]0, 1[$ with a finite number of mass points. $C_{\mathrm{no-CSI}}(\delta)$ is independently of $\delta$ equal to

\begin{align}
\label{eq:cov-cap}
\sup_{\mu \in \widetilde{\Omega}^{>0}}  \sqrt{2}\frac{\E[\mu]{\theta_m^2 X - \log\pr{1+\theta_m^2 X}}}{\sqrt{\E[\mu\otimes \mu]{\frac{X_1X_2}{1-X_1X_2}}}}.
\end{align}
In addition, the following simple bounds hold:
\begin{align}
\max_{\widetilde{x}\in]0, 1]}\td{x}^{-1}\sqrt{2(1-\widetilde{x}^2)} \pr{\theta_m^2\widetilde{x} - \log(1+\theta_m^2 \widetilde{x})}\leq C_{\textnormal{no-CSI}}(\delta)  \leq  \sqrt{2} \theta_m^2.\label{eq:bounds}
\end{align}
\end{theorem}
}

\modfirst{Theorem~\ref{th:general-capacity-non-coherent} provides useful insight into the problem of covert communication over non-coherent channels in several regards. First, a straightforward calculation shows that $\D{p_x}{p_0} = \theta_m^2 x - \log(1+\theta_m^2x)$ and $\chi_2(\wz \circ \mu\|q_0) = \E[\mu\otimes \mu]{\frac{X_1X_2}{1-X_1X_2}}$. The expression in \eqref{eq:cov-cap} is therefore a counterpart of \cite[Corollary 3]{Bloch2015b} and \cite[Eq. (28)]{Wang2016b}. Second,  Theorem~\ref{th:general-capacity-non-coherent} shows that we may restrict the signaling schemes for covert communications to finite and amplitude bounded constellations; while the finite nature of the constellation was somewhat expected from the non-coherent nature of the channel, the bound on the amplitude of the points is perhaps more surprising as it was not imposed a priori. 
We numerically evaluate and plot  in Fig.~\ref{fig:evaluation} \eqref{eq:cov-cap} when \emph{the number of mass points in $\mu$ is fixed} using a brute-force search. 
     Based on our numerical results, we conjecture that two mass points and \modfirst{\ac{OOK}} signaling is optimal for covert communication.}


\begin{figure}[h]
  \centering
  \includegraphics[width=.7\linewidth]{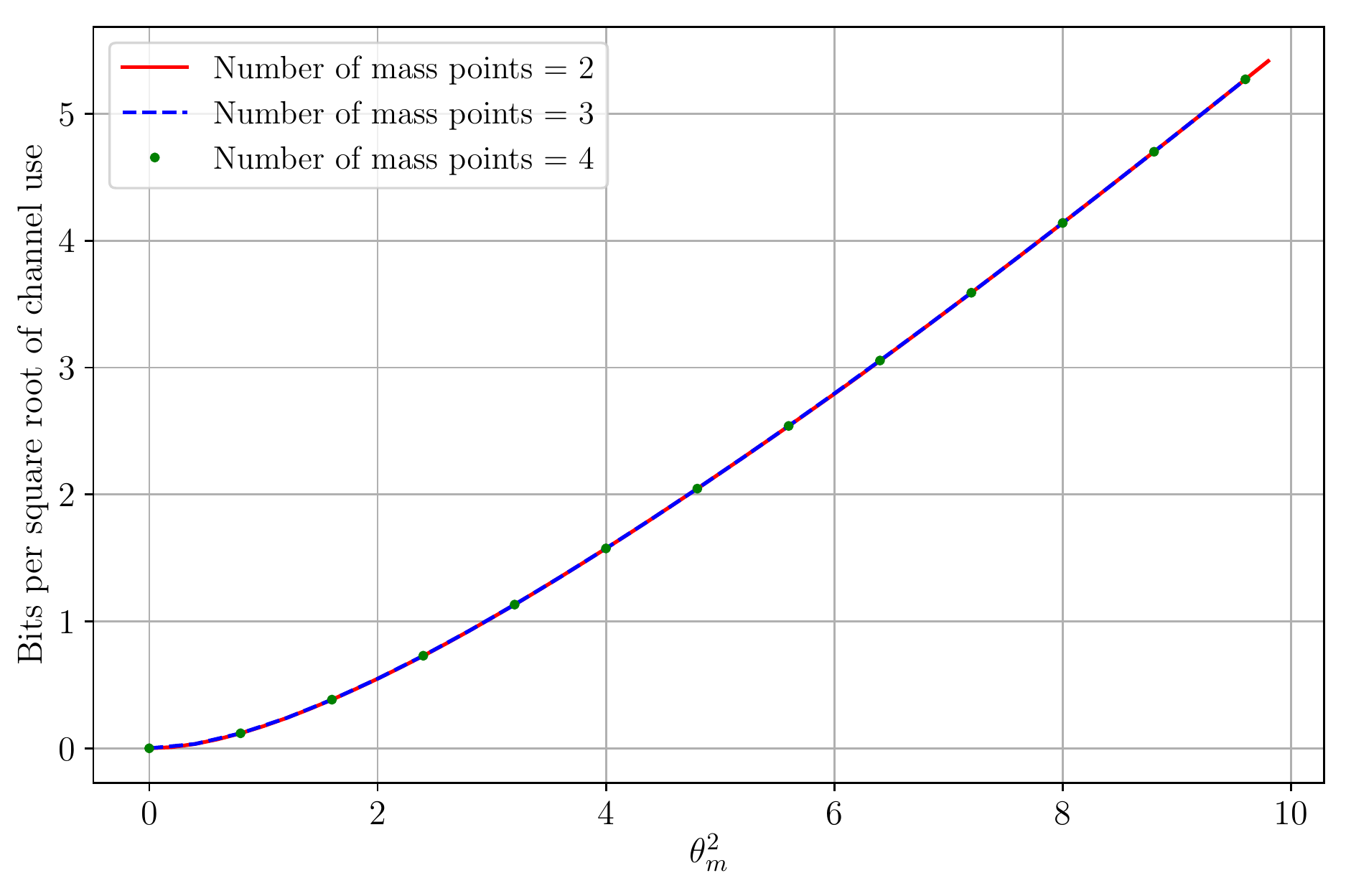}
  \caption{Numerical evaluation of bounds on covert capacity. }
  \label{fig:evaluation}
\end{figure}

\section{Proof of Theorem~\ref{th:general-capacity-non-coherent}}
\label{sec:proof-theor-non-coherent}

\subsection{Achievability proof}
\label{sec:achievability-proof}
\modfirst{
We prove the achievability result in two steps.
\begin{enumerate}
\item Let $\{\mu_n\}_{n\geq 1}$ be a sequence of probability measures over $\calX$  such that for all $n$, 
 \begin{inparaenum}[(i)]
  \item for some $\widetilde{x} > 0$, $\modfirst{\sup}(\supp{\mu_n})\leq \widetilde{x}$; and,
  \item  $\limsup_{n\to \infty} n\D{\wz\circ \mu_n}{q_0} = \delta$.
  \item $nI(\mu_n, \wy) = \omega(\log n)$.
\end{inparaenum}
We then show for all $\zeta>0$ that the cover throughput
\begin{align}
(1-\zeta)\liminf_{n\to \infty} \frac{I(\mu_n, w_{Y|X})}{\sqrt{\D{w_{Z|X}\circ \mu_n}{q_0}}}
\end{align}
is $\delta$-achievable.
\item Let $\mu\in \td{\Omega}^{>0}$. We construct for an arbitrary $\delta>0$, a sequence $\{\mu_n\}_{n\geq 1}$ satisfying 
\begin{align}
\liminf_{n\to \infty} \frac{I(\mu_n, w_{Y|X})}{\sqrt{\D{w_{Z|X}\circ \mu_n}{q_0}}} =  \sqrt{2}\frac{\E[\mu]{\theta_m^2 X - \log\pr{1+\theta_m^2 X}}}{\sqrt{\E[\mu\otimes \mu]{\frac{X_1X_2}{1-X_1X_2}}}},
\end{align}
in addition to the conditions of step 1.
\end{enumerate}}
\subsubsection{Step one: a random coding argument}
\modfirst{Although we pursue the same approach as in \cite{Bloch2015b, WWZ15} in this step, the result  requires a proof of its own because of the continuous nature of the channels. Let $\{\mu_n\}_{n\geq 1}$ be a sequence of probability measures as described earlier, i.e., for all $n$, 
 \begin{inparaenum}[(i)]
  \item for some $\widetilde{x} > 0$, $\modfirst{\sup}(\supp{\mu_n})\leq \widetilde{x}$; and,
  \item  $\limsup_{n\to \infty} n\D{\wz\circ \mu_n}{q_0} = \delta$.
  \item $I(\mu_n, \wy)n = \omega(\log n)$
\end{inparaenum} For any $\zeta > 0$, we shall prove the existence of a sequence of codes $\{\calC_n\}_{n\geq1}$ achieving the covert throughput $(1-\zeta)\liminf_{n\to \infty} \frac{I(\mu_n, w_{Y|X})}{\sqrt{\D{w_{Z|X}\circ \mu_n}{q_0}}}$ with the relative entropy constraint $\delta$.}  We use a random coding argument and in particular, fix some $n$, and consider a random encoder $F:\intseq{1}{K_n}\times\intseq{1}{M_n} \to \calX^n$ whose codewords are \ac{iid} according to $\mu_n^{\pn}$. The transmitter uses the message $W$ and the shared key $S$ together with the encoder $F$ to obtain the codeword $F(S, W)$ that is transmitted through the channel. By~\cite{Polyanskiy2010}, for any $\gamma>0$, we upper-bound the expected value with respect to random coding of the probability of error of an optimal decoder by
\begin{align}
\label{eq:rel-oneshot}
\E[F]{\P{W\neq \widehat{W}}} \leq \P[w_{Y|X}^{\pn} \times \mu_n^{\pn}]{\log \frac{w_{Y|X}^\pn(\mathbf{Y}|\mathbf{X})}{(w_{Y|X}^\pn\circ \mu_n^\pn)(\mathbf{Y})}\geq \gamma} + M_ne^{-\gamma}.
\end{align}
Applying a Chernoff bound to the first term of the right hand side of the above inequality, for all $s>0$, we obtain 
\begin{align}
 \P[w_{Y|X}^{\pn} \times \mu_n^{\pn}]{\log \frac{w_{Y|X}^{\pn}(\mathbf{Y}|\mathbf{X})}{(w_{Y|X}^\pn\circ \mu_n^\pn)(\mathbf{Y})}\geq \gamma} 
 &\modfirst{\leq}  \pr{ \E[w_{Y|X} \times \mu_n]{\pr{\frac{p_X(Y)}{(w_{Y|X}\circ \mu_n)(Y)}}^s}}  ^n\exp\pr{-s\gamma}. \label{eq:rel-exponent}
\end{align}
For any probability measure $\mu$ on $\calX$, upon defining
\begin{align}
\phi_{\text{rel}}(s, \mu) \eqdef -\log\pr{\E[w_{Y|X} \times \mu]{\pr{\frac{p_X(Y)}{(w_{Y|X}\circ \mu)(Y)}}^s}},
\end{align}
we can re-write the right-hand side of \eqref{eq:rel-exponent} as
\begin{align}
\label{eq:pe-phi}
\exp\pr{-n\phi_{\text{rel}}(s, \mu_n) - s\gamma}.
\end{align} 
To upper-bound the above expression, we need the following technical lemma describing the behavior of $\phi_{\text{rel}}(s, \mu)$ for small $s$.
\begin{lemma}
\label{lm:exponent-rel}
For all $\widetilde{x} > 0$, there exist  constants $B>0$, $\widetilde{s} > 0$, and $\widetilde{A}>0$, such that for all probability measures $\mu$ \modfirst{and $\nu>0$} with $\sup(\supp{\modfirst{\mu}}) \leq \widetilde{x}$ and $\D{\wz \circ \mu}{q_0} \leq \nu$, all $s \in ]0, \widetilde{s}]$, and all $A\in [\widetilde{A}, \infty)$, we have
\begin{align}
\phi_{\textnormal{rel}}(s, \mu) \geq -s I(\mu, w_{Y|X}) -  B\pr{\pr{(2\sqrt{\nu} + \nu)e^{2A\widetilde{x} + 2A}A^2 + A^2 e^{-A}}s^2 + s^3}.
\end{align}
\end{lemma}
\begin{proof}
See Appendix~\ref{sec:error-expon-analys}.
\end{proof}
Applying Lemma~\ref{lm:exponent-rel} to \eqref{eq:pe-phi}, we upper-bound~\eqref{eq:rel-exponent} by
\begin{align}
\exp\pr{-n\pr{ -s I(\mu_n, w_{Y|X}) -  B\pr{\pr{\pr{2\sqrt{\frac{\delta}{n}} + \frac{\delta}{n}}e^{2A\widetilde{x} + 2A}A^2 + A^2 e^{-A}}s^2 + s^3}} -s\gamma }.
\end{align}
For $n$ large enough, we then set $A = \log n/(2(4\widetilde{x} +1))$ to ensure
\modfirst{
\begin{align}
B\pr{\pr{2\sqrt{\frac{\delta}{n}} + \frac{\delta}{n}}e^{2A\widetilde{x} + 2A}A^2 + A^2 e^{-A}} 
&= O\pr{A^2\pr{n^{-\frac{1}{2}} e^{2A\td{x} + 2A} e^{-A} }}\\
&= O\pr{\log^2n \pr{n^{-\frac{1}{2} + \frac{2\td{x} + 2}{2(4\td{x} + 1)}} + n^{-\frac{1}{2(4\td{x} + 1)}}}}\\
&= O\pr{n^{-\frac{1}{2(4\widetilde{x} + 1)}} \log^2n},\label{eq:exp-second}
\end{align}

where the constant hidden in $O\pr{\cdot}$ depends on $\widetilde{x}$, $\delta$, and the channel. Therefore, we have for $s=n^{-\beta}$, 
\begin{align}
 B\pr{\pr{\pr{2\sqrt{\frac{\delta}{n}} + \frac{\delta}{n}}e^{2A\widetilde{x} + 2A}A^2 + A^2 e^{-A}}s^2 + s^3}  
&\stackrel{(a)}{=} O\pr{ n^{-\frac{1}{2(4\widetilde{x} + 1)}} \log^2ns^2 + s^3}\\
&=  O\pr{ n^{-\frac{1}{2(4\widetilde{x} + 1)} - 2\beta} \log^2n + n^{-3\beta}}\\
&= O\pr{n^{-\min(3\beta, 2\beta + 1/(2(4\widetilde{x} +1)))  }\log^2n},\label{eq:exp-second-2}
\end{align}
where $(a)$ follows from \eqref{eq:exp-second}. The expression in \eqref{eq:exp-second-2} will be $o(sI(\mu_n, \wy))$ when  $I(\mu_n, \wy) = \Omega\pr{n^{-\frac{1}{2}}}$ and $\max(1/4, 1/2-1/(2(4\widetilde{x} +1))) < \beta$.} Moreover, if we choose $\beta < 1/2$, which is feasible with the previous constraint, we guarantee that $snI(\mu_n, \wy) \modfirst{\geq n^{c}}$ for some $c>0$ and $n$ large enough. Finally, for $\gamma = (1-\zeta/2) I(\mu_n, \wy) n$ and $\log M_n = (1-\zeta) I(\mu_n, \wy) n$, we have by \eqref{eq:rel-oneshot}
\begin{align}
\E[F]{\P{W \neq \widehat{W}}}  
&\leq 
\exp\pr{-(1+o(1))\frac{\zeta}{2} sI(\mu_n, \wy)n} + \exp\pr{-\frac{\zeta}{2} I(\mu_n, \wy) n}\\
&\leq 2\exp\pr{-\modfirst{\zeta}n^c}.
\end{align}
This completes the reliability part of the proof.

We now proceed to the resolvability part. \modfirst{Recall that we denote the induced distribution at the output of the warden's channel by $\widehat{p}_{\mathbf{Z}} \eqdef \frac{1}{M_nK_n} \sum_{s=1}^{K_n}\sum_{w=1}^{M_n} \wz^\pn(\mathbf{z}|F(s, w))$, where $M_n$ and $K_n$ are the message size and the key size, respectively.} \myr{By a modification of \cite[Equation (194)]{hayashi2015quantum}, we know that for all $s\in[0, 1]$,
\begin{align}
\E[F]{\D{\widehat{p}_{\mathbf{Z}}}{\pr{\wz\circ \mu_n}^\pn}} 
&\leq \frac{1}{s}\exp\pr{-s\log (M_nK_n) - n \phi_{\text{res}}(s, \mu_n)},
\end{align}
where
\begin{align}
\phi_{\text{res}}(s, \mu)\eqdef - \log \pr{ \E[\wz\times \mu]{\pr{\frac{q_X(Z)}{(\wz\circ \mu)(Z)}}^s}}.
\end{align}
Since the above function is the same as $\phi_{\text{rel}}$ except that $\wy$ is replaced by $\wz$, $\wz$ is a special case of $\wy$ for $\theta_m=1$, and we choose $s$ in the reliability part so that $\log \frac{1}{s} = O(\log n)$, we can follow the same approach to obtain for some $\widetilde{c}>0$,
\begin{align}
\E[F]{\D{\widehat{p}_{\mathbf{Z}}}{\pr{\wz\circ \mu_n}^\pn}} \leq 2\exp\pr{-n^{\widetilde{c}}},
\end{align}
if $\log M_n + \log K_n \geq (1+\zeta) I(\mu_n, \wz)n$. }Therefore, the expected value of the covertness of the random code is
\begin{align}
\E[F]{\D{\widehat{p}_{\mathbf{Z}}}{q_0^\pn}} 
&= \E[F]{\int_{\mathbb{R}^n} \widehat{p}_{\mathbf{Z}}(\mathbf{z})\log \frac{\widehat{p}_{\mathbf{Z}}(\mathbf{z})}{q_0^\pn(\mathbf{z})} d\mathbf{z}}\\
&= \E[F]{\int_{\mathbb{R}^n} \widehat{p}_{\mathbf{Z}}(\mathbf{z})\log \frac{\widehat{p}_{\mathbf{Z}}(\mathbf{z})}{(\wz\circ \mu_n)^\pn(\mathbf{z})} d\mathbf{z}+\int_{\mathbb{R}^n} \widehat{p}_{\mathbf{Z}}(\mathbf{z})\log \frac{(\wz\circ \mu_n)^\pn(\mathbf{z})}{q_0^\pn(\mathbf{z})} d\mathbf{z}}\\
&= \E[F]{\D{\widehat{p}_{\mathbf{Z}}}{\pr{\wz\circ \mu_n}^\pn}}+ \E[F]{\int_{\mathbb{R}^n} \widehat{p}_{\mathbf{Z}}(\mathbf{z})\log \frac{(\wz\circ \mu_n)^\pn(\mathbf{z})}{q_0^\pn(\mathbf{z})} d\mathbf{z}}\\
&\leq 2\exp\pr{-n^{\widetilde{c}}} + \E[F]{\int_{\mathbb{R}^n} \widehat{p}_{\mathbf{Z}}(\mathbf{z})\log \frac{(\wz\circ \mu_n)^\pn(\mathbf{z})}{q_0^\pn(\mathbf{z})} d\mathbf{z}}\\
&\stackrel{(a)}{=}  2\exp\pr{-n^{\widetilde{c}}} + \int_{\mathbb{R}^n} \E[F]{\widehat{p}_{\mathbf{Z}}(\mathbf{z})}\log \frac{(\wz\circ \mu_n)^\pn(\mathbf{z})}{q_0^\pn(\mathbf{z})} d\mathbf{z}\\
&=  2\exp\pr{-n^{\widetilde{c}}}  +  \int_{\mathbb{R}^n}(\wz\circ \mu_n)^\pn(\mathbf{z})\log \frac{(\wz\circ \mu_n)^\pn(\mathbf{z})}{q_0^\pn(\mathbf{z})} d\mathbf{z}\\
&= 2\exp\pr{-n^{\widetilde{c}}}+ n \D{\wz\circ\mu_n}{q_0},
\end{align}
where $(a)$ follows from Fubini's theorem and  $\E[F]{\int_{\mathbb{R}^n} \widehat{p}_{\mathbf{Z}}(\mathbf{z})\left|\log \frac{(\wz\circ \mu_n)^\pn(\mathbf{z})}{q_0^\pn(\mathbf{z})}\right| d\mathbf{z}} < \infty$ by Lemma~\ref{lm:cross-entropy}. Applying Markov's inequality, for large $n$, we obtain 
\begin{align}
\displaybreak[0]&\P[F]{\D{\widehat{p}_{\mathbf{Z}}}{q_0^\pn} \leq \frac{n + 1}{n}\pr{ 2\exp\pr{-n^{\widetilde{c}}} + n \D{\wz\circ\mu_n}{q_0}}, \P{W \neq \widehat{W}} \leq 4n\exp\pr{-n^c}} \\
\displaybreak[0]&\geq 1 - \P[F]{\D{\widehat{p}_{\mathbf{Z}}}{q_0^\pn} \geq \frac{n + 1}{n}\pr{ 2\exp\pr{-n^{\widetilde{c}}} +n \D{\wz\circ\mu_n}{q_0}}} - \P[F]{\P{W \neq \widehat{W}} \geq 4n\exp\pr{-n^c}}\\
\displaybreak[0]&\geq 1-  \frac{\E[F]{\D{\widehat{p}_{\mathbf{Z}}}{q_0^\pn}}}{\frac{n + 1}{n}\pr{ 2\exp\pr{-n^{\widetilde{c}}} + n \D{\wz\circ\mu_n}{q_0}}} - \frac{\E[F]{\P{W\neq \widehat{W}}}}{4n\exp\pr{-n^c}}\\
\displaybreak[0]&\geq 1-  \frac{2\exp\pr{-n^{\widetilde{c}}} + n \D{\wz\circ\mu_n}{q_0}}{\frac{n + 1}{n}\pr{ 2\exp\pr{-n^{\widetilde{c}}} + n \D{\wz\circ\mu_n}{q_0}}} - \frac{2\exp\pr{-n^c}}{4n\exp\pr{-n^c}}\\
\displaybreak[0]&= 1- \frac{n}{n+1} - \frac{1}{2n} > 0.
\end{align}
This implies that there exists a sequence of codes $\{\calC_n\}_{n\geq 1}$ such that $\calC_n$ satisfies 
\begin{align}
\log M_n &= (1-\zeta) I(\mu_n, \wy) n = \modfirst{\omega(\log n)},\\
\log M_n + \log K_n &\modfirst{=} (1+\zeta) I(\mu_n, \wz)n,\\
P_e &\leq 4n\exp\pr{-n^c},\\
\D{\widehat{p}_{\mathbf{Z}}}{q_0^\pn} &\leq \frac{n + 1}{n}\pr{ 2\exp\pr{-n^{\widetilde{c}}} + n \D{\wz\circ\mu_n}{q_0}}.
\end{align}
\modfirst{
The covert throughput would be then
\begin{align}
\liminf_{n\to \infty} \frac{\log M_n}{\sqrt{n\D{\widehat{p}_{\mathbf{Z}}}{q_0^\pn}}} 
&= \liminf_{n\to \infty} \frac{ (1-\zeta) I(\mu_n, \wy) n}{\sqrt{n\D{\widehat{p}_{\mathbf{Z}}}{q_0^\pn}}} \\
&\geq  \liminf_{n\to \infty} \frac{ (1-\zeta) I(\mu_n, \wy) n}{\sqrt{n\pr{ \frac{n + 1}{n}\pr{ 2\exp\pr{-n^{\widetilde{c}}} + n \D{\wz\circ\mu_n}{q_0}}}}}\\
&= \liminf_{n\to \infty} \frac{ (1-\zeta) I(\mu_n, \wy) }{\sqrt{ \D{\wz\circ\mu_n}{q_0}}}.
\end{align}
Since $\limsup_{n\to \infty} n \D{\wz\circ\mu_n}{q_0} \leq \delta$ by our assumption, we have
\begin{align}
\limsup_{n\to \infty} \D{\widehat{p}_{\mathbf{Z}}}{q_0^\pn}  \leq \limsup_{n\to \infty} \pr{ \frac{n + 1}{n}\pr{ 2\exp\pr{-n^{\widetilde{c}}} + n \D{\wz\circ\mu_n}{q_0}} }\leq \delta.
\end{align}}

\subsubsection{Step two: obtaining the bound in Theorem~\ref{th:general-capacity-non-coherent}}
\modfirst{
Let $\mu \in \td{\Omega}^{>0}$ and $\mu_0$ be the probability measure with a single mass point at zero. We define $\alpha_n \eqdef  \sqrt{\frac{\delta}{n\chi_2(\wz\circ \mu\|q_0)}}$ and $\mu_n \eqdef \alpha_n \mu + (1-\alpha_n) \mu_0$. We have $\max(\supp{\mu_n}) = \max(\supp{\mu}) \eqdef a <  1$ by definition of $ \td{\Omega}^{>0}$. Hence, it is enough to check that
\begin{align}
nI(\mu_n, \wy) &= \omega(\log n),\\
\limsup_{n\to \infty}n \D{\wz \circ \mu_n}{q_0} &\leq \delta,\\
 \liminf_{n\to \infty} \frac{I(\mu_n, \wy)}{\sqrt{\D{\wz\circ \mu_n}{ q_0}}}  &\geq \sqrt{2}\frac{\E[{\mu}]{\theta_m^2 X - \log\pr{1+\theta_m^2 X}}}{\sqrt{\E[{\mu}\otimes {\mu}]{\frac{X_1X_2}{1-X_1X_2}}}}.
\end{align}
We next state a lemma providing a general upper-bound for the relative entropy in terms of the $\chi_2$ divergence.
\begin{lemma}
\label{lm:d-upper}
Let  $\mu\in \td{\Omega}^{\geq 0}$ with $\max(\supp{\mu}) \leq a < 1$. Let $M>0$ and $\epsilon>0$. We have
\begin{multline}
\D{\wz \circ \mu}{q_0} \leq \frac{1}{2}\chi_2\pr{\wz\circ \mu\|q_0}  + (\E[\mu]{X})^3 +\pr{\E[\mu]{X}}^4 \int_0^M e^{z\pr{- 1 + \frac{4a}{1+a}}}dz + \int_M^\infty e^{-\frac{z}{1+\epsilon}}z dz \\
+ \frac{\E[\mu]{X}}{\epsilon} \int_M^\infty e^{-\frac{z}{1+a}} zdz
\end{multline}
\end{lemma}
\begin{proof}
See Appendix~\ref{sec:chi-2-bound}.
\end{proof}
Applying Lemma~\ref{lm:d-upper} to $\mu_n$ with some $M_n$ and $\epsilon$, we obtain
\begin{multline}
\D{\wz \circ \mu_n}{q_0}\leq  \frac{1}{2}\chi_2\pr{\wz\circ \mu_n\|q_0}  + (\E[\mu_n]{X})^3\\ +\pr{\E[\mu_n]{X}}^4 \int_0^{M_n} e^{z\pr{- 1 + \frac{4a}{1+a}}}dz + \int_{M_n}^\infty e^{-\frac{z}{1+\epsilon}}z dz 
+ \frac{1}{\epsilon}\E[\mu_n]{X} \int_{M_n}^\infty e^{-\frac{z}{1+a}} zdz,
\end{multline}
where $a = \max(\supp{\mu}) $. We will prove for appropriately chosen $M_n$ and $\epsilon$ that
\begin{align}
\label{eq:limits-achi}
 (\E[\mu_n]{X})^3 +\pr{\E[\mu_n]{X}}^4 \int_0^{M_n} e^{z\pr{- 1 + \frac{4a}{1+a}}}dz + \int_{M_n}^\infty e^{-\frac{z}{1+\epsilon}}z dz 
+ \frac{1}{\epsilon}\E[\mu_n]{X} \int_{M_n}^\infty e^{-\frac{z}{1+a}} zdz = o(\alpha_n^2)
\end{align}
Note that $\E[\mu_n]{X} = \alpha_n \E[\mu]{X}$, and therefore, $(\E[\mu_n]{X})^3 = O(\alpha_n^3) = o(\alpha_n^2)$. We choose $M_n = B\log \frac{1}{\alpha_n}$, where $B$ is a constant independent of $n$ specified later.  We then have
\begin{align}
\pr{\E[\mu_n]{X}}^4 \int_0^{M_n} e^{z\pr{- 1 + \frac{4a}{1+a}}}dz 
&\leq \begin{cases}O\pr{\alpha_n^4M_ne^{M_n \pr{-1+\frac{4a}{1+a}}}}&\quad a > 1/3  \\ O\pr{\alpha_n^4M_n}&\quad a \leq 1/3\end{cases}\\
&= \begin{cases}O\pr{\alpha_n^{4-B\frac{3a-1}{a+1}}\log \frac{1}{\alpha_n}}&\quad a > 1/3  \\ O\pr{\alpha_n^4\log \frac{1}{\alpha_n}}&\quad a \leq 1/3\end{cases}\\
&\stackrel{(a)}{=} o(\alpha_n^2),
\end{align}
where $(a)$ requires that $B<2\frac{1+a}{3a-1}$ when $a>1/3$. We further have
\begin{align}
\int_{M_n}^\infty e^{-\frac{z}{1+\epsilon}}z dz  
&= (1+\epsilon)^2e^{- \frac{ M_n}{1+\epsilon}} \pr { \frac{M_n}{1+\epsilon} + 1}\\
&= (1+\epsilon)^2 \alpha_n^{\frac{B}{1+\epsilon}} \pr { \frac{B\log \frac{1}{\alpha_n}}{1+\epsilon} + 1}\\
&\stackrel{(a)}{=} o(\alpha_n^2),
\end{align}
where $(a)$ requires that $B>2(1+\epsilon)$.
Finally, we have
\begin{align}
\frac{1}{\epsilon} \E[\mu_n]{X} \int_{M_n}^\infty e^{-\frac{z}{1+a}} zdz
 &=  \frac{1}{\epsilon} \E[\mu_n]{X} (1+a)^2e^{- \frac{ M_n}{1+a}} \pr { \frac{M_n}{1+a} + 1}\\
  &= \frac{1}{\epsilon} \E[\mu_n]{X}  (1+a)^2\alpha_n^{\frac{B}{1+a}} \pr { \frac{B\log \frac{1}{\alpha_n}}{1+a} + 1}\\
  &\stackrel{(a)}{=} o(\alpha_n^2),
\end{align}
where $(a)$ requires that $B>1+a$. If $a\leq 1/3$, we only need to choose $B$ and $\epsilon$ such that $B > \max(2(1+\epsilon), 1+a)$. For $a>1/3$, we choose $0<\epsilon< \frac{1+a}{3a-1} - 1$ so that $\max(2(1+\epsilon), 1+a) < 2\frac{1+a}{3a-1}$. We then choose $B$ such that $\max(2(1+\epsilon), 1+a) < B < 2\frac{1+a}{3a-1}$. This complete the proof of \eqref{eq:limits-achi}. Note next that by Lemma~\ref{lm:chi2-value}
\begin{align}
\chi_2(\wz \circ \mu_n \| q_0) 
&= \E[\mu_n\otimes \mu_n]{\frac{X_1X_2}{1-X_1X_2}}\\
&=   \E[\mu_n\otimes \mu_n]{\frac{X_1X_2}{1-X_1X_2}|X_1> 0, X_2 > 0} \P[\mu_n\otimes \mu_n]{X_1 > 0, X_2 > 0 }\\
&\stackrel{(a)}{=} \alpha_n^2  \E[\mu\otimes \mu]{\frac{X_1X_2}{1-X_1X_2}} \\
&= \alpha_n^2 \chi_2(\wz \circ \mu \| q_0)\\
& \stackrel{(b)}{=} \frac{\delta}{n},
\end{align}
where $(a)$ follows from the definition of $\mu_n$ and $(b)$ follows from the definition of $\alpha_n$.
We therefore have
\begin{align}
\limsup_{n\to \infty}n \D{\wz \circ \mu_n}{q_0} &\leq \delta.
\end{align}
Following the same reasoning, one can show that $ \D{\wy \circ \mu_n}{p_0} = O(\alpha_n^2)$.
Finally, we have
\begin{align}
I(\mu_n, \wy) 
&= \E[\mu_n]{\theta_m^2 X - \log(1+\theta_m^2 X)} - \D{\wy \circ \mu_n}{ p_0}\\
&= \E[\mu_n]{\theta_m^2 X - \log(1+\theta_m^2 X)} - O(\alpha_n^2)\\
&= \alpha_n \E[\mu]{\theta_m^2 X - \log(1+\theta_m^2 X)} - O(\alpha_n^2)\\
&= \Omega(n^{-\frac{1}{2}})\\
&= \omega\pr{\frac{\log n}{n}},
\end{align}
which yields that
\begin{align}
 \liminf_{n\to \infty} \frac{I(\mu_n, \wy)}{\sqrt{\D{\wz\circ \mu_n}{ q_0}}}  &\geq \sqrt{2}\frac{\E[{\mu}]{\theta_m^2 X - \log\pr{1+\theta_m^2 X}}}{\sqrt{\E[{\mu}\otimes {\mu}]{\frac{X_1X_2}{1-X_1X_2}}}}.
\end{align}
}
To obtain the lower-bound in \eqref{eq:bounds}, \modfirst{we choose $\mu$ to be a probability measure with a single mass point at $\td{x} \in ]0, 1[$. We then have
\begin{align}
 \sqrt{2}\frac{\E[{\mu}]{\theta_m^2 X - \log\pr{1+\theta_m^2 X}}}{\sqrt{\E[{\mu}\otimes {\mu}]{\frac{X_1X_2}{1-X_1X_2}}}} =\td{x}^{-1}\sqrt{2(1-\widetilde{x})}\pr{\theta_m^2 \widetilde{x} - \log(1+\theta_m^2\widetilde{x})}.
\end{align}}
\subsection{Converse proof}
\label{sec:converse-proof}
\modfirst{
Before delving into the detailed proofs, we first provide the sketch of the various steps of the converse proof.
\begin{enumerate}
\item We first follow the reasoning of the converse proof of \cite{Bloch2015b} to show that if $R$ is a $\delta$-achievable rate, then there exists a sequence of probability measures $\{\mu_n\}_{n\geq}$ over $\calX$ such that $\D{\wz\circ\mu_n}{q_0}\leq \delta/n$ for $n$ and 
\begin{align}
R\leq  \liminf_{n\rightarrow\infty }\frac{I(\mu_n, \wy)}{\sqrt{\D{\wz \circ\mu_n}{q_0} }}.
\end{align} 
\item We show that the probability measure $\mu_n$ can be further restricted to be discrete  with a finite number of mass points and a mass point at zero. This is achieved by investigating the optimization problem
\begin{align}
\sup_{\mu: \D{\wz \circ \mu}{q_0} \leq \nu} I(\mu, \wy), 
\end{align}
and adapting some techniques developed in \cite{Abou-Faycal2001}.
\item  We prove that we can still upper-bound a covert throughput even if we constraint the amplitude of $\mu_n$ as $\max(\supp{\mu_n}) \leq 1 + \zeta$ for any $\zeta > 0$. 
\item Let $\{\mu_n\}_n{n\geq 1}$  be a sequence of probability measures  such that $\mu_n$ has a finite number of mass and $\max(\supp{\mu_n}) \leq 1 + \zeta$. We show that
\begin{align}
 \liminf_{n\rightarrow\infty }\frac{I(\mu_n, \wy)}{\sqrt{\D{\wz \circ\mu_n}{q_0} }} \leq \sup_{\mu\in \td{\Omega}^{>0}}\sqrt{2}\frac{\E[\mu]{\theta_m^2 X - \log\pr{1+\theta_m^2 X}}}{\sqrt{\E[\mu\otimes \mu]{\frac{X_1X_2}{1-X_1X_2}}}}.
\end{align}
\end{enumerate}}
\subsubsection{Step one: a general converse for covert communication}

We consider a sequence of code $\{\calC_n\}_{n\geq 1}$ where each code $\calC_n$ can transmit $\log M_n$ bits with probability of error $\epsilon_n$ and relative entropy at most $\delta_n$, and we have $\lim_{n\to\infty} \epsilon_n = 0$ and $\limsup_{n\to \infty} \delta_n \leq \delta$. If $(\mathbf{X}, \mathbf{Y}, \mathbf{Z})$ denotes the input and the output of the channels when $\calC_n$ is used and $\widehat{p}_{\mathbf{X}\mathbf{Y}\mathbf{Z}}$ denotes the joint distribution, a standard application of Fano's inequality yields
\begin{align}
  \log M_n \!\leq\! \frac{\mathbb{I}(\mathbf{X};\mathbf{Y})+\Hb{\epsilon_n}}{1-\epsilon_n}\!\leq\! \frac{\mathbb{I}(\mathbf{X};\mathbf{Y})+1}{1-\epsilon_n},
\end{align}
where $\Hb{x} \eqdef -x\log(x) - (1-x)\log(1-x)$. One can then upper-bound the mutual information $\avgI{\mathbf{X};\mathbf{Y}}$ using standard techniques~\cite{ElementsInformationTheory2} to obtain
\begin{align}
\mathbb{I}(\mathbf{X};\mathbf{Y})
&\leq \sum_{i=1}^n \mathbb{I}(X_i;Y_i)\leq n \mathbb{I}(\widetilde{X}_n;\widetilde{Y}_n),
\end{align}
where the random variables $\widetilde{X}_n$ and $\widetilde{Y}_n$ are distributed according to $p_{\widetilde{X}_n}(x)\eqdef \frac{1}{n}\sum_{i=1}^n \widehat{p}_{X_i}(x)$ and 
$p_{\widetilde{X}_n\widetilde{Y}_n}(x,y)\eqdef p_{\widetilde{X}_n}(x)p_x(y)$. 
Note that $\lim_{n\rightarrow\infty}n\avgI{\smash{\widetilde{X}_n;\widetilde{Y}_n}}=\infty$ 
since we assumed that $\log M_n = \omega (\log n) $. Following~\cite{H14,Wang2016b}, one can also lower-bound the relative entropy as
\begin{align}
\myr{\delta_n}
\geq\D{\widehat{p}_{\mathbf{Z}}}{q_0^\pn}
\geq \sum_{i=1}^n \D{\widehat{p}_{Z_i}}{q_0}\geq n \mathbb{D}(p_{\widetilde{Z}_n}||q_0), 
\end{align}
where $\widetilde{Z}_n$ is distributed according to $p_{\widetilde{Z}_n}(z)\eqdef \frac{1}{n}\sum_{i=1}^n \widehat{p}_{Z_i}(z)$. Consequently,
\begin{align}
C_\text{no-CSI}&\leq \liminf_{n \to \infty} \frac{\mathbb{I}(\widetilde{X}_n;\widetilde{Y}_n)}{(1-\epsilon_n)\sqrt{ \D{\smash{p_{\widetilde{Z}_n}}}{q_0}}}\left(1+\frac{1}{n \mathbb{I}(\widetilde{X}_n;\widetilde{Y}_n)}\right) = \liminf_{n\rightarrow\infty }\frac{\mathbb{I}(\widetilde{X}_n;\widetilde{Y}_n)}{\sqrt{ \D{\smash{p_{\widetilde{Z}_n}}}{q_0}}} 
\label{eq:limit-converse}
\end{align}
where the sequence of distributions $\{p_{\widetilde{X}_n\widetilde{Y}_n\widetilde{Z}_n}\}_{n\geq 0}$ is subject to the constraint \myr{$\avgD{\smash{\modfirst{p}_{\widetilde{Z}_n}}}{q_0}\leq \frac{\delta_n}{n}$}. This completes the first step of the converse proof.

\subsubsection{Step two: discreteness of the optimal distribution}
%
 We define the optimization problem
\begin{align}
\label{eq:optimation-a-nu}
A(\nu) \eqdef \sup_{\mu\in \Omega: \D{\wz\circ \mu}{q_0} \leq \nu} I(\mu, \wy),
\end{align}
where $\Omega$ is the set of all probability measures over $\calX$ such as $\mu$ such that $\D{\wz\circ\mu}{q_0}< \infty$.
The next lemma shows that there exists a unique maximizer to the above problem.
\begin{lemma}
\label{lm:exist-unique}
 \modfirst{Let $\nu>0$.} There exists a unique probability measure $\mu^*_{\nu}\in\Omega$ such that  $\D{\wz\circ \mu^*_{\nu}}{q_0} \leq \nu$ and $ I(\mu^*_{\nu}, \wy) = A(\nu)$.
\end{lemma}
\begin{proof}
See Appendix~\ref{sec:optim-probl-eqref}.
\end{proof}
We next characterize the unconstrained form of the optimization in \eqref{eq:optimation-a-nu}.
\begin{theorem}
\label{th:kkt}
\modfirst{Let $\nu>0$.} There exists $\gamma(\nu)\geq 0$ such that the following holds.
\begin{enumerate}
\item We have 
\begin{align}
\label{eq:lagrange-opt}
A(\nu) = \modfirst{\max}_{\mu\in \Omega} \left[I(\mu, \wy) - \gamma(\nu)\pr{\D{\wz\circ \mu}{q_0} - \nu}\right],
\end{align}
and $\mu^*_{\nu}$ is the unique maximizer of  the above optimization.
\item\myr{ Define 
\begin{align}
w(x,  \mu_1, \nu) \eqdef \int_0^\infty p_x(y) \log \frac{p_x(y)}{(\wy \circ \mu_1)(y)} dy - \gamma(\nu)\pr{ \int_{0}^\infty q_x(z) \log \frac{(\wz\circ \mu_1)(z)}{q_0(z)}dz-\nu}.
\end{align}
\modfirst{For all $\mu \in \Omega$, we have
\begin{align}
A(\nu) \geq \E[\mu]{w(X, \mu_\nu^*, \nu)}\label{eq:integral-w}.
\end{align}}
\item Given $\mu_1\in\Omega$, we have for all $\mu \in \Omega$,
\begin{align}
A(\nu) \geq \E[\mu]{w(X, \mu_1, \nu)}.
\end{align}
if and only if
\begin{align}
w(x, \mu_1, \nu) &\leq A(\nu) \quad \forall x \in \calX, \label{eq:point-w}\\
w(x, \mu_1, \nu) &= A(\nu) \quad \forall x \in \supp{\mu_1} \label{eq:point-supp-w}.
\end{align}}
\item We have $\lim_{\nu\to 0^+} \gamma(\nu) = \infty$ and $\lim_{\nu\to0^+} \gamma(\nu)\nu = 0$.
\end{enumerate}
\end{theorem}

\begin{proof}
  See Appendix~\ref{sec:optim-probl-eqref}.
\end{proof}

\begin{lemma}
There exists $\nu_0>0$ such that for all $0<\nu \leq \nu_0$, $\supp{\mu^*_{\nu}}$ is discrete with a finite number of points in any bounded interval.
\end{lemma}

\begin{proof}
Fix some $\nu>0$, and define $\modfirst{r}(y)\eqdef (\wy\circ \mu^*_{\nu})(y)$ and $f(z) \eqdef (\wz\circ \mu^*_{\nu})(z)$. We assume that there exists an interval with an infinite number of points in $\supp{\mu^*_{\nu}}$ and obtain a contradiction for $\nu$ small enough in four steps.

\textbf{Step 1:} \textbf{We first use the argument in \cite{Abou-Faycal2001} to show that the KKT condition in \eqref{eq:point-supp-w} holds for all ${x\geq 0}$.}  By the Bolzano-Weierstrass theorem, there exists a convergent sequence $\{x_i\}_{i\geq 1}$ in $\supp{\mu^*_{\nu}}$.  Moreover, by \eqref{eq:point-supp-w}, for any $x\in\supp{\mu^*_{\nu}}$, we have
\begin{align}
 \phi_\nu(x)  \label{eq:phi-def}
&\eqdef w(x, \mu^*_{\nu}, \nu) - A(\nu)\\
&=\int_0^\infty p_x(y) \log \frac{p_x(y)}{\modfirst{r}(y)} dy - \gamma(\nu) \int_{0}^\infty q_x(z) \log \frac{f(z)}{q_0(z)}dz - A(\nu) + \gamma(\nu)\nu  = 0.
\end{align} 
\modfirst{We now show that $\phi_\nu(x)$ is analytic in $x$ over the domain $\calD \eqdef \{x: \calR(x) > 0\}$. Note that $\int_0^\infty p_x(y) \log p_x(y) dy = -\log(1+\theta_mx) - 1$ and $\int_0^\infty q_x(z) \log q_0(z) dz = -1 - x$, which are analytic over $\calD$.  We furthermore have 
\begin{align}
|p_x(y)| 
&= \frac{1}{|1+\theta_m^2 x|} \left|e^{-\frac{y}{1+\theta_m^2x}}\right|\\
&\stackrel{(a)}{=} \frac{1}{|1+\theta_m^2 x|} e^{-\frac{y\pr{\theta_m^2\calR(x) + 1}}{|1+\theta_m^2 x|^2}},\label{eq:pxy-abs}
\end{align}
where $(a)$ follows from $|e^{z} |= e^{\calR(z)}$. This implies that

\begin{align}
\int_0^\infty |p_x(y)\log r(y)| dy
&\stackrel{(a)}{\leq}  \int_0^\infty |p_x(y)|\pr{\theta_m^2 \E[\mu_\nu^*]{X} + y} dy\\
&\stackrel{(b)}{\leq}  \int_0^\infty |p_x(y)|\pr{\theta_m^2 \pr{2\sqrt{\nu} + \nu}+ y} dy\\
&\stackrel{(c)}{=} \theta_m^2 \pr{2\sqrt{\nu} + \nu} \frac{|1+\theta_m^2 x|}{\theta_m^2\calR(x) + 1 } + \frac{|1+\theta_m^2 x|^3}{(\theta_m^2\calR(x) + 1)^2 },
\end{align}
where $(a)$ follows from \eqref{eq:bound-wy}, $(b)$ follows from Lemma~\ref{lm:expected-x-bound}, and $(c)$ follows by \eqref{eq:pxy-abs}.  Therefore, $\int_0^\infty |p_x(y)\log r(y)| dy$ is uniformly bounded on any compact subset of $\calD$, and Theorem~\ref{th:analytic} yields that $\int_0^\infty |p_x(y)\log r(y)| dy$ is analytic over $\calD$. One can similarly argue that  $\int_{0}^\infty q_x(z) \log f(z) dx$ is also analytic over $\calD$ and therefore $\phi_\nu$ is analytic.} Since $\phi_\nu(x)$ is an analytic function over $\calD$, and $\phi_\nu(x)=0$ over a set with a limit point in $\calD$, the identity theorem~\cite{ablowitz2003complex} states that $\phi_\nu(x)=0$ for all $x\in \calD$. Thus, $\phi_\nu(x)=0$ holds over the entire real line. Using $\int_0^\infty p_x(y) \log p_x(y) dy = -\log(1+\theta_mx) - 1$ and $\int_0^\infty q_x(z) \log q_0(z) dz = -1 - x$, we can re-write
\begin{multline}
0 = \phi_\nu(x) =-\log(\theta_m^2x+1) - 1 -\gamma(\nu)(1+x)- A(\nu) + \gamma(\nu)\nu\\
-\int_0^\infty p_x(y) \log{\modfirst{r}(y)} dy - \gamma(\nu) \int_{0}^\infty q_x(z) \log {f(z)}dz.
\label{eq:phi-kkt}
\end{multline}
To obtain a contradiction, we cannot use the Laplace transform approach of \cite{Abou-Faycal2001} because there are two integrals in \eqref{eq:phi-kkt}, which is therefore the sum of two Laplace transforms with different arguments. Hence, we continue the proof with another approach. 

\textbf{Step 2:} \textbf{In this step, we shall find the supremum of the support of $\mathbf{\mu^*_{\nu}}$ in terms of $\mathbf{\gamma(\nu)}$.} We first consider any non-zero point $\widetilde{x}\in \supp{\mu^*_{\nu}}$ and any $\Delta\in]0;\widetilde{x}[$. Since $\widetilde{x}\in \supp{\mu^*_{\nu}}$, there exists $\delta > 0$ with $\mu^*_{\nu}\left(]\widetilde{x} - \Delta, \widetilde{x} + \Delta[\right) = \delta$. Thus, for any $y$, by definition of $\modfirst{r}(y)$ and the law of total probability, we lower-bound $\modfirst{r}(y)$ by
\begin{align}
\modfirst{r}(y) 
  &= \E[\mu^*_{\nu}]{\frac{1}{1+\theta_m^2 X}e^{-\frac{y}{1+\theta_m^2 X}}}\\
  &{\geq} \E[\mu^*_{\nu}]{\frac{1}{1+\theta_m^2 X}e^{-\frac{y}{1+\theta_m^2 X}}\bigg|X\in ]\widetilde{x} - \Delta, \widetilde{x} + \Delta[} \mu^*_{\nu}(]\widetilde{x} - \Delta, \widetilde{x} + \Delta[) \\
  &\geq \frac{\delta}{1+\theta_m^2(\widetilde{x} + \Delta)} e^{-\frac{y}{1+\theta_m^2(\widetilde{x} - \Delta)}},
\end{align}
and similarly, lower-bound $f(z)$ by
\begin{align}
f(z) \geq \frac{\delta}{1+\widetilde{x} + \Delta} e^{-\frac{z}{1+\widetilde{x} - \Delta}}.
\end{align}
Substituting these bounds in \eqref{eq:phi-kkt}, we obtain
\begin{align}
  \displaybreak[0]0 &\leq -\log(\theta_m^2x+1) - 1 -\gamma(\nu)(1+x)- A(\nu) + \gamma(\nu)\nu\nonumber \\
  \displaybreak[0]  &\phantom{========}-\int_0^\infty p_x(y) \log \frac{\delta}{1+\theta_m^2(\widetilde{x} + \Delta)} e^{-\frac{y}{1+\theta_m^2(\widetilde{x} - \Delta)}} dy - \gamma(\nu) \int_{0}^\infty q_x(z) \log\frac{\delta}{1+\widetilde{x} + \Delta} e^{-\frac{y}{1+\widetilde{x} - \Delta}}dz\\
    \displaybreak[0]&=  -\log(\theta_m^2x+1) - 1 -\gamma(\nu)(1+x)- A(\nu) + \gamma(\nu)\nu\nonumber \\
  \displaybreak[0]  & \phantom{========}-\log \frac{\delta}{1+\theta_m^2(\widetilde{x} + \Delta)} + \frac{1+ \theta_m^2 x}{1+ \theta_m^2(\widetilde{x} - \Delta)} - \gamma(\nu)\pr{ \log\frac{\delta}{1+\widetilde{x} + \Delta} - \frac{1 + x}{1 + \widetilde{x} - \Delta} }\\
  \displaybreak[0]  &= \kappa -\log(\theta_m^2x+1) - x\pr{\gamma(\nu)\frac{\widetilde{x} - \Delta}{1 + \widetilde{x} - \Delta} - \frac{\theta_m^2}{1 + \theta_m^2(\widetilde{x} - \Delta)}    }, \label{eq:kkt-not-discrete-upper}
\end{align}
where $\kappa$ is a constant not depending on $x$. Since \eqref{eq:kkt-not-discrete-upper} holds for all $x$,  by taking the limit $x\to\infty$, we should have
\begin{align}
\gamma(\nu)\frac{\widetilde{x} - \Delta}{1 + \widetilde{x} - \Delta} - \frac{\theta_m^2}{1 + \theta_m^2(\widetilde{x} - \Delta)} \leq 0.
\end{align}
Moreover, by  letting $\Delta$ tend to zero, we obtain
\begin{align}
\label{eq:supp-bound}
\gamma(\nu)\frac{\widetilde{x}}{1 + \widetilde{x}} - \frac{\theta_m^2}{1 + \theta_m^2\widetilde{x} } \leq 0,
\end{align}
which implies that $x^*\eqdef \sup( \supp{\mu^*_{\nu}}) < \infty$. Furthermore, upon finiteness of $x^*$, we have
\begin{align}
\modfirst{r}(y) \leq e^{-\frac{y}{1+\theta_m^2x^*}},
\end{align}
and
\begin{align}
f(z) \leq  e^{-\frac{z}{1+x^*}}.
\end{align}
Replacing these upper-bounds in \eqref{eq:phi-kkt}, we obtain
\begin{align}
0 &\geq -\log(\theta_m^2x+1) - 1 -\gamma(\nu)(1+x)- A(\nu) + \gamma(\nu)\nu- \int_0^\infty p_x(y) \log e^{-\frac{y}{1+\theta_m^2x^*}}dy - \gamma(\nu) \int_{0}^\infty q_x(z) \log e^{-\frac{z}{1+x^*}}dz\\
&=  -\log(\theta_m^2x+1) - 1 -\gamma(\nu)(1+x)- A(\nu) + \gamma(\nu)\nu + \frac{1 + \theta_m^2 x}{1+\theta_m^2x^*}+ \gamma(\nu) \frac{1 + x}{1+x^*}\\
&= \kappa' -\log(\theta_m^2x+1) - x \pr{\gamma(\nu)\frac{{x}^*}{1 + {x}^*} - \frac{\theta_m^2}{1 + \theta_m^2{x}^* }  },\label{eq:kkt-not-discrete-lower}
\end{align}
where $\kappa'$ is a constant not depending on $x$. Since \eqref{eq:kkt-not-discrete-lower} holds for all $x$, we have
\begin{align}
\gamma(\nu)\frac{{x}^*}{1 + {x}^*} - \frac{\theta_m^2}{1 + \theta_m^2{x}^* } \geq 0.
\end{align}
By definition of the support of a distribution, it should be closed, and therefore, $x^*\in\supp{\mu^*_{\nu}}$. Since \eqref{eq:supp-bound} holds for all points in the support, we can set $\widetilde{x} = x^*$ and obtain
\begin{align}
\label{eq:xstar-gamma}
\gamma(\nu)\frac{{x}^*}{1 + {x}^*} - \frac{\theta_m^2}{1 + \theta_m^2{x}^* } =0.
\end{align}

\textbf{Step 3:} \textbf{Using the equality for ${x^*}$ in \eqref{eq:xstar-gamma}, we derive an upper-bound on ${A(\nu)}$ depending on $\gamma(\nu)$ and $\nu$.} By definition of $\mu^*_{\nu}$, it holds that
\begin{align}
A(\nu) 
\displaybreak[0]&= I(\mu^*_{\nu}, \wy)\\
\displaybreak[0]&= \E[\wy \times \mu^*_{\nu}]{\log \frac{p_X(Y)}{\modfirst{r}(Y)}}\\
\displaybreak[0]&= \E[\wy \times \mu^*_{\nu}]{\log \frac{p_X(Y) p_0(Y)}{\modfirst{r}(Y)p_0(Y)}}\\
\displaybreak[0]&= \E[\wy \times  \mu^*_{\nu}]{\log \frac{p_X(Y)}{p_0(Y)}} - \E[\wy \circ  \mu^*_{\nu}]{\log \frac{\modfirst{r}(Y)}{p_0(Y)}}\\
\displaybreak[0]&= \E[\wy \times  \mu^*_{\nu}]{\log \frac{p_X(Y)}{p_0(Y)}}  - \D{\modfirst{r}}{p_0}\\
\displaybreak[0]&\leq \E[\wy \times  \mu^*_{\nu}]{\log \frac{p_X(Y)}{p_0(Y)}}\\
\displaybreak[0]&= \E[\mu^*_{\nu}]{\theta_m^2X - \log(1 + \theta_m^2 X)}\\
\displaybreak[0]&\stackrel{(a)}{\leq}\E[\mu^*_{\nu}]{\frac{1}{2}\theta_m^\modfirst{4} X^2}\\
\displaybreak[0]&\leq \frac{1}{2}\theta_m^\modfirst{4} x^* \E{X}\\
\displaybreak[0]&\stackrel{(b)}{\leq}\frac{1}{2}\theta_m^\modfirst{4} x^* \pr{2\sqrt{\nu} + \nu},
\end{align}
where $(a)$ follows from $\log(1+x) \geq x - x^2/2$ for $x\geq 0$, and $(b)$ follows from Lemma~\ref{lm:expected-x-bound}. Therefore, we can use \eqref{eq:xstar-gamma} to obtain
\begin{align}
A(\nu) 
&\leq \frac{1}{2}\theta_m^\modfirst{4} \pr{\frac{\theta_m^\modfirst{4}(1+x^*)}{\gamma(\nu)(1+\theta_m^\modfirst{4} x^*)}}\pr{2\sqrt{\nu} + \nu}\\
&\leq \frac{2\sqrt{\nu} + \nu}{\gamma(\nu)}\pr{\frac{1}{2}\theta_m^\modfirst{4}(1+|1-\theta_m^\modfirst{4}|)}.
\end{align}

\textbf{Step 4:} \textbf{We complete the proof by obtaining a contradiction.} \modfirst{Lemma~\ref{lm:a-prop} part 4 implies that there exists $\nu_0 > 0$ and $C>0$ such that $A(\nu) \geq C\sqrt{\nu}$ for all $0<\nu\leq \nu_0$. By Theorem~\ref{th:kkt} part 4, we can choose $\nu_0$ small such that   $\gamma(\nu) > \frac{3}{C}\pr{\frac{1}{2}\theta_m^\modfirst{4}(1+|1-\theta_m^\modfirst{4}|)}$  in addition to $A(\nu) \geq C\sqrt{\nu}$ for all $0<\nu \leq \nu_0$. Since by decreasing $\nu_0$, the statement would be weaker, we can always assume that $\nu_0 < 1$.} Thus,
\begin{align}
C\sqrt{\nu} 
&\leq A(\nu)\\
&\leq \frac{2\sqrt{\nu} + \nu}{\gamma(\nu)}\pr{\frac{1}{2}\theta_m^\modfirst{4}(1+|1-\theta_m^\modfirst{4}|)} \\
&< \frac{2\sqrt{\nu} + \nu}{\frac{3}{C}\pr{\frac{1}{2}\theta_m^\modfirst{4}(1+|1-\theta_m^\modfirst{4}|)}}\pr{\frac{1}{2}\theta_m^\modfirst{4}(1+|1-\theta_m^\modfirst{4}|)}\\
&\leq C \sqrt{\nu}.
\end{align}

\end{proof}

\begin{lemma}
  \label{lm:finite-mass-points}
There exists $\nu_0 > 0$ such that for any $\nu_0>\nu > 0$, the support of $\mu^*_{\nu}$ has a finite number of points. 
\end{lemma}

\begin{proof}
We proceed by contradiction. Assume that the support of $\mu^*_{\nu}$ has infinitely many points $\{x_i\}_{i=1}^{\infty}$ in increasing order with probabilities $\{\alpha_i\}_{i=1}^{\infty}$. Since we proved that in any bounded interval, we can only have a finite number of points, $\lim_{i\to\infty}x_i=\infty$. Note that for any $j\geq 1$, we have 
\begin{align}
(\wy\circ \mu^*_{\nu})(y)
& =\sum_{i=1}^{\infty}\alpha_i p_{x_i}(y)\\
&\geq  \alpha_j p_{x_j}(y),
\end{align}
and
\begin{align}
(\wz\circ \mu^*_{\nu})\pr{z} \geq \alpha_j q_{x_j}\pr{z}.
\end{align}
Therefore, for all $j\geq 1$, we can upper-bound  $\phi_{\nu}(x)$ defined in \eqref{eq:phi-def} as \modfirst{
\begin{align}
\phi_{\nu}(x) 
&= \int_0^\infty p_x(y) \log \frac{p_x(y)}{(\wy \circ \mu_\nu^*)(y)} dy - \gamma(\nu) \int_{0}^\infty q_x(z) \log \frac{(\wz\circ \mu_\nu^*)(z)}{q_0(z)}dz - A(\nu) + \gamma(\nu)\nu\\
&\leq \int_0^\infty p_x(y) \log \frac{p_x(y)}{ \alpha_j p_{x_j}(y)} dy - \gamma(\nu) \int_{0}^\infty q_x(z) \log \frac{\alpha_j q_{x_j}(z)}{q_0(z)}dz - A(\nu) + \gamma(\nu)\nu\\
&= \log(\theta_m^2x +1) -1 - \log \frac{ \alpha_j}{1+\theta_m^2x_j} + \frac{1+\theta_m^2x}{1+\theta_m^2x_j} - \gamma(\nu)\pr{1+x  - \log \frac{ \alpha_j}{1+x_j} + \frac{1+x}{1+x_j} }-A(\nu) + \gamma(\nu)\nu \\
&= \kappa + \log(\theta_m^2x+1) + \pr{-\gamma(\nu) +\frac{\gamma(\nu)}{1+x_j} +\frac{\theta_m^2}{1+\theta_m^2x_j} }x,\label{eq:kkt-finite}
\end{align}}
where $\kappa$ is a constant not depending on $x$. Furthermore, the KKT condition in \eqref{eq:point-supp-w} implies that \eqref{eq:kkt-finite} is non-negative for all $x_i$, and since $x_i$ can be large enough, we should have
\begin{align}
-\gamma(\nu) +\frac{\gamma(\nu)}{1+x_j} +\frac{\theta_m^2}{1+\theta_m^2x_j} \geq 0.
\end{align}
Because $x_j$ can be large enough, we have $-\gamma(\nu) \geq 0$. This cannot be true for small $\nu$ since $\lim_{\nu\to0^+} \gamma(\nu)= \infty$ by Theorem~\ref{th:kkt}.
\end{proof}

\begin{lemma}
\label{lm:mass-point-zero}
There exists $\nu_0>0$ such that  for all $\nu_0>\nu > 0$, $\mu^*_{\nu}$ has a mass point at 0.
\end{lemma}
 The proof of Lemma~\ref{lm:mass-point-zero} will require the following technical result which is a modification of \cite[Lemma 1]{Abou-Faycal2001}.
 \begin{lemma}\label{prop:croissance}
 Let $f(z)$ be a \ac{PDF} with mean $m$ and $g(z)$ be a strictly monotonically increasing function, then $\int(z-m)f(z)g(z)dz>0$.
 \end{lemma}

 \begin{proof}
  $(z-m)(g(z)-g(m))$ is always positive as either the product of two negative terms if $z<m$ or two 
  positive terms if $z>m$. Thus, $(z-m)g(z)>(z-m)g(m)$ and $\int (z-m)g(z)f(z)dz>\int(z-m)g(m)f(z)dz=0.$
 \end{proof}

\begin{proof}[Proof of Lemma~\ref{lm:mass-point-zero}] 
\modfirst{Let $\nu_0$ be as in Lemma~\ref{lm:finite-mass-points} so that $\mu_\nu^*$ has finite number of mass points for all $0< \nu \leq \nu_0$. For the sake of a contradiction, assume that $\mu_\nu^*$ is a  discrete probability measure over $\calX$ with $k$ mass points $0 < x_1 < \cdots < x_k$ with corresponding  probabilities $\alpha_1, \cdots ,\alpha_k$.} In \cite{Abou-Faycal2001}, it is proved that reducing $x_1$ increases the mutual information $I(\mu, \wy)$. Therefore, to complete the proof, it is enough to show that $\frac{\partial  \D{\wz\circ\mu}{q_0}}{\partial x_1}>0$. Defining $f(x_1, z) \eqdef (\wz \circ \mu)(z) \log \frac{(\wz \circ \mu)(z) }{q_0(z)}$, we have 
\begin{align}
\frac{\partial \D{\wz\circ\mu}{q_0}}{\partial x_1} = \frac{\partial}{\partial x_1}\int_{\calZ} f(x_1, z) dz.
\end{align}
By Lemma~\ref{lm:cross-entropy}, $\int_{\calZ}|f(x_1, z)|dz < \infty$, and we have
\begin{align}
\frac{\partial f}{\partial x_1} (x_1, z) = \frac{\alpha_1}{(1+x_1)^2} q_{x_1}(z) \pr{z-\E[q_{x_1}]{Z}} \pr{\log \frac{(\wz\circ\mu)(z)}{q_0(z)} + 1},
\end{align}
which satisfies that
\begin{align}
\left| \frac{\partial f}{\partial x_1} (x_1, z)  \right| \leq e^{-\frac{z}{1+x_1}} (z + x_1 + 1)\pr{2z + \E[\mu]{X} + 1}.\label{eq:bound-f-z-x1}
\end{align}
The right hand side of \eqref{eq:bound-f-z-x1}, is bounded with an integrable function of $z$ independent of $x_1$, if $x_1$ is bounded. Hence, Theorem~\ref{th:leibniz} implies that
 \begin{align}
 \frac{\partial \D{\wz\circ \mu}{q_0}}{\partial x_1}=
 \alpha_1 \frac{1}{(1+ x_1)^2}\int_0^\infty (z-\E[q_{x_1}]{Z})q_{x_1}(z)\pr{\log \frac{(\wz\circ \mu)(z)}{q_0(z)} + 1} dz.\label{eq:eq:derv-d-x}
 \end{align} 
 Note that 
 \begin{align}
  \log \frac{(\wz\circ \mu)(z)}{q_0(z)}
  &=\log \frac{\sum_{i=1}^k \alpha_i \frac{1}{x_i+1}e^{-\frac{z}{x_i+1}}}{e^{-z}}\\
  &=\log \sum_{i=1}^k \alpha_i\frac{1}{x_i+1}e^{z\frac{x_i}{x_i +1}}.
 \end{align}
Since $1>\frac{1}{x_1+1}>\cdots>\frac{1}{x_k+1}$, 
 $\log \frac{(\wz\circ \mu)(z)}{q_0(z)} + 1$ is strictly monotonically increasing in $z$. \modfirst{Using Lemma~\ref{prop:croissance}, $\frac{\partial  \D{\wz\circ\mu}{q_0}}{\partial x_1}>0$, and hence, by decreasing $x_1$, the constraint $ \D{\wz\circ\mu}{q_0} \leq \nu$ still holds and $I(\mu, \wy)$ is increased. This contradicts with the definition of $\mu_\nu^*$, and therefore, there exists a mass point at zero.}
\end{proof}

\subsubsection{Step three: an amplitude constraint}

For a probability measure $\mu$ on $\calX$ and $a>0$, we define $\mathbb{C}_a[\mu]$ as a new probability measure $\widetilde{\mu}$ on $\calX$ such that
\begin{align}
\widetilde{\mu}(]-\infty, x[) =\begin{cases} {\mu}(]-\infty, x[)\quad & x < a,\\ 1\quad & x \geq a.\end{cases}
\end{align}
Intuitively, $\widetilde{\mu}$ is obtained by moving all probability of $]a, \infty[$ in $\mu$ to a mass point at $a$. 
\begin{theorem}
\label{th:bounded-supp}
Let $\{\nu_n\}_{n\geq1}$ be $o(1)$. For all $a > 1$, if $n$ is large enough, we have $\mathbb{C}_a[\mu^*_{\nu_n}] \in \Omega_a\pr{\nu_n}$ and
\begin{align}
\liminf_{n\to\infty}\frac{I(\mu^*_{\nu_n}, \wy)}{\sqrt{\D{\wz\circ\mu^*_{\nu_n}}{q_0}}} \leq \liminf_{n\to\infty}\frac{I(\mathbb{C}_a[\mu^*_{\nu_n}], \wy)}{\sqrt{\D{\wz\circ\mathbb{C}_a[\mu^*_{\nu_n}] }{q_0}}}.
\end{align}
\end{theorem}
To prove this result, we need the following lemmas.
\begin{lemma}
\label{lm:ca-i-d}
If $\mu$ is a discrete probability measure on $\calX$ with finite number of mass points $x_1<\cdots<x_k$ and corresponding probabilities $\alpha_1, \cdots, \alpha_k$, then
\begin{align}
\D{\wz\circ\mathbb{C}_a[\mu] }{q_0} &\leq \D{W_{Z|X}\circ \mu}{q_0}, \\
I(\mathbb{C}_a[\mu], \wy) &\geq I(\mu, \wy) - \theta_m^2\modfirst{\max}(\supp{\mu}	) \mu(]a, \infty[).
\end{align}
\end{lemma}
\begin{proof}
Similar to~\eqref{eq:eq:derv-d-x}, for all $i\in\intseq{1}{k}$, we have
\begin{align}
\frac{\partial }{\partial x_i}\D{\wz\circ \mu}{q_0} =  \alpha_i \frac{1}{(1+ x_i)^2}\int_0^\infty (z-\E[q_{x_i}]{Z})q_{x_i}(z)\log \frac{(\wz\circ \mu)(z)}{q_0(z)}dz\geq 0.
\end{align}
Hence, by moving all mass points located in $]a, \infty[$ to $a$ to obtain $\mathbb{C}_a[\modfirst{\mu}]$, we decrease the relative entropy. Applying the same argument to the channel $\wy$, we  have $\D{\wy\circ \mathbb{C}_a[\mu]}{p_0} \leq \D{\wy\circ \mu}{p_0}$. Additionally, we have
\begin{align}
I(\mu, \wy) 
&= \sum_{i=1}^k \alpha_i \D{p_{x_i}}{p_0} - \D{\wy\circ \mu}{p_0}\\
&=  \sum_{i=1}^k \alpha_i \pr{\theta_m^2x_i - \log(1+\theta_m^2x_i)}- \D{\wy\circ \mu}{p_0},
\end{align}
which implies that
\begin{align}
&I(\mu, \wy)  - I(\mathbb{C}_a[\mu], \wy)\\
&= \pr{\sum_{i:x_i > a}\alpha_i \pr{\theta_m^2x_i  -  \log(1+\theta_m^2x_i) }  - \mu(]a, \infty[) \pr{\theta_m^2a  -  \log(1+\theta_m^2a) }}  +\nonumber\\
&\phantom{==}\pr{ -  \D{\wy\circ \mu}{p_0} + \D{\wy\circ \mathbb{C}_a[\mu]}{p_0}}\\
&\leq \sum_{i:x_i > a}\alpha_i \pr{\theta_m^2x_i  -  \log(1+\theta_m^2x_i) }  - \mu(]a, \infty[) \pr{\theta_m^2a  -  \log(1+\theta_m^2a) } \\
&\leq \sum_{i:x_i > a}\alpha_i \pr{\theta_m^2x_i  -  \log(1+\theta_m^2x_i) }\\
&\leq \theta_m^2\modfirst{\max}(\supp{\mu}) \mu(]a, \infty[).
\end{align}
\end{proof}
\begin{lemma}
\label{lm:a-prob-bound-exp}
For all $a>0$, there exist $\nu_0 > 0$, $\widetilde{x}\in\calX$, and $\xi>0$ such that for all $0<\nu\leq \nu_0$, if $\modfirst{\max}( \supp{\mu_v^*}) \geq \widetilde{x}$, then $\mu_v^*(]a, \infty[) \leq 2^{-\xi \modfirst{\max}(\supp{\mu_v^*}) }$.
\end{lemma}
\begin{proof}
Fix $\nu > 0$ small enough and suppose that $\mu \eqdef \mu^*_{\nu}$ has mass points $x_1 < \cdots < x_k$ with corresponding probabilities  $\alpha_1, \cdots, \alpha_k$. Let $\modfirst{r}(y) \eqdef (\wy\circ \mu)(z)$ and $f(z) \eqdef (\wz\circ\mu)(z)$. Substituting the lower-bounds
\begin{align}
\modfirst{r}(y) \geq \frac{\mu(]a, \infty[)}{1 + \theta_m^2x_k}e^{-\frac{y}{1 + \theta_m^2a}}, \text{ and } f(z) \geq \frac{\mu(]a, \infty[)}{1 + x_k}e^{-\frac{z}{1 + a}},
\end{align}
in the KKT condition \eqref{eq:point-supp-w} for the point $x = x_k$, we obtain
\begin{multline}
\label{eq:kkt-bound-last}
0\leq -\log(\theta_m^2x_k+1) - 1 -\gamma(\nu)(1+x_k)- A(\nu) + \gamma(\nu)\nu- \\
\log  \frac{\mu(]a, \infty[)}{1 + \theta_m^2x_k} + \frac{1+\theta_m^2 x_k}{1 + \theta_m^2 a}+ \gamma(\nu) \pr{- \log  \frac{\mu(]a, \infty[)}{1 + x_k}  + \frac{1+ x_k}{1 +  a}  } .
\end{multline}
Since $\lim_{\nu\to0^+}\gamma(\nu)\nu = 0$, for small $\nu$, $-1 - A(\nu) + \gamma(\nu)\nu \leq 0$, and therefore, \eqref{eq:kkt-bound-last}  implies that
\begin{align}
0 
&\leq-\gamma(\nu)(1+x_k) \frac{a}{1+a} + \gamma(\nu) \log(1+x_k) + \frac{1+\theta_m^2 x_k}{1 + \theta_m^2 a} -  (1+\gamma(\nu)) \log ( \mu(]a, \infty[)).
\end{align}
Furthermore, if $x_k$  is large enough, we have $\log(1+x_k) \leq\frac{ (1+x_k) a}{4(1+a)}$, and if $\nu$ is small enough and $x_k$ is large enough, by Theorem~\ref{th:kkt} part 4,  we have $ \frac{1+\theta_m^2 x_k}{1 + \theta_m^2 a} \leq \gamma(\nu)(1+x_k) \frac{a}{4(1+a)} $. Hence, there exist $\nu_0>0$ and $\widetilde{x}>0$ such that if $\nu\leq \nu_0$ and $x_k\geq \widetilde{x}$, we have 
\begin{align}
0\leq -\frac{1}{2}\gamma(\nu)(1+x_k) \frac{a}{1+a} - (1+\gamma(\nu)) \log ( \mu(]a, \infty[)),
\end{align}
which yields that 
\begin{align}
\mu(]a, \infty[) \leq \exp\pr{-\frac{1}{2}\frac{\gamma(\nu)}{1+\gamma(\nu)} (1+x_k)\frac{a}{1+a}}.
\end{align}
\modfirst{Since $\lim_{\nu\to 0^+} \gamma(\nu) = \infty$, there exists $\nu_0>0$ such that $\inf_{\nu\in ]0, \nu_0]} \frac{\gamma(\nu)}{1+\gamma(\nu)} \geq \frac{1}{2}$. Hence, for $\xi \eqdef \frac{a}{4(1+a)}$ and all $0<\nu<\nu_0$, we have $\mu(]a, \infty[) \leq 2^{-\xi x_k}$.}
\end{proof}

We are now ready to establish the upper bound in~\eqref{eq:bounds} of Theorem~\ref{th:bounded-supp}.
\begin{proof}[Proof of Theorem~\ref{th:bounded-supp}]
Let $x^*_n \eqdef \modfirst{\max}( \supp{\mu^*_{\nu_n}})$. By Lemma~\ref{lm:finite-mass-points}, if $n$ is large enough $\mu^*_{\nu_n}$ is a discrete probability measure with finite number of mass points, and so is $\mathbb{C}_a[\mu^*_{\nu_n}]$. By Lemma~\ref{lm:ca-i-d}, we have $\D{\wz\circ\mathbb{C}_a[\mu^*_{\nu_n}]}{q_0}\leq \D{\mu^*_{\nu_n}}{q_0}= \nu_n$, and 
\begin{align}
\frac{I(\mathbb{C}_a[\mu^*_{\nu_n}], \wy)}{\sqrt{\D{\wz\circ\mathbb{C}_a[\mu^*_{\nu_n}] }{q_0}}}
&\geq \frac{I(\mu^*_{\nu_n}, \wy) -\theta_m^2x^*_n\mu^*_{\nu_n}(]a, \infty[) }{\sqrt{\D{\wz\circ\mu^*_{\nu_n}}{q_0}}}.
\end{align}
Therefore, it is enough to show that
\begin{align}
\label{eq:o-x-p}
x^*_n \mu^*_{\nu_n}(]a, \infty[)  = o\pr{\sqrt{\D{\wz\circ\mu^*_{\nu_n}}{q_0}}} = o\pr{\sqrt{\nu_n}}.
\end{align}
To do so, we consider $\nu_0$, $\widetilde{x}$, and $\xi$  from Lemma~\ref{lm:a-prob-bound-exp}. For $n$ large enough such that $ \frac{2}{\xi} \log\frac{1}{\nu_n}> \widetilde{x}$, if $x^*_n\geq \frac{2}{\xi} \log\frac{1}{\nu_n}$, then
\begin{align}
x^*_n\mu^*_{\nu_n}(]a, \infty[)  \leq x^*_n2^{- \xi x^*_n },
\end{align}
which is less than $2^{-\frac{1}{2} \xi x^*_n }$ for large enough $n$. Thus, $x^*_n\geq \frac{2}{\xi} \log\frac{1}{\nu_n}$ implies that $x^*_n\mu^*_{\nu_n}(]a, \infty[) \leq \frac{1}{\nu_n}$. For  the other case when $x^*_n < \frac{2}{\xi} \log\frac{1}{\nu_n}$, let $\widetilde{\mu}$ be a probability distribution on $\calX$ with two mass points at $0$ and $a$ with probabilities $1-\mu^*_{\nu_n}(]a, \infty[)$ and $\mu^*_{\nu_n}(]a, \infty[)$, respectively. Then, we have
\begin{align}
\nu_n = \D{\wz\circ\mu^*_{\nu_n}}{q_0} \stackrel{(a)}{\geq} \D{\wz\circ \widetilde{\mu}}{q_0} \stackrel{(b)}{\geq} K \pr{\mu^*_{\nu_n}(]a, \infty[)}^{\frac{a+1}{a}},
\end{align}
where $(a)$ follows from the same argument as in the proof of Lemma~\ref{lm:ca-i-d}, and $(b)$ follows from Lemma~\ref{lm:d-two-points} for a constant $K$ depending on $a$. Therefore, we have
\begin{align}
x_n^* \mu^*_{\nu_n}(]a, \infty[) \leq \frac{2}{\xi} \log\frac{1}{\nu_n} \pr{\frac{\nu_n}{K}}^{\frac{a}{a+1}}.
\end{align}
Since both $\frac{1}{\nu_n}$ and $\frac{2}{\xi} \log\frac{1}{\nu_n} \pr{\frac{\nu_n}{K}}^{\frac{a}{a+1}}$ are $o(\sqrt{\nu_n})$, we have \eqref{eq:o-x-p}.
\end{proof}

\subsubsection{Step four: obtaining the bound in Theorem~\ref{th:general-capacity-non-coherent}}\modfirst{
We first prove  a lemma that relates the constraint on the relative entropy to $\chi_2$ divergence. Let $\td{\Omega}^{\geq 0}$  be the set of discrete probability measures over $[0, 1[$ with finite number of mass points.
\begin{lemma}
Let $\epsilon>0$ be small enough and $\{\nu_n\}_{n \geq }$ be a sequence of real numbers such that $\lim_{n\to \infty} \nu_n = 0$ and $2\sqrt{\nu_n} + \nu_n \leq 0.5$ for all $n$.  There exists a sequence of probability measures $\{\lambda_n\}_{n\geq 1}$ such that $\lambda_n \in \td{\Omega}^{\geq 0}$ and
\begin{align}
\label{eq:A-bound}
\limsup_{n\to \infty} \frac{A(\nu_n)}{\sqrt{\nu_n}} \leq \limsup_{n\to \infty} \frac{I(\lambda_n, \wy)}{\sqrt{\frac{1}{2}\chi_2(\wz\circ \lambda_n\|q_0)}} + \epsilon.
\end{align}
\end{lemma}
\begin{proof}

Let $\xi > 0$ and $\zeta \eqdef \frac{6\xi}{1-6\xi}$. Define
\begin{align}
\mu_n &\eqdef \mu_{\nu_n}^*\\
\mu_n' &\eqdef C_{1+\zeta}[\mu_n]\\
\mu_n'' &\eqdef C_{a_n}[\mu_n'],
\end{align}
where $a_n' \eqdef \inf_{a: \mu_n(]a, \infty[) \leq \nu_n^{\frac{1}{2} + \xi} } a$ and $a_n \eqdef \min(1-\zeta, a_n')$. Let $\{\nu_n\}_{n \geq }$ be a sequence of real numbers such that $\lim_{n\to \infty} \nu_n = 0$ and $2\sqrt{\nu_n} + \nu_n \leq 0.5$ for all $n$.  By construction, we have $\mu_n''([a_n', \infty[) \geq \nu_n^{\frac{1}{2} + \xi}$ and $\mu_n''(]a_n', \infty[) \leq \nu_n^{\frac{1}{2} + \xi}$. We next use the following lemma to upper-bound $\chi_2(\wz\circ \mu_n''\|q_0)$.
\begin{lemma}
\label{lm:chi-2-bound}
Let $\mu \in \td{\Omega}^{\geq 0}$ such that $\D{ \wz \circ \mu}{q_0} \leq \nu$ and $\max\pr{\supp{\mu}} \leq a < 1$. If  $2\sqrt{\nu} + \nu < 1/2$  and for some $M>0$, we have $(\wz\circ \mu)(M)/q_0(M) \geq e$, then
\begin{align}
 \frac{1}{2}\chi_2(\wz \circ \mu \|q_0) \leq \D{\wz \circ \mu}{q_0}  + \frac{1}{2}  (\E[\mu]{X})^3 \int_0^M e^{z\pr{-1+\frac{3a}{1+a}}}dz + \frac{1}{2} (\E[\mu]{X})^2 \int_{M}^\infty e^{z\pr{-1 + \frac{2a}{1+a}}}dz + 2 (\E[\mu]{X})^3.
\end{align}
\end{lemma}
\begin{proof}
See Appendix~\ref{sec:chi-2-bound}.
\end{proof}
We first establish a lower-bound on $(\wz \circ \mu_n''(z))/q_0(z)$ to use Lemma~\ref{lm:chi-2-bound}. Since $a_n \leq a_n'$, we have
\begin{align}
\mu_n''([a_n, \infty[)   \geq \mu_n''([a_n', \infty[)  \geq \nu_n^{\frac{1}{2} + \xi},
\end{align}
which yields that 
\begin{align}
\frac{(\wz \circ \mu_n'')(z)}{q_0(z)}\geq \nu_n^{\frac{1}{2} + \xi} \frac{ e^{\frac{a_n}{1+a_n} z} }{1+a_n}.
\end{align}
Choosing $M_n = \frac{1+a_n}{a_n} \pr{2 + \pr{\frac{1}{2} + \zeta} \log \frac{1}{\nu_n} }$, we have $(\wz \circ \mu_n'')(M_n)/q_0(M_n)\geq e$. Therefore, Lemma~\ref{lm:chi-2-bound} implies that
\begin{align}
&\frac{1}{2}\chi_2(\wz\circ \mu_n''\|q_0)\\ 
&\leq \D{\wz \circ \mu_n''}{q_0} + \frac{1}{2}  (\E[\mu_n'']{X})^3 \int_0^{M_n}e^{z\pr{-1+\frac{3a_n}{1+a_n}}}dz  +  \frac{1}{2} (\E[\mu_n'']{X})^2\int_{M_n}^\infty e^{z\pr{-1+\frac{2a_n}{1+a_n}}}dz  + 2 (\E[\mu_n'']{X})^3 \\
&\stackrel{(a)}{\leq} \nu_n + \frac{1}{2}  (\E[\mu_n'']{X})^3 \int_0^{M_n}e^{z\pr{-1+\frac{3a_n}{1+a_n}}}dz  +  \frac{1}{2} (\E[\mu_n'']{X})^2\int_{M_n}^\infty e^{z\pr{-1+\frac{2a_n}{1+a_n}}}dz  + 2 (\E[\mu_n'']{X})^3 \\
&\stackrel{(b)}{\leq} \nu_n\pr{1 + \frac{27}{2} \nu_n^{\frac{1}{2}} \int_0^{M_n}e^{z\pr{-1+\frac{3a_n}{1+a_n}}}dz  +  \frac{9}{2}\int_{M_n}^\infty e^{z\pr{-1+\frac{2a_n}{1+a_n}}}dz + 27\nu_n^{\frac{1}{2}}   },\label{eq:chi2-bound}
\end{align}
where $(a)$ follows since by Lemma~\ref{lm:ca-i-d}
\begin{align}
\D{\wz \circ \mu_n''}{q_0} \leq \D{\wz \circ \mu_n'}{q_0} \leq  \D{\wz \circ \mu_n}{q_0}  \leq \nu_n,
\end{align}
and $(b)$ follows since by Lemma~\ref{lm:expected-x-bound}, we have $\E[\mu_n'']{X}\leq 2\sqrt{\nu_n}+ \nu_n \leq 3\sqrt{\nu_n}$. 
We now show that
\begin{align}
\label{eq:limits}
\lim_{n\to \infty } \frac{27}{2} \nu_n^{\frac{1}{2}} \int_0^{M_n}e^{z\pr{-1+\frac{3a_n}{1+a_n}}}dz = \lim_{n\to \infty}  \frac{9}{2} \int_{M_n}^\infty e^{z\pr{-1+\frac{2a_n}{1+a_n}}}dz = \lim_{n\to \infty} 27M_0\nu_n^{\frac{1}{2}} = 0.
\end{align}
For  the first limit, we consider two cases. 

If $a_n \leq 1/4$, then $-1 + \frac{3a_n}{1+a_n} \leq -2/5$ and
\begin{align}
 \frac{27}{2} \nu_n^{\frac{1}{2}}\int_0^{M_n}e^{z\pr{-1 + \frac{3a_n}{1+a_n}}}dz
 & \leq  \frac{27}{2} \nu_n^{\frac{1}{2}}\int_0^{\infty}e^{-\frac{2z}{5}}dz\\
 &= \frac{135}{4} \nu_n^{\frac{1}{2}}.
\end{align}
If $a_n\geq 1/4$, then
\begin{align}
\int_0^{M_n}e^{z\pr{-1 + \frac{3a_n}{1+a_n}}}dz 
&\leq  M_n \max\pr{1, e^{M_n \pr{-1 + \frac{3a_n}{1+a_n}} }}\\
&=  \frac{1+a_n}{a_n} \pr{2 + \pr{\frac{1}{2} + \xi} \log \frac{1}{\nu_n} } \max\pr{1, e^{ \frac{1+a_n}{a_n} \pr{2 + \pr{\frac{1}{2} + \xi} \log \frac{1}{\nu_n} } \pr{-1 + \frac{3a_n}{1+a_n}} }}\\
&\leq  5 \pr{2 + \pr{\frac{1}{2} + \xi} \log \frac{1}{\nu_n} } \max\pr{1, e^{ \frac{2a_n - 1}{a_n}\pr{2 + \pr{\frac{1}{2} + \xi} \log \frac{1}{\nu_n} }  }}.
\end{align}
Note that 
\begin{align}
5 \pr{2 + \pr{\frac{1}{2} + \xi} \log \frac{1}{\nu_n} }  = O\pr{\log \frac{1}{\nu_n}}
\end{align}
and
\begin{align}
\max\pr{1, e^{ \frac{2a_n - 1}{a_n}\pr{2 + \pr{\frac{1}{2} + \xi} \log \frac{1}{\nu_n} }  }} = O\pr{1 + \nu_n^{-\frac{2a_n - 1}{a_n}\pr{\frac{1}{2} + \xi}}}
\end{align}
For $a_n \leq 1 - \frac{6\xi}{1-6\xi}$, we have $\frac{2a_n - 1}{a_n}\pr{\frac{1}{2} + \xi} \leq \frac{1}{2} - \xi$.
For the second limit in \eqref{eq:limits}, note that
\begin{align}
 \int_{M_n}^\infty e^{z\pr{-1+\frac{2a_n}{1+a_n}}}dz\leq 
  \int_{M_n}^\infty e^{-z\frac{\zeta}{2-\zeta}}dz.
\end{align}
Because $\int_0^\infty e^{-z\frac{\xi}{2-\xi}} < \infty$ and  $M_n =  \frac{1+a_n}{a_n} \pr{2 + \pr{\frac{1}{2} + \zeta} \log \frac{1}{\nu_n} }$ goes to infinity as $n$ goes to infinity, $\lim_{n\to \infty}  \int_{M_n}^\infty e^{z\pr{-1+\frac{2a_n}{1+a_n}}}dz = 0$. The third limit in \eqref{eq:limits} follows since $\lim_{n\to \infty} \nu_n = 0$. We thus obtain \eqref{eq:limits}, which together with  \eqref{eq:chi2-bound}  results in
\begin{align}
\label{eq:chi-limit}
\limsup_{n\to \infty} \frac{\frac{1}{2}\chi_2(\wz\circ \mu_n''\|q_0) }{\nu_n} \leq 1.
\end{align}
We now consider $I(\mu_n'', \wy)$ and show that it is close to $I(\mu_n, \wy) = A(\nu_n)$.

If $a_n = 1 - \zeta$, then by a modification of Lemma~\ref{lm:ca-i-d}
\begin{align}
I(\mu_n '', \wz) 
&\geq I(\mu_n', \wz) - 2 \zeta \mu_n'(]1-\zeta, \infty [)\\
&=I(\mu_n', \wz) - 2 \zeta \mu_n(]1-\zeta, \infty [)\\
&\geq I(\mu_n', \wz) - 2 \zeta \frac{\E[\mu_n]{X}}{1-\zeta}\\
&\geq I(\mu_n', \wz) - 6 \zeta \frac{\sqrt{\nu_n}}{1-\zeta}.
\end{align}
If $a_n = a_n'$, by Lemma~\ref{lm:ca-i-d}
\begin{align}
I(\mu_n '', \wz) 
&\geq I(\mu_n', \wz) - 2 (1+\zeta) \mu_n'(]a_n', \infty [)\\
&=I(\mu_n', \wz) - 2(1+\zeta) \nu_n^{\frac{1}{2} + \zeta}.
\end{align}
Therefore, 
\begin{align}
I(\mu_n '', \wz) &\geq I(\mu_n', \wz) - \max\pr{6 \zeta \frac{\sqrt{\nu_n}}{1-\zeta},2(1+\zeta) \nu_n^{\frac{1}{2} + \zeta}}\\
&\stackrel{(a)}{\geq}  I(\mu_n, \wz) - \max\pr{6 \zeta \frac{\sqrt{\nu_n}}{1-\zeta},2(1+\zeta) \nu_n^{\frac{1}{2} + \zeta}} - o(\nu_n^{\frac{1}{2}}), \label{eq:I-bound}
\end{align}
where $(a)$ follows from the argument of Theorem~\ref{th:bounded-supp}. Taking $\lambda_n = \mu_n'' \in \td{\Omega}^{\geq 0}$, by \eqref{eq:I-bound} and \eqref{eq:chi-limit}, we have \eqref{eq:A-bound} for $\epsilon = \frac{6\zeta}{1-\zeta}$. 
\end{proof}
Let $\mu\in \td{\Omega}^{\geq 0}$. We claim that
\begin{align}
\frac{I(\mu, \wy)}{\sqrt{\chi_2(\wz\circ \mu\| q_0)}} \leq \sup_{\td{\mu} \in \td{\Omega}^{> 0}}\frac{\E[\td{\mu}]{\theta_m^2 X - \log\pr{1+\theta_m^2 X}}}{\sqrt{\E[\td{\mu}\otimes \td{\mu}]{\frac{X_1X_2}{1-X_1X_2}}}}.
\end{align}

Let us define $\td{\mu}$ as
\begin{align}
\td{\mu}(A) \eqdef \frac{\mu(A\cap]0, 1[)}{\mu(]0, 1[)}.
\end{align}
In other words, $\td{\mu}$ is the probability measure $\mu$ conditioned to the event $]0, 1[$. We have
\begin{align}
I(\mu, \wy) \leq \E[\mu]{\theta_m^2 X - \log\pr{1+\theta_m^2 X}} \stackrel{(a)}{=} \mu(]0, 1[)  \E[\td{\mu}]{\theta_m^2 X - \log\pr{1+\theta_m^2 X}} ,
\end{align}
where $(a)$ follows since $\theta_m^2 x - \log(1+\theta_m^2x) = 0$ for $x = 0$.
Moreover,
\begin{align}
\chi_2(\wz \circ \mu \| q_0) = \E[\mu \circ \mu]{\frac{X_1 X_2}{1- X_1 X_2}} = \mu(]0, 1[)^2  \E[\td{\mu} \circ \td{\mu}]{\frac{X_1 X_2}{1- X_1 X_2}} 
\end{align}
Therefore,
\begin{align}
\frac{I(\mu, \wy)}{\sqrt{\frac{1}{2}\chi_2(\wz\circ \mu\| q_0)}} \leq \sqrt{2}\frac{\E[\td{\mu}]{\theta_m^2 X - \log\pr{1+\theta_m^2 X}}}{\sqrt{\E[\td{\mu}\otimes \td{\mu}]{\frac{X_1X_2}{1-X_1X_2}}}}.
\end{align}

Furthermore, with the help of Lemma~\ref{lm:a-prop}, Eq.~\eqref{eq:upper-a}, we have that 
\begin{align}
\limsup_{\nu \to 0^+}\frac{A(\nu)}{\sqrt{\nu}} \leq \sqrt{2}\theta_m^2.
\end{align}
Therefore, we obtain the upper-bound in \eqref{eq:bounds}.}
\section{Conclusion}
For covert communications over non-coherent wireless channels, we showed that discrete constellations with \modfirst{an} amplitude constraint are optimal. This differs from the results for coherent Gaussian channels in which using the phase is required to achieve the covert capacity. Supported by numerical results, we also conjectured that the optimal number of points is two and that their positions are fixed.

\appendices

\section{Leibniz Integral Rule}
\label{sec:leibn-integr-rule}
\modfirst{For a reader's convenience, we recall Leibniz integral rule here as it is used extensively throughout the paper.}
\begin{theorem}
\label{th:leibniz}
Let $\calO$ be an open subset of $\mathbb{R}$ and $\Omega$ be a  measure space. Suppose $f:\calO\times \Omega \to \mathbb{R}$ satisfies the following conditions
\begin{enumerate}
\item $f(x, w)$ is a Lebesgue-integrable function of $\omega$ for each $x\in \calO$
\item For almost all $\omega\in\Omega$, the derivative $\frac{\partial f}{\partial x}$ exists for all $x\in \calO$
\item There is an integrable function $\theta:\Omega\to \mathbb{R}$ such that $\left|\frac{\partial f}{\partial x}(x, \omega)\right| \leq \theta(\omega)$ for all $x\in \calO$ and almost every $\omega\in\Omega$.
\end{enumerate}
Then, for all $x\in\calO$, we have
\begin{align}
\frac{d}{dx} \int f(x, \omega)d\omega = \int \frac{\partial f}{\partial x}(x, \omega) d\omega
\end{align}
\end{theorem}
\modfirst{
\section{An Analyticity Criterion}
\label{sec:analytic}
\begin{theorem}
\label{th:analytic}
Let $g: \calD \times \mathbb{R} \to \mathbb{C}$ be a function such that $\calD$ is a simple connected subset of $\mathbb{C}$, $g(\cdot, y)$ is analytic for all $y\in \mathbb{R}$, and $\sup_{z\in \calC}\int_{\mathbb{R}} |g(z, y)|dy < \infty$ for all compact $\calC \subset \calD$. The function $f: z \mapsto \int_{\mathbb{R}} g(z, y) dy $ is analytic over the domain $\calD$.
\end{theorem}
\begin{proof}
The proof is  a straightforward application of Fubini's theorem and Morera's theorem. Fixing any closed piecewise  $C^1$ curve $\gamma$ in $\calD$, we have
\begin{align}
\int_\gamma f(z) dz 
&= \int_\gamma \int_{\mathbb{R}}g(z, y) dy dz\\
&\stackrel{(a)}{=} \int_{\mathbb{R}} \int_\gamma g(z, y) dz dy\\
&\stackrel{(b)}{=} 0,
\end{align}
where $(a)$ follows from Fubini's theorem and our assumption on $g$, and $(b)$ follows since $g(\cdot, z)$ is analytic and from Cauchy's integral theorem. Therefore, $f$ satisfies the condition of Morera's theorem and is analytic.
\end{proof}
}
\section{Auxiliary Results}
\label{sec:preliminary-results}

We gather here essential technical tools to prove the achievability and converse results. To begin with, we bound the \ac{PDF} of the  output distributions of the channels $\wy$ and $\wz$ for an arbitrary input distribution $\mu$.
\begin{proposition}
\label{prop:bound-pdf}
For any probability measure $\mu$ on $\calX$ with $\E[\mu]{X} < \infty$ and all $y\in \calY, z\in\calZ$, we have 
\begin{align}
  -\theta_m^2\E[\mu]{X} -y &\leq \log((w_{Y|X}\circ \mu)(y)) \leq 0,\label{eq:bound-wy}\\
  -\E[\mu]{X}  -z &\leq \log((w_{Z|X}\circ \mu)(z)) \leq 0,\label{eq:bound-wz}\\
  \E[\wy\circ\mu]{Y} &= 1 + \theta_m^2\E[\mu]{X},\label{eq:expected-y}\\
  \E[\wz\circ\mu]{Z} &= 1 + \E[\mu]{X}.\label{eq:expected-z}
\end{align} 
\end{proposition}
\begin{proof}
We only prove \eqref{eq:bound-wy} and \eqref{eq:expected-y}, from which \eqref{eq:bound-wz} and \eqref{eq:expected-z} follow by setting $\theta_m = 1$. To obtain \eqref{eq:bound-wy}, observe that for any $x\in\calX$, we have $p_x(y) \eqdef \frac{1}{1+\theta_m^2x} e^{-\frac{y}{1+\theta_m^2x}} \leq  1$, and
\begin{align}
\log \pr{(\wy\circ \mu)(y)} 
&= \log \pr{\E[\mu]{p_X(y)}}\\
&\stackrel{(a)}{\geq} \E[\mu]{\log\pr{p_X(y)}}\\
&= \E[\mu]{-\log\pr{1+\theta_m^2X} - \frac{y}{1+\theta_m^2X}}\\
&\stackrel{(b)}{\geq} \E[\mu]{ -\theta_m^2X  - \frac{y}{1+\theta_m^2X}}\\
&\stackrel{(c)}{\geq} \E[\mu]{ -\theta_m^2X - y}\\
&= -\theta_m^2\E[\mu]{X} - y,
\end{align}
where $(a)$ follows from Jensen's inequality, $(b)$ follows from $\log(1+x) \leq x$ for $x>-1$, and $(c)$ follows from $\P[\mu]{X\geq 0} = 1$. To obtain \eqref{eq:expected-y}, note that
\begin{align}
\E[\wy\circ\mu]{Y} 
&= \int_0^\infty y (\wy \circ \mu)(y) dy\\
&=  \int_0^\infty y \pr{\int_{\calX} p_x(y)d\mu }dy\\
&\stackrel{(a)}{=}  \int_{\calX}\pr{ \int_0^\infty y p_x(y)dy} d\mu\\
&=   \int_{\calX}(1+\theta_m^2x) d\mu\\
&= 1+\theta_m^2\E[\mu]{X},
\end{align}
where $(a)$ follows from Fubini's theorem and the fact that for all $x,y$, $yp_x(y)\geq 0$.
\end{proof}
\modfirst{
\begin{lemma}
\label{lm:mean-finite}
Let $\mu$ be a probability measure over $\calX$. If $\D{\wz\circ \mu}{q_0}$ exists and is finite, then $\E[\mu]{X} < \infty$.
\end{lemma}
\begin{proof}
We proceed by contradiction. Consider a positive real number $\gamma_1$ and let $2\epsilon \eqdef \mu([\gamma_1, \infty)$. We have $\epsilon > 0$, because otherwise $\E[\mu]{X} \leq \gamma_1 < \infty$.  By the continuity of a probability, we have
\begin{align}
\lim_{\gamma \to \infty} \mu([\gamma_1, \gamma]) = 2\epsilon.
\end{align}
Therefore, there exists $\gamma_2 \geq \gamma_1$ such that $\mu([\gamma_1, \gamma_2]) \geq \epsilon$.
We then have 
\begin{align}
(\wz\circ \mu)(z) \geq \frac{\epsilon e^{-\frac{z}{1+\gamma_1}}}{1+\gamma_2}.
\end{align}
This implies that $(\wz\circ \mu)(z) \geq q_0(z) = e^{-z}$ for all $z \geq z_0 \eqdef \frac{1+\gamma_1}{\gamma_1} \log \frac{1+\gamma_2}{\epsilon} > 0$.
Since $\D{\wz\circ \mu}{q_0} < \infty$, we have 
\begin{align}
\infty 
&> \int_{z_0}^\infty (\wz\circ \mu)(z) \log\frac{(\wz\circ \mu)(z)}{q_0(z)} dz\\
&\geq \int_{z_0}^\infty (\wz\circ \mu)(z) \log\frac{ \frac{\epsilon e^{-\frac{z}{1+\gamma_1}}}{1+\gamma_2}}{ e^{-z}} dz\\
&\geq \log \frac{\epsilon}{1+\gamma_2} + \frac{\gamma_1}{1+ \gamma_1} \int_{z_0}^\infty (\wz\circ \mu)(z)zdz\\
&= \log \frac{\epsilon}{1+\gamma_2} + \frac{\gamma_1}{1+ \gamma_1} \int_{z_0}^\infty \int_{\calX}\frac{ e^{-\frac{z}{1+x}}}{1+x} d\mu zdz\\
&=  \log \frac{\epsilon}{1+\gamma_2} + \frac{\gamma_1}{1+ \gamma_1}  \int_{\calX}\int_{z_0}^\infty\frac{ e^{-\frac{z}{1+x}}}{1+x}  zdzd\mu\\
&=   \log \frac{\epsilon}{1+\gamma_2} + \frac{\gamma_1}{1+ \gamma_1}  \int_{\calX}{(1+x)\pr{1 + \frac{z_0}{1+x}} e^{-\frac{z_0}{1+x}}} d\mu\\
&\geq  \log \frac{\epsilon}{1+\gamma_2} + \frac{\gamma_1}{1+ \gamma_1}  \int_{\calX}{(1+x) e^{-z_0}} d\mu\\
& \geq \log \frac{\epsilon}{1+\gamma_2} + \frac{\gamma_1}{1+ \gamma_1} \pr{\E[\mu]{X} + 1}e^{-z_0},
\end{align}
which implies that $\E[\mu]{X} < \infty$.
\end{proof}}
The next result shows that an upper-bound on $\D{\wz\circ \mu}{q_0} $ leads to an upper-bound on  $\E[\mu]{X}$.
\begin{lemma}
\label{lm:expected-x-bound}
For any $\nu>0$ and for any probability measure $\mu$ on $\calX$, $\D{\wz\circ \mu}{q_0} \leq \nu$ implies that $\E[\mu]{X}\leq2\sqrt{\nu} + \nu$.
\end{lemma}
\begin{proof}
\modfirst{For any $x\in\mathbb{R}^+$, we first consider the relative entropy $\D{w_{Z|X}\circ \mu}{q_x}$ and show that it exists. By~\eqref{eq:bound-wz} in Proposition~\ref{prop:bound-pdf} applied to a distribution with a single mass point at $x$, $|\log q_x(z)| \leq x + z$. We thus have $\int_0^{\infty} (w_{Z|X}\circ \mu)(z) \abs{\log q_x(z)} dz\leq x + \E[\wz\circ \mu]{Z} = x + 1 + \E[\mu]{X}$, which is finite by Lemma~\ref{lm:mean-finite}. Consequently, $\int_0^{\infty} (w_{Z|X}\circ \mu)(z) \log q_x(z) dz$ is finite, and therefore by \cite[Lemma 8.3.1]{InformationTheory}, the relative entropy $\D{\wz\circ\mu}{q_x}$ exists and is finite.} Accordingly, we have
\begin{align}
0
&\geq -\int_0^{\infty}  (w_{Z|X}\circ \mu)(z) \log \frac{(w_{Z|X} \circ \mu)(z)}{q_x(z)}dz \\
&=\int_0^{\infty}  (w_{Z|X}\circ \mu)(z) \pr{ -\log ((w_{Z|X} \circ \mu)(z)) -\log(1+x) -\frac{z}{1+x}} dz. \label{eq:d-x}
\end{align}
Furthermore, by our assumption that $\D{w_{Z|X}\circ \mu}{q_0} \leq \nu$, we have
\begin{align}
\nu 
&\geq \int_0^{\infty} (w_{Z|X}\circ \mu)(z)  \log \frac{(w_{Z|X} \circ \mu)(z)}{q_0(z)}dz \\
&=\int_0^{\infty}  (w_{Z|X}\circ \mu)(z)  \pr{ \log ((w_{Z|X} \circ \mu)(z)) + z}dz. \label{eq:d-nu} 
\end{align}
Adding the inequalities in \eqref{eq:d-x} and \eqref{eq:d-nu}, we obtain
\begin{align}
\nu 
&\geq \int_0^{\infty}  (w_{Z|X}\circ \mu)(z) \pr{-\log(1+x) + \frac{xz}{1+x}}dz\\
&= -\log(1+x)+ \frac{x}{1+x} \E[\wz\circ \mu]{Z}\\
&\stackrel{(a)}{=} -\log(1+x) + \frac{x}{1+x}\pr{\E[\mu]{X} + 1},
\end{align}
where $(a)$ follows from \eqref{eq:expected-z}. Hence, we have
\begin{align}
\E[\mu]{X} &\leq \pr{\nu+ \log(1+x)} \frac{1+x}{x} - 1\\
&\leq \pr{\nu + x} \frac{1+x}{x} -1.
\end{align}
Choosing $x=\sqrt{\nu}$, we obtain the desired upper-bound.
\end{proof}

\begin{lemma}
For any probability measure $\mu$ on $\calX$ with $\E[\mu]{X} < \infty$, $I(\mu, \wy)$  is well-defined and finite, and
\begin{align}
  I(\mu, \wy) = -\int_0^\infty (\wy \circ \mu)(y) \log  ((\wy \circ \mu)(y)) dy - \E[\mu]{ \log\pr{1 + \theta_m^2 X}} - 1.\label{eq:mutual-exp}
\end{align}
\end{lemma}
\begin{proof}
To check that $I(\mu, \wy)$ is well-defined and finite, it is enough to show that ${\displaystyle \int }\abs{\log \frac{p_x(y)}{(w_{Y|X}\circ\mu)(y)}}d(w_{Y|X}\times \mu) < \infty$, which holds since
\begin{align}
\int \left|\log \frac{p_x(y)}{(w_{Y|X}\circ\mu)(y)}\right|d(w_{Y|X}\times \mu) 
&\leq \int \pr{|\log p_x(y)| + |\log((w_{Y|X}\circ\mu)(y)) |}d(w_{Y|X}\times \mu) \\
&\stackrel{(a)}{\leq}  \int \pr{ \theta_m^2(x + \E[\mu]{X}) +2y}d(w_{Y|X}\times \mu) \\
&= 2\theta_m^2 \E[\mu]{X} + 2\E[\wy\circ\mu ]{Y}\\
&\stackrel{(b)}{=} 4\theta_m^2 \E[\mu]{X}  + 2 <\infty,
\end{align}
where $(a)$ follows from~\eqref{eq:bound-wy}, and $(b)$ follows from \eqref{eq:expected-y}. Note next that
\begin{align}
I(\mu, \wy)
&= \E[w_{Y|X}\times \mu]{\log \frac{p_X(Y)}{(\wy\circ\mu)(Y)}}\\
&= \E[w_{Y|X}\times \mu]{-\log(1+\theta_m^2 X) - \frac{Y}{1+\theta_m^2X} -\log ((\wy\circ\mu)(Y))}.
\end{align}
Moreover, $\E{\log(1+\theta_m^2 X)} \leq \theta_m^2 \E{X} < \infty$ and $\E{ \frac{Y}{1+\theta_m^2X}} \leq \E{Y} < \infty$, and therefore, we can use the linearity of expectation to write 
\begin{align}
&\E[w_{Y|X}\times \mu]{-\log(1+\theta_m^2 X) - \frac{Y}{1+\theta_m^2X} -\log \pr{(\wy\circ\mu)(Y)}}\\
&= -\E{\log(1+\theta_m^2 X)} - \E{\frac{Y}{1+\theta_m^2X}} -\E{\log\pr{(\wy\circ\mu)(Y)}}\\
&= -\E{\log(1+\theta_m^2 X)} - \E{\E{\frac{Y}{1+\theta_m^2X}\bigg|X}} -\E{\log\pr{(\wy\circ\mu)(Y)}}\\
&= -\E{\log(1+\theta_m^2 X)} - \E{\frac{1+\theta_m^2X}{1+\theta_m^2X}} -\E{\log\pr{(\wy\circ\mu)(Y)}}\\
&= -\E{\log(1+\theta_m^2 X)} - 1 -\E{\log \pr{(\wy\circ\mu)(Y)}},
\end{align}
which completes the proof of \eqref{eq:mutual-exp}.
\end{proof}
\begin{lemma}
\label{lm:cross-entropy}
Suppose that $\D{\wz \circ \mu_1}{q_0}$ and $\D{\wz \circ \mu_2}{q_0}$ exist  and are finite for two probability measures $\mu_1$ and $\mu_2$ on $\calX$. Then, the cross entropy $\int_0^\infty (\wz\circ \mu_1)(z) \log (\wz \circ \mu_2(z)) dz$ exists and is finite.
\end{lemma}
\begin{proof}
We shall show that $\int_0^\infty (\wz\circ \mu_1)(z) |\log ((\wz \circ \mu_2)(z))| dz < \infty$. By Lemma~\ref{lm:expected-x-bound}, we know that $\E[\mu_1]{X}$ and $\E[\mu_2]{X}$ are finite. Therefore, we have
\begin{align}
\int_0^\infty (\wz\circ \mu_1)(z) |\log (\wz \circ \mu_2(z))| dz
&\stackrel{(a)}{\leq} \int_0^\infty (\wz\circ \mu_1)(z) \pr{\E[\mu_2]{X} + z}dz\\
&=\E[\mu_2]{X} + \E[\wz\circ \mu_1]{Z}\\
&\stackrel{(b)}{=}  \E[\mu_2]{X} + 1+ \E[ \mu_1]{X} < \infty
\end{align}
where $(a)$ follows from \eqref{eq:bound-wz}, and $(b)$ follows from \eqref{eq:expected-z}.
\end{proof}
\modfirst{
\begin{lemma}
\label{lm:chi2-value}
Let $\mu$ be a probability measure over $\calX$ such that $\sup(\supp{\mu}) < \infty$. We then have
\begin{align}
I(\mu, \wy) = \E[\mu]{\theta_m^2 X - \log(1+\theta_m^2X)} - \D{\wy\circ \mu}{p_0}.
\end{align}
Furthermore, if we have $\sup(\supp{\mu}) < 1$, then
\begin{align}
\chi_2(\wz \circ \mu\| q_0) = \E[\mu \circ \mu]{\frac{X_1 X_2}{1 - X_1 X_2}}.
\end{align}
\end{lemma}
\begin{proof}
We have
\begin{align}
I(\mu, \wy) 
&= \int \log \frac{p_X(Y)}{(\wy \circ \mu)(Y)} d(\wy \times \mu)\\
&=  \int \log \frac{p_X(Y)}{p_0(Y)} d(\wy \times \mu)+  \int \log \frac{p_0(Y)}{(\wy \circ \mu)(Y)} d(\wy \times \mu)\\
&= \E[\mu]{\D{p_X}{p_0}} - \D{\wy \circ \mu}{p_0}\\
&\stackrel{(a)}{=}  \E[\mu]{\theta_m^2 X - \log(1+\theta_m^2X)} - \D{\wy\circ \mu}{p_0},
\end{align}
where $(a)$ follows from the straightforward calculation of the relative entropy between two exponential distribution.
Additionally, we have
\begin{align}
\chi_2(\wz \circ \mu\| q_0) 
&= \int_0^\infty \frac{(\wz\circ\mu)(z) ^2}{q_0(z)} dz -1\\
&=  \int_0^\infty \E[\mu\otimes \mu]{\frac{1}{(1+X_1)(1+X_2)}e^{z\pr{1 - \frac{1}{1+X_1} - \frac{1}{1+X_2}}} }dz -1\\
&\stackrel{(a)}{=}  \E[\mu\otimes \mu]{ \int_0^\infty\frac{1}{(1+X_1)(1+X_2)}e^{z\pr{1 - \frac{1}{1+X_1} - \frac{1}{1+X_2}}}dz } -1\\
&=  \E[\mu\otimes \mu]{\frac{1}{1-X_1X_2}} - 1\\
&=  \E[\mu\otimes \mu]{\frac{X_1X_2}{1-X_1X_2}},
\end{align}
where $(a)$ follows from Fubini theorem and $\frac{1}{(1+X_1)(1+X_2)}e^{z\pr{1 - \frac{1}{1+X_1} - \frac{1}{1+X_2}}} \geq 0$ almost surely.
\end{proof}
}
\begin{lemma} 
\label{lm:d-two-points}
If $a > 1$ and $\beta>0$ is small enough, then
\begin{multline}
\D{\beta q_a + (1-\beta)q_0}{q_0} =\\ \beta^{1+\frac{1}{a}}\pr{{1+a}}^{-1-\frac{1}{a}}\pr{1+\frac{1}{a}}\pr{\frac{ \Gamma\pr{-\frac{1}{a}}\Gamma\pr{2+\frac{1}{a}}}{\pr{1+\frac{1}{a}}^2} + a^2 \Gamma\pr{1-\frac{1}{a}}\Gamma\pr{1+\frac{1}{a}}} + O(\beta^2),
\end{multline}
where $\Gamma(x) \eqdef \int_0^\infty y^{x-1}e^{-y}dy$.

If $a < 1$ and $\beta > 0$ is small enough, then
\begin{align}
\D{\beta q_a + (1-\beta)q_0}{q_0} = \frac{\modfirst{a^2}}{2(1-a^2)}\beta^2 + o(\beta^2).
\end{align}
\end{lemma}
\begin{proof}

We only consider the case where $a>1$ and the other case follows from similar approach. By definition, we have
\begin{align}
\D{\beta q_a + (1-\beta)q_0}{q_0} 
&= \int_0^\infty \pr{\beta q_a(z) + (1-\beta)q_0(z)}\log \pr{\frac{\beta q_a(z) + (1-\beta)q_0(z)}{q_0(z)}} dz\\
&= \int_0^\infty  \pr{\beta \frac{e^{-\frac{z}{1+a}}}{1+a} + (1-\beta)e^{-z}}\log \pr{1-\beta + \frac{\beta}{1+a} e^{\frac{az}{1+a}}} dz\\
&= \log( 1 - \beta) + \int_0^\infty  \pr{\beta \frac{e^{-\frac{z}{1+a}}}{1+a} + (1-\beta)e^{-z}}\log \pr{1+ \frac{\beta}{(1-\beta)(1+a)} e^{\frac{az}{1+a}}} dz.\label{eq:d-two-p}
\end{align}
By substitution $u\eqdef e^{\frac{az}{1+a}}$ in the above integral, we obtain
\begin{multline}
 \int_0^\infty  \pr{\beta \frac{e^{-\frac{z}{1+a}}}{1+a} + (1-\beta)e^{-z}}\log \pr{1+ \frac{\beta}{(1-\beta)(1+a)} e^{\frac{az}{1+a}}} dz 
 =\\ \pr{1 + \frac{1}{a}}\int_1^\infty \pr{(1-\beta)u^{-2-\frac{1}{a}} + \frac{\beta }{1+a}u^{-1-\frac{1}{a}}}\log\pr{1+\frac{\beta}{(1-\beta) (1+a)}u}du\label{eq:integral-u}
\end{multline}
Note next that for all real numbers $\lambda_1, \lambda_2$, a primitive function of $u^{\lambda_1}\log\pr{1+\lambda_2 u}$ is
\begin{align}
\int u^{\lambda_1}\log\pr{1+\lambda_2 u} du = \frac{u^{\lambda_1 + 1}\pr{{}_2F_1(1, \lambda_1 + 1; \lambda_1 + 2; -\lambda_2 u) + (\lambda_1 +1)\log(\lambda_2 u + 1) - 1}}{(\lambda_1 + 1)^2} + \text{ constant},
\end{align}
where ${}_2F_1(a, b;c;x)$ is the hypergeometric function. Additionally, for $\lambda_1 < - 1$, the limit of  this primitive function at $u=\infty$ is 
\begin{align}
\frac{\lambda_2^{-\lambda_1 - 1} \Gamma(2+\lambda_1)\Gamma(-\lambda_1)}{(\lambda_1 + 1)^2}.
\end{align}
Therefore, if we define $\lambda\eqdef \frac{\beta}{(1-\beta)(1+a)}$, by linearity of integral, we have
\begin{align}
 \displaybreak[0]&\int_1^\infty \pr{(1-\beta)u^{-2-\frac{1}{a}} + \frac{\beta u^{-\frac{1}{a}}}{1+a}}\log\pr{1+\frac{\beta}{(1-\beta) (1+a)}u}du\\
 \displaybreak[0]&= (1-\beta)\pr{\frac{\lambda^{1+\frac{1}{a}} \Gamma\pr{-\frac{1}{a}}\Gamma\pr{2+\frac{1}{a}}}{\pr{1 + \frac{1}{a}}^2} -  \frac{{}_2F_1\pr{1, -1-\frac{1}{a}; -\frac{1}{a}; -\lambda} - (1+\frac{1}{a})\log(\lambda + 1) - 1}{\pr{1 + \frac{1}{a}}^2}}+ \\
\displaybreak[0]&~~~~~~~~\frac{\beta}{1+a}\pr{\frac{\lambda^{\frac{1}{a}} \Gamma\pr{1-\frac{1}{a}}\Gamma\pr{1+\frac{1}{a}}}{\pr{\frac{1}{a}}^2} -  \frac{{}_2F_1\pr{1, -\frac{1}{a}; 1- \frac{1}{a}; -\lambda} - (\frac{1}{a})\log(\lambda + 1) - 1}{\pr{\frac{1}{a}}^2}}\\
\displaybreak[0]&\stackrel{(a)}{=}  (1-\beta)\pr{\frac{\lambda^{1+\frac{1}{a}} \Gamma\pr{-\frac{1}{a}}\Gamma\pr{2+\frac{1}{a}}}{\pr{1 + \frac{1}{a}}^2} -  \frac{\pr{1 + \frac{\lambda\pr{1+\frac{1}{a}}}{-\frac{1}{a}} } - (1+\frac{1}{a})\lambda - 1 + O(\beta^2)}{\pr{1 + \frac{1}{a}}^2}}+ \\
\displaybreak[0]&~~~~~~~~\frac{\beta}{1+a}\pr{\frac{\lambda^{\frac{1}{a}} \Gamma\pr{1-\frac{1}{a}}\Gamma\pr{1+\frac{1}{a}}}{\pr{\frac{1}{a}}^2} -  \frac{\pr{1+\frac{\lambda \frac{1}{a}}{1-\frac{1}{a}}} - (\frac{1}{a})\lambda + O(\beta^2) - 1}{\pr{\frac{1}{a}}^2}}\\
\displaybreak[0]&=  (1-\beta)\pr{\frac{\lambda^{-1-\frac{1}{a}} \Gamma\pr{-\frac{1}{a}}\Gamma\pr{2+\frac{1}{a}}}{\pr{1 + \frac{1}{a}}^2} -  \frac{\frac{\lambda\pr{1+\frac{1}{a}}}{-\frac{1}{a}}  - (1+\frac{1}{a})\lambda  + O(\beta^2)}{\pr{1 + \frac{1}{a}}^2}}+ \\
\displaybreak[0]&~~~~~~~~\lambda(1-\beta)\pr{\frac{\lambda^{-\frac{1}{a}} \Gamma\pr{1-\frac{1}{a}}\Gamma\pr{1+\frac{1}{a}}}{\pr{\frac{1}{a}}^2} -  \frac{\frac{\lambda \frac{1}{a}}{1-\frac{1}{a}} - (\frac{1}{a})\lambda + O(\beta^2) }{\pr{\frac{1}{a}}^2}},
\end{align}
where $(a)$ follows since for $x$ going to zero ${}_2F_1(a, b;c;x) = 1 +abx/c  + O(x^2)$ and $\log(1+x) = x + O(x^2)$ by Taylor's expansion. By rearranging the terms in above expression and disregarding the higher order terms, we obtain
\begin{align}
 &\lambda^{1+\frac{1}{a}}(1-\beta)\pr{\frac{ \Gamma\pr{-\frac{1}{a}}\Gamma\pr{2+\frac{1}{a}}}{\pr{1+\frac{1}{a}}^2} + \frac{ \Gamma\pr{1-\frac{1}{a}}\Gamma\pr{1+\frac{1}{a}}}{\pr{\frac{1}{a}}^2}} + \lambda \frac{1+a}{1+\frac{1}{a}}(1-\beta) + O(\beta^2)\\
 &=\beta^{1+\frac{1}{a}}\pr{\frac{1}{(1-\beta)(1+a)}}^{1+\frac{1}{a}}(1-\beta)\pr{\frac{ \Gamma\pr{-\frac{1}{a}}\Gamma\pr{2+\frac{1}{a}}}{\pr{1+\frac{1}{a}}^2} + \frac{ \Gamma\pr{1-\frac{1}{a}}\Gamma\pr{1+\frac{1}{a}}}{\pr{\frac{1}{a}}^2}} +\frac{\beta}{1+\frac{1}{a}} + O(\beta^2).\label{eq:integral-d}
\end{align}
Combining \eqref{eq:d-two-p}, \eqref{eq:integral-u}, and \eqref{eq:integral-d}, we have
\begin{align}
&\D{\beta q_a + (1-\beta)q_0}{q_0} \\
&= \beta^{1+\frac{1}{a}}\pr{\frac{1}{(1-\beta)(1+a)}}^{1+\frac{1}{a}}(1-\beta)(1+\frac{1}{a})\pr{\frac{ \Gamma\pr{-\frac{1}{a}}\Gamma\pr{2+\frac{1}{a}}}{\pr{1+\frac{1}{a}}^2} + \frac{ \Gamma\pr{1-\frac{1}{a}}\Gamma\pr{1+\frac{1}{a}}}{\pr{\frac{1}{a}}^2}} + O(\beta^2).
\end{align}\end{proof}

\section{Error Exponents Analysis}
\label{sec:error-expon-analys}
\begin{lemma}
\label{lm:density-2m}
For a probability measure on $\calX$,  $\mu$,  for which we have $\modfirst{\max}( \supp{\mu}) \eqdef x_{\max} < \infty$ and $\D{w_{Z|X}\circ \mu}{q_0} \leq \nu$ and for any $A > 0$, it holds that 
\begin{multline}
\E[w_{Y|X}\times \mu]{\log^2\pr{\frac{p_X(Y)}{(w_{Y|X}\circ\mu)(Y)}}}\\
\leq2(3+x_{\max})(2\sqrt{\nu} + \nu)(1+\theta_m^2x_{\max})^4 \pr{e^{Ax_{\max}}A + \theta_m^2}^2e^{2A} + 20 \pr{(1+\theta_m^2x_{\max})  + A }^2e^{-A}.
\end{multline}
\end{lemma}
\begin{proof}
%
We first define $f(x) \eqdef \int_0^{\infty} p_{x}(y) \log^2\pr{\frac{p_{x}(y)}{(\wy \circ \mu)(y)}}dy$ for which we have
\begin{align}
\label{eq:integral-split-f}
 f(x) = \int_0^{A} p_{x}(y) \log^2\pr{\frac{p_{x}(y)}{(\wy \circ \mu)(y)}}dy +  \int_A^{\infty} p_{x}(y) \log^2\pr{\frac{p_{x}(y)}{(\wy \circ \mu)(y)}}dy,
\end{align}
for any $A>0$. To upper-bound the first term, we note that
\begin{align}
\left|\frac{(\wy \circ \mu)(y)}{p_{x}(y)} - 1\right|
&= \left|\frac{\E[\mu]{\frac{1}{1+\theta_m^2 \widetilde{X}}{e^{-\frac{y}{1+\theta_m^2\widetilde{X}}}}}}{p_{x}(y)} - 1\right|\\
&= \left|{\E[\mu]{\frac{1+\theta_m^2x}{1+\theta_m^2 \widetilde{X}}{e^{\frac{y(x - \widetilde{X})}{(1+\theta_m^2\widetilde{X})(1+\theta_m^2x)}}}-1} }\right|\\
&\leq  \E[\mu]{\left|\frac{1+\theta_m^2x}{1+\theta_m^2 \widetilde{X}}{e^{\frac{y(x - \widetilde{X})}{(1+\theta_m^2\widetilde{X})(1+\theta_m^2x)}}}-1 \right|}\\
&\leq  \E[\mu]{\left|\frac{1+\theta_m^2x}{1+\theta_m^2 \widetilde{X}}\pr{e^{\frac{y(x - \widetilde{X})}{(1+\theta_m^2\widetilde{X})(1+\theta_m^2x)}}-1} \right|} +  \E[\mu]{\left|\frac{1+\theta_m^2x}{1+\theta_m^2 \widetilde{X}}-1 \right|}.\label{eq:log2arg}
\end{align}
Considering each term separately in the above expression, we have
\begin{align}
 \E[\mu]{\left|\frac{1+\theta_m^2x}{1+\theta_m^2 \widetilde{X}}\pr{e^{\frac{y(x - \widetilde{X})}{(1+\theta_m^2\widetilde{X})(1+\theta_m^2x)}}-1} \right|} 
 &\leq (1+\theta_m^2x_{\max})  \E[\mu]{\left|e^{\frac{y(x - \widetilde{X})}{(1+\theta_m^2\widetilde{X})(1+\theta_m^2x)}}-1 \right|} \\
 &\stackrel{(a)}{\leq} (1+\theta_m^2x_{\max}) e^{yx_{\max}} \E[\mu]{\left|{\frac{y(x - \widetilde{X})}{(1+\theta_m^2\widetilde{X})(1+\theta_m^2x)}} \right|}\\
&\leq (1+\theta_m^2x_{\max}) e^{yx_{\max}}y\pr{x+\E[\mu]{\widetilde{X}}},
\end{align}
where $(a)$ follows from the mean value theorem and an upper-bound on derivative. For the next term in \eqref{eq:log2arg}, we have
\begin{align}
\E[\mu]{\left|\frac{1+\theta_m^2x}{1+\theta_m^2 \widetilde{X}} -1\right|}
&= \theta_m^2\E[\mu]{\left|\frac{x - \widetilde{X}}{1+\theta_m^2 \widetilde{X}} \right|}\\
&\leq \theta_m^2\E[\mu]{|x - \widetilde{X}|}\\
&\leq \theta_m^2\pr{x + \E[\mu]{\widetilde{X}}}.
\end{align}
Combining these two inequalities, we obtain
\begin{align}
\left|\frac{(\wy \circ \mu)(y)}{p_{x}(y)} - 1\right| \leq \pr{x+\E[\mu]{\widetilde{X}}}\pr{(1+\theta_m^2x_{\max}) e^{yx_{\max}}y + \theta_m^2}.
\end{align}
Hence, using the inequalities $\log^2(x) \leq (1-x)^2(1+x^{-2})$ for $x>-1$ and $\frac{p_{x}(y)}{(\wy \circ \mu)(y)} \leq (1+\theta_m^2x_{\max})e^{y}$, we have
\begin{align}
\log ^2 \pr{\frac{p_{x}(y)}{(\wy \circ \mu)(y)}} &\leq \pr{\pr{x+\E[\mu]{\widetilde{X}}}\pr{(1+\theta_m^2x_{\max}) e^{yx_{\max}}y + \theta_m^2}}^2 \pr{1 + \pr{(1+\theta_m^2x_{\max})e^{y}}^2}.
\end{align}
This yields that 
\begin{align}
&\int_0^A p_x(y) \log ^2 \pr{\frac{p_{x}(y)}{(\wy \circ \mu)(y)}} dy\\
&\leq \sup_{y\in[0, A]}  \pr{\pr{x+\E[\mu]{\widetilde{X}}}\pr{(1+\theta_m^2x_{\max}) e^{yx_{\max}}y + \theta_m^2}}^2 \pr{1 + \pr{(1+\theta_m^2x_{\max})e^{y}}^2}\\
&= \pr{x+\E[\mu]{\widetilde{X}}}^2\pr{(1+\theta_m^2x_{\max}) e^{Ax_{\max}}A + \theta_m^2}^2 \pr{1 + \pr{(1+\theta_m^2x_{\max})e^{A}}^2}\\
&\leq 2 \pr{x+\E[\mu]{\widetilde{X}}}^2(1+\theta_m^2x_{\max})^4 \pr{e^{Ax_{\max}}A + \theta_m^2}^2e^{2A}.
\end{align}
For the second term in \eqref{eq:integral-split-f}, if $x\leq x_{\max}$, then we have
\begin{align}
&\int_A^\infty p_x(y) \log ^2 \pr{\frac{p_{x}(y)}{(\wy \circ \mu)(y)}} dy\\
&\leq 4 \int_A^\infty p_x(y) \pr{\log(1+\theta_m^2x_{\max}) + y}^2 dy\\
&= 4 \pr{-\pr{\log^2(1+\theta_m^2x_{\max}) + 2\log(1+\theta_m^2x_{\max}) \pr{y+ 1+\theta_m^2x}  + \pr{y^2+2(1+\theta_m^2x)y +2(1+\theta_m^2x)^2 }}e^{-\frac{y}{1+\theta_m^2x}}  }\bigg |_A^\infty\\
&= 4 \pr{\log^2(1+\theta_m^2x_{\max}) + 2\log(1+\theta_m^2x_{\max}) \pr{A+ 1+\theta_m^2x}  + \pr{A^2+2(1+\theta_m^2x)A +2(1+\theta_m^2x)^2 }}e^{-\frac{A}{1+\theta_m^2x}}\\
&\leq 20 \pr{(1+\theta_m^2x_{\max})  + A }^2e^{-A}.
\end{align}
Therefore, for all $x\in \calX$, it holds that
\begin{align}
f(x) \leq 2 \pr{x+\E[\mu]{\widetilde{X}}}^2(1+\theta_m^2x_{\max})^4 \pr{e^{Ax_{\max}}A + \theta_m^2}^2e^{2A} + 20 \pr{(1+\theta_m^2x_{\max})  + A }^2e^{-A},
\end{align}
which implies that
\begin{align}
&\E[w_{Y|X}\times \mu]{\log^2\pr{\frac{p_X(Y)}{(w_{Y|X}\circ\mu)(Y)}}}\\
&= \E[\mu]{f(X)}\\
&\leq  2 \E[\mu]{\pr{X+\E[\mu]{\widetilde{X}}}^2}(1+\theta_m^2x_{\max})^4 \pr{e^{Ax_{\max}}A + \theta_m^2}^2e^{2A} + 20 \pr{(1+\theta_m^2x_{\max})  + A }^2e^{-A}.
\end{align}
Finally, by  Lemma~\ref{lm:expected-x-bound}, $ \E[\mu]{\pr{X+\E[\mu]{\widetilde{X}}}^2} = \E[\mu]{X^2} + 3\pr{\E[\mu]{X}}^2 \leq (3+x_{\max}) \pr{\nu + 2\sqrt{\nu}}$ which completes the proof.
\end{proof}
\begin{proof}[Proof of Lemma~\ref{lm:exponent-rel}]
We  fix $\mu$ with $\sup(\supp{\mu})\eqdef \widetilde{x} < \infty$ and use Theorem~\ref{th:leibniz} along with induction to show that for a small neighborhood around zero and all $i\geq 0$, we have
\begin{align}
\label{eq:deriv-g}
\frac{\partial^i g}{\partial s^i}(s, \mu) = \E[\wy \times \mu]{\log^i \pr{ \frac{p_X(Y)}{(\wy\circ \mu)(Y)}} \pr{\frac{p_X(Y)}{(\wy \circ \mu)(Y)}}^s},
\end{align}
where 
\begin{align}
g(s, \mu) \eqdef  \E[\wy \times \mu]{\pr{\frac{p_X(Y)}{(\wy\circ \mu)(Y)}}^s}.
\end{align}
The statement is true for $i=0$ by definition. For $i>0$, we take $\calO =[0, \widetilde{s}]$, $\Omega = (\calX\times \calY, w_{Y|X}\times \mu)$, and $f(s, x, y) = \log^{i-1} \pr{ \frac{p_x(y)}{(\wy\circ \mu)(y)}}\pr{\frac{p_x(y)}{(\wy\circ \mu)(y)}}^s$ and check the three conditions in Theorem~\ref{th:leibniz}:
\begin{enumerate}
\item For $x \leq \widetilde{x}$, we have
\begin{align}
|f(s, x, y)| 
&= \left|\log^{i-1} \pr{ \frac{p_x(y)}{(\wy\circ \mu)(y)}}\pr{\frac{p_x(y)}{(\wy\circ \mu)(y)}}^s\right|\\
&\stackrel{(a)}{\leq} \left| \pr{\theta_m^2\pr{\E[\mu]{X}+x}  + 2y}^{i-1}\pr{\frac{p_x(y)}{(\wy\circ \mu)(y)}}^s  \right|\\
&\leq \left| \pr{2\theta_m^2\widetilde{x} + 2y}^{i-1}\pr{\frac{p_x(y)}{(\wy\circ \mu)(y)}}^s  \right|\\
&\leq \left| \pr{2\theta_m^2\widetilde{x}+ 2y} \right|^{i-1}(1+\widetilde{x})^s e^{\frac{s\widetilde{x}y}{1+\widetilde{x}}},
\end{align}
where $(a)$ follows from Proposition~\ref{prop:bound-pdf}. Because the above upper-bound does not depend on $x$, we can write
\begin{align}
\E[\wy\times \mu]{|f(s, X, Y)|}
&\leq \E[\wy \circ \mu]{\left| \pr{2\theta_m^2\widetilde{x} + 2Y} \right|^{i-1}(1+\widetilde{x})^s e^{\frac{s\widetilde{x}Y}{1+\widetilde{x}}}}.
\end{align}
Moreover, note that the moment generating function of a random variable with exponential distribution and mean $\lambda$ exists in $[0, \lambda)$, which implies that the moment generating function of distribution $\wy \circ \mu$ exists in $[0, 1/(1+\widetilde{x}))$. Hence,  there exists $\widetilde{s}$ depending on $\widetilde{x}$ such that
\begin{align}
\E[\wy \circ \mu]{\left| \pr{2\theta_m^2\widetilde{x} + 2Y} \right|^{i-1}(1+\widetilde{x})^s e^{\frac{s\widetilde{x}Y}{1+\widetilde{x}}}} < \infty.
\end{align}
\item Since for all $(x, y)\in \calX \times \calY$, it holds that $0<\frac{p_x(y)}{(\wy\circ\mu)(y)} <\infty$, $\frac{\partial f}{\partial s}$ exists, and we have
\begin{align}
\frac{\partial f}{\partial s}(s, x, y) = \log^{i} \pr{ \frac{p_x(y)}{(\wy\circ \mu)(y)}}\pr{\frac{p_x(y)}{(\wy\circ \mu)(y)}}^s.
\end{align}
\item Similar to the first part, we can upper-bound the partial derivative as
\begin{align}
\left|\frac{\partial f}{\partial s}(s, x, y)\right|
&= \left|  \log^{i} \pr{ \frac{p_x(y)}{(\wy\circ \mu)(y)}}\pr{\frac{p_x(y)}{(\wy\circ \mu)(y)}}^s\right|\\
&\leq  \left| \pr{2\theta_m^2\widetilde{x} + 2y} \right|^{i}(1+\widetilde{x})^s e^{\frac{s\widetilde{x}y}{1+\widetilde{x}}}
\end{align}
The above bound is increasing in $s$. Thus, by choosing $\widetilde{s}$ small enough such that the expectation is finite for $s = \widetilde{s}$, we can choose
\begin{align}
\theta(x, y) \eqdef  \left| \pr{2\theta_m^2\widetilde{x} + 2y} \right|^{i}(1+\widetilde{x})^{\widetilde{s}} e^{\frac{\widetilde{s}\widetilde{x}y}{1+\widetilde{x}}}.
\end{align}
Then, $\E[\wy\times \mu]{\theta(X, Y)} < \infty$ and for all $s\leq \widetilde{s}$, we have $\left|\frac{\partial f}{\partial s}(s, x, y)\right| \leq \theta(x, y)$.
\end{enumerate}
We can now use Theorem~\ref{th:leibniz} and obtain
\begin{align}
&\frac{\partial }{\partial s}\E[\wy \times \mu]{\log^{i-1} \pr{ \frac{p_X(Y)}{(\wy\circ \mu)(Y)}}\pr{\frac{p_X(Y)}{(\wy\circ \mu)(Y)}}^s} \\
&= \frac{\partial }{\partial s}\E[\wy \times \mu]{f(s, X, Y)}\\
&=\E[\wy \times \mu]{ \frac{\partial }{\partial s}f(s, X, Y)}\\
&= \E[\wy \times \mu]{\log^{i} \pr{ \frac{p_X(Y)}{(\wy\circ \mu)(Y)}}\pr{\frac{p_X(Y)}{(\wy\circ \mu)(Y)}}^s}.
\end{align}
Therefore, the induction hypothesis implies \eqref{eq:deriv-g}. By the chain rule, $\phi_{\text{rel}}(s, \mu)$ is also a smooth function on an interval $[0, \widetilde{s}]$ for all $\mu$ with $\sup(\supp{\mu}) \leq \widetilde{x}$. Hence, we can use Taylor's theorem to obtain
\begin{align}
\phi(s, \mu) = \phi_{\text{rel}}(0, \mu) + \frac{\partial \phi_{\text{rel}}}{\partial s}(0, \mu)s +  \frac{\partial^2 \phi_{\text{rel}}}{\partial s^2}(0, \mu)\frac{s^2}{2} +   \frac{\partial^3 \phi_{\text{rel}}}{\partial s^3}(\eta, \mu)\frac{s^3}{6},
\end{align}
for some $\eta \in [0, \widetilde{s}]$. The derivatives of $\phi_{\text{rel}}$ would be
\begin{align}
\phi_{\text{rel}}(0, \mu)  &= - \log(g(0, \mu)) = 0\\
\frac{\partial \phi_{\text{rel}}}{\partial s}(0, \mu) &= -\frac{\frac{\partial g}{\partial s}(0, \mu)}{g(0 ,\mu)} = -\E[\wy \times \mu]{\log \frac{p_X(Y)}{(\wy \circ \mu)(Y)}} = -I(\mu, \wy)\\
\frac{\partial^2 \phi_{\text{rel}}}{\partial s^2}(0, \mu) &= -\frac{g(0, \mu)\frac{\partial^2 g}{\partial s^2}(0, \mu) -\pr{\frac{\partial g}{\partial s}(0, \mu)}^2}{g(0 ,\mu)} = -\E[\wy \times \mu]{\log^2 \frac{p_X(Y)}{(\wy \circ \mu)(Y)}} + I(\mu, \wy)^2.
\end{align}
Moreover, Lemma~\ref{lm:density-2m} yields that
\begin{align}
\frac{\partial^2 \phi_{\text{rel}}}{\partial s^2}(0, \mu)  
&\geq -2(3+\widetilde{x})(2\sqrt{\nu} + \nu)(1+\theta_m^2\widetilde{x})^4 \pr{e^{A\widetilde{x}}A + \theta_m^2}^2e^{2A} + 20 \pr{(1+\theta_m^2\widetilde{x})  + A }^2e^{-A}\\
&\geq -B_1\pr{(2\sqrt{\nu} + \nu)e^{2A\widetilde{x} + 2A}A^2 + A^2 e^{-A}},
\end{align}
for some $B_1$ depending on $\theta_m^2$ and $\widetilde{x}$. With similar arguments as we had to check the third condition of Theorem~\ref{th:leibniz}, we can prove that there exists $B_2$ depending on $\widetilde{x}$, such that for all $\eta\in[0, \widetilde{s}]$, we have 
\begin{align}
\left|\frac{\partial^3 \phi_{\text{rel}}}{\partial s^3}(\eta, \mu) \right| \leq B_2.
\end{align}
Choosing $B = \max(B_1/2, B_2/6)$ completes the proof.
\end{proof}
\section{Proof of Lemma~\ref{lm:chi-2-bound} and  \ref{lm:d-upper}}
\label{sec:chi-2-bound}
\modfirst{
We first introduce some notation and facts, which will be useful in both proofs. Let $f(z) \eqdef (\wz\circ\mu)(z)$ and $\phi(z) \eqdef f(z) /q_0(z) -1$. Defining $P_X$ as the associated \ac{PMF} of $\mu$, we can write $\phi(z) = \sum_x P_X(x) \frac{e^{\frac{xz}{1+x}}}{1+x} - 1$, which is increasing and
\begin{align}
\phi(z) \geq \phi(0) = \E[\mu]{\frac{1}{1+X}} - 1 \geq -\E[\mu]{X} \geq - 2\sqrt{\nu} - \nu \geq -0.5.
\end{align}
Furthermore, there exists a unique $M_0$ such that $\phi(z) \leq 0$ if and only if $z\leq M_0$.

\begin{proof}[Proof of Lemma~\ref{lm:d-upper}]
Using the bound $\log(1+x) \leq x  -x^2/2+x^3/3$ for $x>-1$, we obtain
\begin{align}
\D{\wz \circ \mu}{q_0}
&= \int_0^{\infty} q_0(1+\phi)\log(1+\phi)\\
&\leq \int_0^M q_0 \phi + \frac{1}{2} \int_0^M q_0 \phi^2 - \frac{1}{6} \int_0^M q_0\phi^3 + \frac{1}{3} \int_0^M q_0 \phi^4 + \int_M^\infty f \log( f/q_0).
\end{align}
We consider each term separately.
\begin{enumerate}
\item We have

\begin{align}
\int_0^M q_0 \phi
&= \int_0^M f - \int_0^Mq_0\\
&= e^{-M} - \sum_x P_X(x) e^{-\frac{M}{1+x}}\\
&\leq 0.
\end{align}
\item We have

\begin{align}
 \frac{1}{2} \int_0^M q_0 \phi^2 
 &\leq  \frac{1}{2} \int_0^\infty q_0 \phi^2\\
 &= \frac{1}{2}\chi_2(\wz\circ \mu \|q_0).
\end{align}
\item We have
\begin{align}
 -\frac{1}{6} \int_0^M q_0 \phi^3
 &\leq  -\frac{1}{6} \int_0^{M_0} q_0 \phi^3\\
  &\leq  \frac{1}{6}M_0(\E[\mu]{X})^3 \\
  &\stackrel{(a)}{\leq} \frac{1}{3}(\E[\mu]{X})^3,
\end{align}
where $(a)$ follows since $M_0 \leq 2$ by the argument in the proof of Lemma~\ref{lm:chi-2-bound}.
\item We have
\begin{align}
\int_0^M q_0\phi^4
&= \int_0^{M_0} q_0\phi^4 + \int_{M_0}^M q_0\phi^4\\
&\leq M_0 (\E[\mu]{X})^4 + \int_{M_0}^M e^{-z} \pr{\E[\mu]{X} e^{\frac{az}{1+a}}}^4dz\\
&= (\E[\mu]{X})^4 \pr{2 + \int_0^M e^{z\pr{-1 + \frac{4a}{1+a}}}dz}.
\end{align}
\item We have
\begin{align}
\int_M^\infty f\log(f/q_0)
&\leq \frac{a}{1+a} \int_M^\infty f(z)z dz\\
&\leq \frac{a}{1+a}\pr{\int_M^\infty e^{-\frac{z}{1+\epsilon}}z dz + \frac{1}{\E[\mu]{X}+\epsilon} \int_M^\infty e^{-\frac{z}{1+a}} dz}.
\end{align}
\end{enumerate}
\end{proof}
\begin{proof}[Lemma~\ref{lm:chi-2-bound}]
We use the notations introduced in the beginning of Appendix~\ref{sec:chi-2-bound}. 
Since  we have $\log(1+x) \geq x - x^2/2 +\indic{x \leq 0}2x^3/3$ for  $x \geq -0.5$, we have for all $z\in \calZ$,
\begin{align}
f(z) \log(\phi(z) + 1)  \geq f(z)\pr{\phi(z) -\phi^2(z)/2 \indic{\phi(z) \leq 0}2\phi(z)^3/3}.
\end{align}
We therefore obtain
\begin{align}
\D{f}{q_0} 
&= \int_0^\infty f \log(\phi+1)\\
&\geq  \int_0^M f\pr{\phi -\phi^2/2} + \int_0^{M_0} 2f \phi^3/3 + \int_M^\infty f\log\pr{f/ q_0}.
\end{align}
We consider each term separately in the following.
\begin{enumerate}
\item We have 
\begin{align}
 \int_0^M f\pr{\phi - \phi^2/2}
 &=  \int_0^M q_0(\phi+1)\pr{\phi - \phi^2/2} \\
 &= \int_0^M q_0 \phi^2/2+ \int_0^M q_0 \phi - \int_0^Mq_0\phi^3/2\\
 &=  \int_0^\infty q_0 \phi^2/2+ \int_0^M q_0 \phi - \int_0^Mq_0\phi^3/2 - \int_M^\infty q_0 \phi^2/2.
\end{align}
We again separately lower-bound each term in the above expression.
\begin{enumerate}
\item We have by definition, 

\begin{align}
 \int_0^\infty q_0 \phi^2/2 = \frac{1}{2}\chi_2(f\|q_0).
\end{align}
\item We have
\begin{align}
  \int_0^M q_0 \phi  
  &= \int_0^M (f-q_0)\\
  &\geq -\int_M^\infty f.
\end{align}
\item To lower-bound $- \int_0^Mq_0\phi^3/2$, we first upper-bound $\phi$ as follows.
\begin{align}
\phi(z) 
&= \sum_x P_X(x) \frac{ e^{\frac{xz}{1+x}}}{1+x} -1\\
&\leq \sum_x P_X(x) { e^{\frac{xz}{1+a}}} -1\\
&\stackrel{(a)}{\leq} \sum_x P_X(x) \pr{ 1 + \pr{e^{\frac{az}{1+a}} - 1} x} -1\\
&=  \pr{e^{\frac{az}{1+a}} - 1} \E[\mu]{X}\\
&= e^{\frac{az}{1+a}} \E[\mu]{X},
\end{align}
where $(a)$ follows since $e^{\frac{xz}{1+a}} \leq 1 + \pr{e^{\frac{az}{1+a}} - 1} x$ for $x\in [0, a]$. Since $q_0 > 0$ and $x\mapsto x^3$ is increasing, we have
\begin{align}
\int_0^{M} q_0\phi^3 
&\leq \int_0^M e^{-z} \pr{e^{\frac{az}{1+a}} \E[\mu]{X}}^3 dz\\
&= (\E[\mu]{X})^3 \int_0^M e^{z\pr{-1 + \frac{3a}{1+a}}}dz.
\end{align}
\item Since $\phi(z) \geq 0$ for $z\geq M \geq M_0$ and $x\mapsto x^2$ is increasing for $x\geq 0$, we have
\begin{align}
\int_M^\infty q_0\phi^2
& \leq \int_M^\infty e^{-z} \pr{e^{\frac{az}{1+a}} \E[\mu]{X}}^2 dz\\
& = (\E[\mu]{X})^2  \int_0^M e^{z\pr{-1+\frac{2a}{1+a}}}dz.
\end{align}
\end{enumerate}
As a conclusion, we obtain that
\begin{align}
 \int_0^M f\pr{\phi - \phi^2/2} 
&\geq \frac{1}{2}\chi_2(f\|q_0) - \frac{1}{2} (\E[\mu]{X})^2   \int_0^M e^{z\pr{-1+\frac{2a}{1+a}}}dz - \frac{1}{2} (\E[\mu]{X})^3 \int_0^M e^{z\pr{-1 + \frac{3a}{1+a}}}dz - \int_M^\infty f.
\end{align}

\item Using $|f(z)| \leq 1$ for all $z\in \calZ$ and $0\geq \phi(z) \geq -\E[\mu]{X}$ for all $0\leq z\leq M_0$, we have

\begin{align}
\int_0^{M_0} f \phi^3 \geq - M_0 \pr{\E[\mu]{X}}^3.
\end{align}
We now show that $M_0 \leq 2$, for which it is enough to show that $ f(2) \geq q_0(2)$. Note that
\begin{align}
\log f(2) 
&\stackrel{(a)}{\geq} \sum_x P_X(x) \log q_x(2) \\
&= - \E[\mu]{\log\pr{1+X}} - 2 \E[\mu]{\frac{1}{1+X}}\\
&\geq \E[\mu]{X} - 2 + 2\E[\mu]{\frac{X}{1+X}}\\
&\geq \E[\mu]{X} - 2 + 2\E[\mu]{\frac{X}{1+a}}\\
&\geq -2 = \log q_0(2).
\end{align}
\item By our assumption that $f(M)/q_0(M) \geq e$, we have
\begin{align}
\int_M^\infty f\log f/q_0 \geq \int_M^\infty f.
\end{align}

\end{enumerate}
Combining the bounds in the above three parts, we obtain the desired result.
\end{proof}}
\section{Optimization Problem in \eqref{eq:optimation-a-nu}}
\label{sec:optim-probl-eqref}
\subsection{Prokhorov's Theorem}
\begin{theorem}
\label{th:prokhorov}
Let $\{\mu_n\}$ be a sequence of tight probability measures on $\mathbb{R}$, i.e., for all $\epsilon > 0$, there exists a compact set $K\subset \mathbb{R}$ such that for all $n\geq 1$, $\mu_n(\mathbb{R}\setminus K) \leq \epsilon$. Then, there exists a sub-sequence $\{\mu_{n_k}\}_{k\geq 1}$ and another probability measure $\mu$ on $\mathbb{R}$ such that $\{\mu_{n_k}\}_{k\geq 1}$ converges weakly to $\mu$.
\end{theorem}
\subsection{Convex Optimization for General Vector Spaces}
\begin{theorem}\modfirst{(\cite[Theorem 1, Page 217]{luenberger1997optimization}).}
Let $\calV$ be a vector space, $\Omega\subset \calV$ a convex set,  $\calU$ be a normed vector space, and $\calP\subset\calU$ be a positive cone, i.e., for all $u_1, u_2\in\calU$ and all $\alpha, \beta \geq 0$, we have $\alpha u_1+\beta u_2 \in \calP$. Suppose the interior of $\calP$ is non-empty, and $\phi:\Omega \to \mathbb{R}$ and $G:\Omega \to \calU$ are convex functions such that there exists $\omega_1 \in \Omega$ for which $G(\omega_1) \prec_{\calP} 0$ and $A \eqdef \inf_{\omega \in \Omega: G(\omega) \preceq_{\calP} 0} \modfirst{\phi}(\omega) > -\infty$.  Then, there exists $u_0^*\succeq_{\calP^*} 0$ in  $\calU^*$ such that $A = \inf_{\omega \in \Omega} \phi(\omega) + \langle G(\omega), u_0^* \rangle$. Moreover, if $\omega_0$ is a solution to  the first optimization problem, the infimum of the second optimization problem is also achieved by $\omega_0$ and $\langle G(\omega_0), u_0^*\rangle = 0$.
\label{th:lagrange1}
\end{theorem}
 We next recall a result from \cite{Smith1969} to find an expression for the KKT conditions of an abstract convex optimization. To this end, we introduce the notation of weak differentiablity for a function $f:\Omega \to \mathbb{R}$ where $\Omega$ is convex. We say that $f'_{\omega_0}:\Omega \to \mathbb{R}$ is the weak derivative of $f$ at $\omega_0$, if
\begin{align}
\label{eq:weak-dev}
f'_{\omega_0}(\omega) = \lim_{\theta \to 0^+} \frac{f(\theta \omega + (1-\theta) \omega_0)}{\theta}.
\end{align} 
\begin{theorem}[\cite{Smith1969}]
\label{th:kkt-general}
Let $\calV$ be a linear space, $\Omega\subset \calV$ be convex, and $f:\Omega \to \mathbb{R}$ be convex and have weak derivative for all $\omega\in\Omega$. $f(\omega^*) = \inf_{\omega\in\Omega} f(\omega)$ if and only if for all $\omega\in\Omega$, we have $f'_{\omega^*}(\omega) \geq 0$.
\end{theorem}

\subsection{Technical Results}
\label{sec:technical-results}

\begin{lemma}
\label{lm:a-prop}
$A(\nu)$ defined in \eqref{eq:optimation-a-nu} satisfies the following properties.
\begin{enumerate}
\item It is concave and non-decreasing on $[0, \infty)$.
\item It is continuous on $[0, \infty)$.
\item The one-sided derivatives,
\begin{align}
A'(\nu^+) \eqdef \lim_{h\to 0^+} \frac{A(\nu+h) - A(\nu)}{h} \text{ and } A'(\nu^-) \eqdef \lim_{h\to 0^+} \frac{A(\nu) - A(\nu - h)}{h},
\end{align}
exist for all $\nu > 0$, and for all $0 < \nu_1 < \nu_2$, we have $A'(\nu_1^-) \geq A'(\nu_1^+) \geq A'(\nu_2^-) \geq A'(\nu_2^+)$.
\item There exist \modfirst{ constants $\nu_0 > 0$ and $C > 0$} such that for all $0<\nu\leq \nu_0$, we have $A(\nu) \geq C\sqrt{\nu}$.
\item We have $\lim_{\nu\to 0^+} A'(\nu^+) = \lim_{\nu\to 0^+} A'(\nu^-) = \infty$.
\end{enumerate}
\end{lemma}
\begin{proof}
\begin{enumerate}

\item By definition of $A(\nu)$, it follows that $A(\nu)$ is non-decreasing. To check concavity, we take any $\nu_1, \nu_2 > 0$, $\mu_1, \mu_2\in\Omega$ with $\D{\wz\circ\mu_1}{q_0} \leq \nu_1$ and $\D{\wz\circ\mu_2}{q_0} \leq \nu_2$, and $\lambda\in[0, 1]$. By convexity of the relative entropy, we have
\begin{align}
\D{\wz\circ\pr{\lambda \mu_1 + (1-\lambda) \mu_2}}{q_0} \leq \lambda\nu_1 + (1-\lambda)\nu_2.
\end{align}
Therefore, by concavity of the mutual information,
\begin{align}
A(\lambda\nu_1 + (1-\lambda)\nu_2)
& \geq I(\lambda \mu^*_{\nu_1} + (1-\lambda) \mu^*_{\nu_2}, \wy) \\
&\geq \lambda I(\mu_1, \wy) + (1-\lambda) I( \mu_2, \wy).
\end{align}
Hence, by definition of supremum, we have
\begin{align}
\displaybreak[0]A(\lambda\nu_1 + (1-\lambda)\nu_2) 
\displaybreak[0]&\geq \sup_{\mu_1, \mu_2 \in \Omega:\D{\wz\circ\mu_1}{q_0} \leq \nu_1, \D{\wz\circ\mu_2}{q_0} \leq \nu_2 } \lambda I(\mu_1, \wy) + (1-\lambda) I( \mu_2, \wy)\\
\displaybreak[0]&= \lambda \sup_{\mu_1\in \Omega:\D{\wz\circ\mu_1}{q_0} \leq \nu_1} I(\mu_1, \wy) + (1-\lambda)\sup_{\mu_2 \in \Omega: \D{\wz\circ\mu_2}{q_0} \leq \nu_2 } I( \mu_2, \wy)\\
\displaybreak[0]&= \lambda A(\nu_1) + (1-\lambda) A(\nu_2).
\end{align}
\item Since $A(\nu)$ is concave on $[0, \infty)$, it is continuous on $(0, \infty)$ \cite[Page 153, Problem 4]{stein2009real}. To check the continuity at $0$, we consider $\nu > 0$ and $\mu\in \Omega$ with $\D{\wz\circ \mu}{q_0} \leq \nu$. Using \eqref{eq:mutual-exp}, we have
\begin{align}
I(\mu, \wy) 
&= -\int_0^\infty (\wy\circ \mu)(y) \log (\wy\circ \mu)(y) dy - \E[\mu]{\log(1+\theta_m^2X)} - 1.
\end{align}
Furthermore, since $\E[\wy\circ \mu]{Y} = 1 + \theta_m^2 \modfirst{\E[\mu]{X}}$ by \eqref{eq:expected-y} and the support of $\wy\circ \mu$ is included in $[0, \infty)$, the differential entropy of $\wy\circ \mu$ is upper-bounded by the differential entropy of an exponential distribution with the same mean \cite{park2009maximum}. Therefore, we have
\begin{align}
I(\mu, \wy) 
&\leq 1 + \log (1 + \theta_m^2 \E[\mu]{X})  -  \E[\mu]{\log(1+\theta_m^2X)} - 1\\
&\leq \theta_m^2 \E[\mu]{X}.
\end{align}
\modfirst{Furthermore, we have
\begin{align}
\nu &\geq\D{\wz\circ\mu}{q_0}\\
&=\int_0^\infty (\wz\circ\mu)(z) \log \frac{(\wz\circ\mu)(z)}{q_0(z)} dz\\ 
& \stackrel{(a)}{=} \int_0^\infty (\wz\circ\mu)(z) \log ((\wz\circ\mu)(z)) dz + \E[\wz\circ\mu]{Z}\\
&\geq -1 - \log(\E[\wz\circ\mu]{Z}) +  \E[\wz\circ\mu]{Z}\\
&\stackrel{(b)}{=} -\log(1+\E[\mu]{X}) + \E[\mu]{X}\\
&\stackrel{(c)}{\geq} \frac{1}{2}\E[\mu]{X}^2 - \frac{1}{3}\E[\mu]{X}^3\\
&\stackrel{(d)}{\geq} \frac{1}{2}\E[\mu]{X}^2 \pr{1 - \frac{2}{3}\pr{2\sqrt{\nu} + \nu}}
\end{align}
where $(a)$ follows since $\log q_0(z) = -z$ and $\E[\wz\circ\mu]{Z} < \infty$ by Lemma~\ref{lm:mean-finite} and \eqref{eq:expected-z},  $(b)$ follows from \eqref{eq:expected-z},  $(c)$ follows from $\log(1+x) \leq x -x^2/2 + x^3/3$ for $x>-1$, and $(d)$ follows from Lemma~\ref{lm:expected-x-bound}. We obtain for $\nu< 1/4$ that $\E[\nu]{X} \leq \frac{\sqrt{2\nu}}{1 - (2/3)(2\sqrt{\nu}+\sqrt{\nu)}} \leq  \frac{\sqrt{2\nu}}{1 -2\sqrt{\nu}} $, and hence,
\begin{align}
I(\mu, \wy)  \leq  \frac{\theta_m^2 \sqrt{2\nu}}{1-2\sqrt{\nu}}.\label{eq:upper-a}
\end{align}}
Additionally, since $A(\nu)$ is non-decreasing and non-negative, we have
\begin{align}
|A(\nu) - A(0)| 
= A(\nu) - A(0)
\leq A(\nu) 
\leq \modfirst{\frac{\theta_m^2\sqrt{2\nu}}{1-2\nu}},
\end{align}
which implies that $A(\nu)$ is continuous at zero.
\item Follows from  \cite[Page 153, Problem 4]{stein2009real} and concavity of $A(\nu)$.
\item For $\nu>0$ small enough, it is enough to find a probability measure $\mu$ satisfying  $\D{\wz\circ \mu}{q_0} \leq \nu$ and $I(\mu, \wy) \geq C\sqrt{\nu}$. Let $\mu$ be a discrete probability measure on $\calX$ with two mass points at $0$ and $\widetilde{x}$ with probabilities $1-\alpha$ and $\alpha$, respectively, such that $\widetilde{x}< \min(1, 1/\theta_m^2)$. Then, by Lemma~\ref{lm:d-two-points},
\begin{align}
\D{\wz\circ \mu}{q_0} = \frac{\alpha^2\td{x}^2}{2(1-\widetilde{x}^2)} + o(\alpha^2).
\end{align}            
Similarly, we can obtain  $\D{\wy\circ\mu}{p_0} \leq  \alpha^2\theta_m^2\td{x}^2/(2(1-\theta_m^2\widetilde{x}^2)) + o(\alpha^2)$. Therefore, we can lower-bound the mutual information by
\begin{align}
I(\mu, \wy)
&= \alpha\D{p_{\widetilde{x}}}{p_0} - \D{\wy\circ\mu}{p_0} \\
& \geq \alpha\D{p_{\widetilde{x}}}{p_0} -   \frac{\alpha^2\theta_m^2\td{x}^2}{2(1-\theta_m^2\widetilde{x}^2)} -o(\alpha^2)\\
&= \alpha \pr{\theta_m^2 \widetilde{x} - \log(1+\theta_m^2\widetilde{x})} -     \frac{\alpha^2\theta_m^2\td{x}^2}{2(1-\theta_m^2\widetilde{x}^2)} -o(\alpha^2).                     
\end{align}
Hence, by choosing $\alpha = \modfirst{ \td{x}^{-1}\sqrt{2(1-\widetilde{x}^2)\D{\wz\circ \mu}{q_0}  }} =  \td{x}^{-1}\sqrt{2(1-\widetilde{x}^2)\nu(1-o(1))}$, we have $\D{\wz\circ \mu}{q_0}  \leq \nu$ and 
\begin{align}
\label{eq:lower-aa}
I(\mu, \wy) \geq \sqrt{\nu}(1-o(1)) \pr{  \td{x}^{-1}\sqrt{2(1-\widetilde{x}^2)}\pr{\theta_m^2 \widetilde{x} - \log(1+\theta_m^2\widetilde{x})} - \sqrt{\nu}\frac{1-\widetilde{x}^2}{1-\theta_m^2\widetilde{x}^2  }}.
\end{align}
Choosing $\nu_0>0$ such that
\begin{align}
\td{x}^{-1} \sqrt{2(1-\widetilde{x}^2)}\pr{\theta_m^2 \widetilde{x} - \log(1+\theta_m^2\widetilde{x})} \geq \modfirst{2}\sqrt{\nu_0}\frac{1-\widetilde{x}^2 }{1-\theta_m^2\widetilde{x}^2}, 
\end{align}
the claim of the lemma holds for
\begin{align}
C = \frac{1}{2}\td{x}^{-1} \sqrt{2(1-\widetilde{x}^2)}\pr{\theta_m^2 \widetilde{x} - \log(1+\theta_m^2\widetilde{x})}.
\end{align}
\item Since $A'(\nu^+) \leq A'(\nu^-)$, we only need to compute $\lim_{\nu \to 0^+}A'(\nu^+)$. Since $A(\nu)$ is concave $A'(\nu^+)$ is decreasing, and therefore, it is enough to show that for any $L>0$ there exists some $\nu>0$ with $A'(\nu^+) \geq L$. To this end, we fix some $\widetilde{\nu} > 0$ and define $B(\nu) \eqdef A(\nu) - \frac{A(\widetilde{\nu})}{\widetilde{\nu}}\nu$. $B(\nu)$ is continuous on $[0, \widetilde{\nu}]$ and therefore it achieves its maximum and minimum on $[0, \widetilde{\nu}]$. Hence, either we have $B(\nu) = 0$ for all $\nu \in [0, \widetilde{\nu}]$ or there exists a $\nu\in(0, \widetilde{\nu})$ such that $B(\nu)$ achieves its maximum or minimum at $\nu$. Then, we should have $B'(\nu-) = A'(\nu^-) - A(\widetilde{\nu}) / \widetilde{\nu} \geq 0$ or $B'(\nu^+) = A'(\nu^+) - A(\widetilde{\nu}) / \widetilde{\nu} \geq 0$. In both cases, we have $A'\pr{\frac{\nu}{2}^+} \geq \frac{A(\widetilde{\nu})}{\widetilde{\nu}}$. However, by Lemma~\ref{lm:a-prop}, $A'\pr{\frac{\nu}{2}^+}  \geq C/\sqrt{\widetilde{\nu}}$, if $\widetilde{\nu} \leq \nu_0$. Since $\widetilde{\nu}$ is arbitrary, we can choose it such that $C/\sqrt{\widetilde{\nu}} > L$.
\end{enumerate}
\end{proof}
\begin{proof}[Proof of Lemma~\ref{lm:exist-unique}]
We  only prove  the existence of a solution and the uniqueness follows from  strict concavity of the mutual information \cite{Abou-Faycal2001}. Consider a sequence $\{\mu_n\}_{n\geq 1}$ in $\Omega$ such that $\D{\wz\circ \mu_n}{q_0} \leq \nu$ and $\lim_{n\to\infty}I(\mu_n, \wy) = A(\nu)$. To use \ref{th:prokhorov}, we first check that this sequence is tight. For any $\epsilon>0$, we have
\begin{align}
\modfirst{ \P[\mu_n]{X \notin [0, (2\sqrt{\nu}+\nu)/\epsilon]}}
&\stackrel{(a)}{\leq} \frac{\E[\mu_n]{X}\epsilon}{{2\sqrt{\nu} + \nu}}\\
&\stackrel{(b)}{\leq} \epsilon,
\end{align}
where $(a)$ follows from \modfirst{applying Markov's inequality to the almost surely non-negative random variable $X$}, and $(b)$ follows from Lemma~\ref{lm:expected-x-bound}. Since $[0, (2\sqrt{\nu}+\nu)/\epsilon]$ is compact, the sequence $\{\mu_n\}_{n\geq 1}$ is tight. Therefore, we are permitted to use Theorem~\ref{th:prokhorov} that shows the existence of a subsequence $\{\mu_{n_k}\}_{k\geq1}$ and probability measure $\mu$ on $\mathbb{R}$ such that $\{\mu_{n_k}\}_{k\geq1}$ converges weakly to $\mu$. We claim that $\mu^*_{\nu}$ is indeed $\mu$ and prove it in three steps.

\textbf{Step 1:} Theorem~\ref{th:prokhorov} only guarantees  the existence of a probability measure on $\mathbb{R}$ which can possibly have positive measure on negative numbers. In this step, we show that this is not the case. By the Portmanteau theorem, the weak convergence of $\{\mu_{n_k}\}_{k\geq1}$ to $\mu$ implies that $\liminf_{k\to \infty} \mu_{n_k}(U)\geq \mu(U)$ for any open set $U\subset \mathbb{R}$. Taking $U = ]-\infty,0[$, we obtain that
\begin{align}
0 = \liminf_{k\to\infty}\mu_{n_k}(]-\infty, 0[) \geq \mu(]-\infty, 0[) \geq 0,
\end{align}   
which means that $\mu(]-\infty, 0[) = 0$.

\textbf{Step 2:} In this step we prove that $\mu$ satisfies the optimization constraint, i.e., $\D{\wz\circ \mu}{q_0}\leq \nu$. Let us define $f_k(z) \eqdef (\wz\circ \mu_{n_k})(z)$ and $f(z)\eqdef (\wz\circ\mu)(z)$. Since for any $z\in\calZ$, $q_x(z) = e^{-z/(1+x)}/(1+x)$ is a continuous and bounded function in $x$, by weak convergence definition, we have 
\begin{align}
f_k(z) = \E[\mu_{n_k}]{q_X(z)}  \to \E[\mu]{q_X(z)} =f(z).
\end{align}
In the next lemma, we show that $|f_k(z)\log f_k(z)|$ is uniformly upper-bounded by an integrable function.
\begin{lemma}
\label{lm:bound-f-k}
 There exists some $\widetilde{z}$ such that for all $k$,
\begin{align}
|f_k(z)\log f_k(z)| \leq g(z) \eqdef \begin{cases}e^{-1}\quad &z\in[0, \widetilde{z}],\\ \frac{2\sqrt{\nu} + \nu}{e(z^{\frac{3}{2}}-z^{\frac{1}{2}} )} + z^{-\frac{1}{2}}e^{-\sqrt{z}}\quad &z\in[\widetilde{z}, \infty[,\end{cases}
\end{align}
and $\int_0^\infty |g(z)| dz < \infty$.
\end{lemma}
\begin{proof}
Note first that for all $x\in[0, 1]$, we have $|x\log x|\leq e^{-1}$, and for all $x\in[0, e^{-1}]$, we have $|x\log x| \leq |x|$. Thus, it is enough to show that there exist $\widetilde{z}$ such that for all $k\geq 1$ and $z\geq \widetilde{z}$, 
\begin{align}
f_k(z) \leq \frac{2\sqrt{\nu} + \nu}{e(z^{\frac{3}{2}}-z^{\frac{1}{2}} )} + z^{-\frac{1}{2}}e^{-\sqrt{z}}.
\end{align}
By law of total probability, for all $\lambda > 0$, we have
\begin{align}
f_k(z) 
&\eqdef \E[\mu_{n_k}]{q_X(z)} \\
&= \E[\mu_{n_k}]{q_X(z) | X \geq \lambda}\P[\mu_{n_k}]{X\geq \lambda} + \E[\mu_{n_k}]{q_X(z) | X < \lambda}\P[\mu_{n_k}]{X< \lambda}\\
&\stackrel{(a)}{\leq} \E[\mu_{n_k}]{q_X(z) | X \geq \lambda}\frac{\E[\mu_{n_k}]{X}}{\lambda} + \E[\mu_{n_k}]{q_X(z) | X < \lambda}\\
&\stackrel{(b)}{\leq}  \E[\mu_{n_k}]{q_X(z) | X \geq \lambda}\frac{2\sqrt{\nu} + \nu}{\lambda} + \E[\mu_{n_k}]{q_X(z) | X < \lambda},
\label{eq:f-k-app}
\end{align}
where $(a)$ follows from Markov's inequality, and $(b)$ follows from Lemma~\ref{lm:expected-x-bound}.
We also have for all $z\geq 1$, $q_x(z) \leq (ze)^{-1}$, and for all $0\leq x \leq \lambda \leq z-1$, $q_x(z) \leq e^{-\frac{z}{1+\lambda}}/(1+\lambda)$. Substituting these upper-bounds in \eqref{eq:f-k-app} for $\lambda =z^{\frac{1}{2}} - 1$, which is less than $z-1$ for $z \geq 1$, we obtain
\begin{align}
f_k(z)  
&\leq \frac{1}{ze}\frac{2\sqrt{\nu} + \nu}{z^{\frac{1}{2}} - 1} + \frac{1}{z^{\frac{1}{2}}}e^{-z^{\frac{1}{2}}}\\
&=\frac{2\sqrt{\nu} + \nu}{e(z^{\frac{3}{2}}-z^{\frac{1}{2}} )} + z^{-\frac{1}{2}}e^{-\sqrt{z}}.
\end{align}
\end{proof}

We are now eligible to use dominated convergence theorem and exchange limit and integral to obtain
\begin{align}
\lim_{k\to \infty} \int_0^\infty f_k(z) \log f_k(z) dz 
&=  \int_0^\infty \lim_{k\to \infty} f_k(z) \log f_k(z) dz \\
&= \int_0^\infty f(z)\log f(z) dz.\label{eq:step2-1}
\end{align}
Since $f_k(z)z\geq 0$ for all $z\in\calZ$ and $k\geq 1$, Fatou's lemma yields that
\begin{align}
\int_0^{\infty} f(z) z dz 
&= \int_0^{\infty} \liminf_{k\to\infty} f_k(z)z dz \\
&\leq \liminf_{k\to \infty}  \int_0^{\infty} f_k(z)z dz.\label{eq:step2-2}
\end{align}
Combing \eqref{eq:step2-1} and \eqref{eq:step2-2}, we have
\begin{align}
\D{\wz \circ \mu }{q_0}&=\int_0^{\infty} f(z)\pr{z + \log f(z)} dz
\\
&\leq \liminf_{k\to\infty} \int_{0}^\infty f_k(z) \pr{z + \log f_k(z)} dz = \liminf_{k\to\infty} \D{\wz\circ\mu_{n_k}}{q_0}\leq \nu.
\end{align}
\textbf{Step 3}: It remains to show that $I(\mu, \wy) \geq A(\nu)$. We again define $h_k(z) \eqdef (\wy\circ \mu_{n_k})(z)$ and $h(z) \eqdef (\wy \circ \mu)(z)$. 
 With the same argument of the previous step, we can prove that
\begin{align}
\lim_{k\to \infty} \int_{0}^\infty h_k(z) \log h_k(z) dz = \int_0^\infty h(z) \log h(z) dz.
\end{align}
Furthermore, by \cite[Page 86]{durrett2010probability}, we have
\begin{align}
\liminf_{k\to \infty}\E[\mu_{n_k}]{1+\log(1+\theta_m^2X)} \geq \E[\mu]{1+\log(1+\theta_m^2X)}.
\end{align}
Hence, \eqref{eq:mutual-exp} implies that $I(\mu, \wy) \geq A(\modfirst{\nu})$.
\end{proof}

\begin{proof}[Proof of Theorem~\ref{th:kkt}]
We prove all four statements in order. The proof heavily  relies  on results from convex optimization for general vector spaces and properties of the optimization problem in~\eqref{eq:lagrange-opt}, which we have gathered in Appendix~\ref{sec:optim-probl-eqref} for the reader's convenience.
\begin{enumerate}
\item In Theorem~\ref{th:lagrange1}, taking $\Omega$ as the set of all probability measures $\mu$ on $\calX$ with $\D{\wz\circ\mu}{q_0} < \infty$, $\calU = \mathbb{R}$, $\calP = \mathbb{R}^{+}$, $\phi(\mu) = -I(\mu, \wy)$, $G(\mu) = \D{\wz\circ\mu}{q_0} - \nu$, we note that 
\begin{align}
-\infty<-A(\nu) 
&= -\sup_{\mu\in \Omega:G(\mu) \leq 0 } -\phi(\mu)\\
&= \inf_{\mu\in \Omega:G(\mu) \leq 0 } \phi(\mu).
\end{align}
By convexity of the relative entropy and concavity of mutual information in the input distribution,  $\phi$ and $G$ are convex functions, with $\mu_1$  the deterministic probability measure with all mass point at zero, we also have $G(\mu_1) = -\nu < 0$. Therefore, we can apply Theorem~\ref{th:lagrange1} to show the existence of $\gamma(\nu) \geq 0$ such that
\begin{align}
\inf_{\mu\in \Omega:G(\mu) \leq 0 } \phi(\mu)
&= \inf_{\mu \in \Omega} \left[\phi(\mu) + \gamma(\nu) G(\mu)\right]\\
&= -\sup_{\mu \in \Omega} \left[I(\mu, \wy) - \gamma(\nu) \pr{\D{\wz\circ \mu}{q_0} - \nu}\right],
\end{align}
which results in the unconstrained reformulation of $A(\nu)$ as $\sup_{\mu \in \Omega} \left[I(\mu, \wy) - \gamma(\nu) \pr{\D{\wz\circ \mu}{q_0} - \nu}\right]$. Theorem~\ref{th:lagrange1} also implies that  $\mu^*_{\nu}$ is a solution to this new optimization problem, and since $ I(\mu, \wy) - \gamma(\nu) \pr{\D{\wz\circ \mu}{q_0} - \nu}$ is strictly concave \modfirst{\cite[Appendix I.B]{Abou-Faycal2001}}, the solution is unique.
\item \myr{
With the help of Lemma~\ref{lm:weak-diff} in Appendix~C to show the existence of weak derivatives (defined in \eqref{eq:weak-dev}), we use Theorem~\ref{th:kkt-general} with $f(\mu) = I(\mu, \wy) - \gamma(\nu) \pr{\D{\wz\circ \mu}{q_0} - \nu}$ to obtain that $\mu_1 = \mu^*_{\nu}$ if and only if for any $\mu \in \Omega$,
\begin{align}
0
&\geq f_{\mu_1}'(\mu)\displaybreak[0]\\
&=\E[\wy \times \mu ]{\log \frac{p_X(Y)}{(\wy \circ \mu_1)(Y)}} - I(\mu_1, \wy) \nonumber\\
&~~~~~~~~~~~~~~~~~~~~~~ ~~~~~~~~~~~~~~~~~~~~ -\gamma(\nu)\pr{\E[\wz\circ \mu]{\log \frac{(\wz \circ \mu_1)(Z)}{q_0(Z)}} - \D{\wz \circ \mu_1}{q_0}}\displaybreak[0]\\
&=\E[\wy \times \mu ]{\log \frac{p_X(Y)}{(\wy \circ \mu_1)(Y)}}  - \gamma(\nu) \pr{ \E[\wz\circ \mu]{\log \frac{(\wz \circ \mu_1)(Z)}{q_0(Z)}}- \nu} - f(\mu_1)\displaybreak[0]\\
&= \E[\mu]{w(X, \mu_1, \nu)} - f(\mu_1).
\end{align}

This implies that $\mu_1 = \mu_\nu^*$ if and only if for all $\mu\in \Omega$, we have $f(\mu_1) \geq \E[\mu]{w(X, \mu_1, \nu)}$.  Since $A(\nu) = \sup_{\mu\in\Omega}f(\mu) \geq f(\mu_1)$, if $\mu_1 = \mu_\nu^*$, then for all $\mu \in \Omega$, we have $A(\nu) \geq f(\mu_1)\geq \E[\mu]{w(X, \mu_1, \nu)}$. }
\item Assume \eqref{eq:point-w} is true, we take the expectation and obtain \eqref{eq:integral-w}. We now show the opposite direction and prove that if \eqref{eq:integral-w} holds, we have \eqref{eq:point-w} and \eqref{eq:point-supp-w}. Applying \eqref{eq:integral-w} with $\mu$  a deterministic probability measure with all mass point at $x$, we obtain
\begin{align}
A(\nu) \geq \E[\mu]{w(X, \mu_1, \nu)} = w(x, \mu_1, \nu).
\end{align}
Furthermore, for any $x\in \supp{\mu_1}$, we prove that $w(x, \mu_1, \nu) = A(\nu)$ by contradiction. If $ A(\nu) - w(x, \mu_1, \nu) \eqdef \delta > 0 $, by continuity of $A(\nu) - w(x, \mu_1, \nu) $ in $x$, there exists a neighborhood $\calN$ of $x$ such that for all $x'\in \calN$, we have $A(\nu) - w(x', \mu_1, \nu) \geq \delta/2$. Also, since $x\in\supp{\mu_1}$, we know that $\P[\mu_1]{X\in \calN} = \epsilon > 0$.  Therefore, we obtain
\begin{align}
A(\nu)= \E[\mu_1]{w(X, \mu_1, \nu)} 
&= \E[\mu_1]{w(X, \mu_1, \nu)\indic{X\in \calN}} + \E[\mu_1]{w(X, \mu_1, \nu)\indic{X\notin \calN}} \\
&\leq (1-\epsilon) A(\nu) + \epsilon\pr{A(\nu) - \frac{\delta}{2}} \\
&= A(\nu) - \frac{\delta\epsilon}{2} < A(\nu), 
\end{align}
which is a contradiction.
\item To prove that $\lim_{\nu \to 0^+}\gamma(\nu) = \infty$, we prove that $\gamma(\nu)\geq A'(\nu^+)$, and the result will follow from $\lim_{\nu \to 0^+} A'(\nu^+) = \infty$ as shown in Lemma~\ref{lm:a-prop}. Consider any $\nu_1, \nu_2>0$, and similar to the sensitivity analysis in \cite[Section 5.6]{boyd2004convex}, note that
\begin{align}
A(\nu_1) 
&= I(\mu^*_{\nu_1}, \wy) - \gamma(\nu_1) (\D{\wz \circ \mu^*_{\nu_1}}{q_0} - \nu_1)\\
&\stackrel{(a)}{\geq}  I(\mu^*_{\nu_2}, \wy) - \gamma(\nu_1) (\D{\wz \circ \mu^*_{\nu_2}}{q_0}-\nu_1)\\
&= I(\mu^*_{\nu_2}, \wy) - \gamma(\nu_1) (\D{\wz \circ \mu^*_{\nu_2}}{q_0}-\nu_2) + \gamma(\nu_1)(\nu_1 - \nu_2)\\
&\stackrel{(b)}{\geq} I(\mu^*_{\nu_2}, \wy) + \gamma(\nu_1)(\nu_1 - \nu_2),
\end{align}
where $(a)$ follows since $\mu^*_{\nu_1}$ is the maximizer of $\sup_{\mu} I(\mu, \wy) - \gamma(\nu_1)(\D{\wz\circ \mu}{q_0} - \nu_1)$, and $(b)$ follows since $\gamma(\nu_1) \geq 0$. Thus, for any $\nu>0$ and $\nu>h>0$, we have
\begin{align}
\frac{A(\nu) - A(\nu-h)}{h} \geq \gamma(\nu) \text{ and } \frac{A(\nu + h) - A(\nu)}{h} \leq \gamma(\nu).
\end{align}
Taking the limit $h\to 0^+$, we obtain $A'(\nu^+) \leq \gamma(\nu) \leq A'(\nu^-)$.

To prove that $\lim_{\nu \to 0^+}\gamma(\nu)\nu = 0$, note that for all $\nu > 0$,
\begin{align}
\gamma(\nu)\nu  
&\leq A'(\nu^-) \nu \\
&\stackrel{(a)}{\leq} \frac{A(\nu)}{\nu} \nu = A(\nu),
\end{align}
where $(a)$ follows from concavity of $A$. In the proof of Lemma~\ref{lm:a-prop}, we show that $\lim_{\nu \to 0^+}A(\nu) = 0$, which yields the result.
\end{enumerate}
\end{proof}

\begin{lemma}
\label{lm:weak-diff}
 $f(\mu) \eqdef I(\mu, \wy) - \gamma(\nu) \pr{\D{\wz\circ \mu}{q_0} - \nu}$ is weakly differentiable, and
 \begin{multline}
 f'_{\mu_1}(\mu) = \E[\wy \times \mu ]{\log \frac{p_X(Y)}{(\wy \circ \mu_1)(Y)}} -\\ I(\mu_1, \wy) - \gamma(\nu)\pr{\E[\wz\circ \mu]{\log \frac{(\wz \circ \mu_1)(Z)}{q_0(Z)}} - \D{\wz \circ \mu_1}{q_0}}.
 \end{multline}
\end{lemma}

\begin{proof}
In \cite[Equation (63)]{Abou-Faycal2001}, the weak derivative of $I(\mu, \wy)$ \modfirst{at $\mu_1$} is proved to be
\begin{align}
 \E[\wy \times \mu ]{\log \frac{p_X(Y)}{(\wy \circ \mu_1)(Y)}} - I(\mu_1, \wy).
\end{align}
Thus, we only check the weak differentiability of $G(\mu) \eqdef \D{\wz \circ \mu}{q_0} - \nu$. Let $\mu_1, \mu\in\Omega$, and define
\begin{align}
\mu_\theta &\eqdef (1-\theta) \mu_1 + \theta \mu\\
f_1(z)&\eqdef (\wz\circ \mu_1)(z)\\
f(z)&\eqdef (\wz\circ\mu)(z)\\
f_\theta(z) &\eqdef (\wz\circ \mu_\theta)(z).
\end{align}
Then, we have
\begin{align}
&G(\mu_\theta) - G(\mu_1) \\
&= \D{f_\theta}{q_0}-\D{f_1}{q_0}\\
&=\int_0^\infty f_\theta(z)\log \frac{ f_\theta(z)}{q_0(z)}dz-\int_0^\infty f_1(z)\log \frac{ f_1(z)}{q_0(z)}dz\\
&\stackrel{(a)}{=}\int_0^\infty f_\theta(z)\log \frac{ f_\theta(z)}{q_0(z)}dz -\int_0^\infty f_\theta(z)\log \frac{ f_1(z)}{q_0(z)}dz+\int_0^\infty f_\theta(z)\log \frac{ f_1(z)}{q_0(z)}dz -\int_0^\infty f_1(z)\log \frac{ f_1(z)}{q_0(z)}dz\\
&=\int_0^\infty f_\theta(z)\log \frac{ f_\theta(z)}{f_1(z)}dz+\int_0^\infty f_\theta(z)\log \frac{ f_1(z)}{q_0(z)}dz -\int_0^\infty f_1(z)\log \frac{ f_1(z)}{q_0(z)}dz\\
&\stackrel{(b)}{=} \int_0^\infty f_\theta(z)\log \frac{ f_\theta(z)}{f_1(z)}dz+\theta\pr{\int_0^\infty f(z)\log \frac{ f_1(z)}{q_0(z)}dz -\int_0^\infty f_1(z)\log \frac{ f_1(z)}{q_0(z)}dz},\label{eq:weak-derv-div}
\end{align}
where $(a)$ holds since by Lemma~\ref{lm:cross-entropy}, $\int_0^\infty f_\theta(z)\log \frac{ f_1(z)}{q_0(z)}dz < \infty$, and  $(b)$ follows from $f_\theta = (1-\theta)f_1 + \theta f$. The second term in \eqref{eq:weak-derv-div} is differentiable with respect to $\theta$, and the derivative is
\begin{align}
\int_0^\infty f(z)\log \frac{ f_1(z)}{q_0(z)}dz -\int_0^\infty f_1(z)\log \frac{ f_1(z)}{q_0(z)}dz.
\end{align}
To take derivative from the first term in \eqref{eq:weak-derv-div}, we use Theorem~\ref{th:leibniz}. Note that by Lemma~\ref{lm:cross-entropy}, $\int_0^\infty f_\theta(z)\left|\log \frac{ f_\theta(z)}{f_1(z)}\right|dz < \infty$, and also, for all $z$ and $\theta$, 
\begin{align}
\frac{\partial }{\partial \theta} \pr{f_\theta(z)\log \frac{ f_\theta(z)}{f_1(z)}} 
&= -(f_1(z) - f(z))\pr{1 + \log\frac{f_\theta(z)}{\modfirst{f_1}(z)}}.
\end{align}
Additionally, for all $\theta\in[0, 1]$, if we apply \eqref{eq:bound-wz}, we obtain 
\begin{align}
|f_\theta(z)\log \frac{ f_\theta(z)}{f_1(z)}| 
&{\leq} |f_1(z) + f(z)|\pr{|\log f_1(z)| + \log(1+\E[\mu_{\theta}]{X}) + z} \\
&\leq |f_1(z) + f(z)|\pr{|\log f_1(z)| + \log(1+\E[\mu_1]{X} + \E[\mu]{X}) + z},
\end{align}
which is a integrable function  with respect to Lebesgue measure on $\calZ$ by Lemma~\ref{lm:cross-entropy} and does not depend on $\theta$. Hence, all condition in Theorem~\ref{th:leibniz} hold, and we have
\begin{align}
\frac{\partial }{\partial \theta} \pr{ \int_0^\infty f_\theta(z)\log \frac{ f_\theta(z)}{f_1(z)}dz}
&= \int_0^\infty  \frac{\partial }{\partial \theta}\pr{f_\theta(z)\log \frac{ f_\theta(z)}{f_1(z)}}dz\\
&=  \int_0^\infty  -(f_1(z) - f(z))\pr{1 + \log\frac{f_\theta(z)}{\modfirst{f_1}(z)}}dz\\
&= \int_0^\infty  -(f_1(z) - f(z))\log\frac{f_\theta(z)}{\modfirst{f_1}(z)}dz,
\end{align}
which vanishes at $\theta = 0$. Therefore, $G$ is weakly differentiable at $\mu_1$ and
\begin{align}
G'_{\mu_1}(\mu)&= \int_0^\infty f(z)\log \frac{ f_1(z)}{q_0(z)}dz -\int_0^\infty f_1(z)\log \frac{ f_1(z)}{q_0(z)}dz.\label{eq:diff_g}
\end{align}
Since the mutual information and the divergence are weakly differentiable, so is $ I(\mu, \wy) - \gamma(\nu) \pr{\D{\wz\circ \mu}{q_0} - \nu}$.
\end{proof}


\bibliographystyle{IEEEtran}
\bibliography{biblio}
\end{document}